\documentclass[%
 reprint,
superscriptaddress,
nofootinbib,
 amsmath,amssymb,
 aps,
pra,
floatfix,
]{revtex4-2}
\pdfoutput=1
\usepackage[utf8]{inputenc}
\usepackage[english]{babel}
\usepackage[T1]{fontenc}
\usepackage{amsmath}
\usepackage[hidelinks]{hyperref}

\usepackage{tikz}
\usepackage{lipsum}
\usepackage[letterpaper,margin=1in]{geometry}
\usepackage{amsmath,amsfonts,amssymb}
\usepackage{xcolor}
\usepackage{caption}
\usepackage{subcaption}
\usepackage{graphicx}
\usepackage{float}
\usepackage{amsthm}
\usepackage{appendix}
\usepackage{blkarray}
\usepackage{multirow}
\usepackage{comment}
\usepackage[ruled,vlined]{algorithm2e}
\usepackage{quantikz}

\theoremstyle{definition}
\newtheorem{definition}{Definition}[section]
\newtheorem{theorem}{Theorem}[section]

\newtheorem{proposition}{Proposition}[section]

\newtheorem{remark}{Remark}[section]

\numberwithin{equation}{section}
\newenvironment{hproof}{%
  \proof}{\endproof}
\newcommand{\be}{\begin{equation}}
\newcommand{\ee}{\end{equation}}
\newcommand{\D}{\ensuremath{q}} % macro for default name for bond dimension  
\newcommand{\CC}[1]{\textcolor{black}{(CC) #1}}

\DeclareMathOperator{\Tr}{Tr}

\begin{document}

\title{Quantum Lego Expansion Pack: Enumerators from Tensor Networks}

\author{ChunJun Cao}
\affiliation{Joint Center for Quantum Information and Computer Science, NIST/University of Maryland, College Park, MD  20742, USA}
\affiliation{Institute for Quantum Information and Matter, Caltech, Pasadena, CA 91125, USA}
\affiliation{Department of Physics, Virginia Tech, Blacksburg, VA 24060, USA}

\author{Michael J. Gullans}
\affiliation{Joint Center for Quantum Information and Computer Science, NIST/University of Maryland, College Park, MD 20742, USA}
\author{Brad Lackey}
\affiliation{Microsoft Quantum, Redmond, WA 98052, USA}
\author{Zitao Wang}
\affiliation{Meta Platforms Inc., Menlo Park, CA 94025, USA}
%\affiliation{Menlo Park, CA 94025, USA}

\begin{abstract}
We provide the first tensor network method for computing quantum weight enumerator polynomials in the most general form. If a quantum code has a known tensor network construction of its encoding map, our method is far more efficient, and in some cases exponentially faster than the existing approach. As a corollary, it produces decoders and an algorithm that computes the code distance. For non-(Pauli)-stabilizer codes, this constitutes the current best algorithm for computing the code distance. For degenerate stabilizer codes, it can be substantially faster compared to the current methods. We also introduce novel weight enumerators and  their applications. {\color{black} 
 In particular, we show that these enumerators can be used to compute logical error rates exactly and thus construct (optimal) decoders for any i.i.d. single qubit or qudit error channels. The enumerators also provide a more efficient method for computing non-stabilizerness in quantum many-body states. As the power for these speedups rely on a Quantum Lego decomposition of quantum codes, we further provide systematic methods for decomposing quantum codes and graph states into a modular construction for which our technique applies.} As a proof of principle, we perform exact analyses of the deformed surface codes, the holographic pentagon code, and the 2d Bacon-Shor code under (biased) Pauli noise and limited instances of coherent error at sizes that are inaccessible by brute force.

\end{abstract}

\maketitle
\tableofcontents

\section{Introduction}

Topological and geometrical insights  have led to a number of recent breakthroughs in quantum error correction, e.g. \cite{Kitaev03,fiberbundlecode,PK}. On the other hand, quantum weight enumerator polynomials \cite{ShorLaflamme} provide a complementary, algebraic perspective on quantum error correcting codes (QECCs). Quantum weight enumerators contain crucial information of the code property. A number of variants and generalizations have also been applied to derive linear programming bounds \cite{rains1,rains2,rains3}, to understand error detection under symmetric \cite{Ashikhmin:1999ef} and asymmetric\cite{hu2020weight} Pauli errors, and for generating magic state distillation protocols \cite{Rall}. {\color{black}In quantum many-body physics, enumerators, also known as sector lengths, have been used to study entanglement structure \cite{SLD} of quantum states. The weight distributions of operators have also played an important role in quantum chaos \cite{schuster2022operator}.} However, wider applications of the quantum weight enumerators have been relatively limited beyond codes or states of small sizes compared to the other approaches partly to due their prohibitive computational costs.

Building upon the previous framework of quantum lego (QL) \cite{CL2021} and the recently developed tensor weight enumerator formalism \cite{CL2022}, we revisit the weight enumerator perspective of quantum error correction and provide a more efficient method to compute them. {\color{black}We present new results in both formalism and in algorithm that enable a number of novel applications for quantum error correction, measurement-based quantum computation, and quantum many-body physics.} On the formalism level, we review abstract weight enumerators and their corresponding MacWilliams identities \cite{CL2022}. We then introduce mixed enumerators, higher genus enumerators, coset enumerators and generalized enumerators, which are useful for the study of subsystem codes, decoders, and logical error probability under general independent and identically distributed (i.i.d.) single qubit error channels.

On the algorithmic level, we provide a tensor network method for computing these quantum weight enumerators in their most abstract forms. Because one can read off the code distance from weight enumerators, the problem of finding them is at least as hard as the minimal distance problem for classical linear codes, which is NP-hard \cite{mindist,distance1,distance2,distance3}. We show that quantum weight enumerators also produce optimal decoders, hence the general problem is at least $\#$P-complete, which is the hardness of evaluating weight enumerators for classical linear codes \cite{Vyalyi}. However, more efficient algorithms are possible if additional structures are known. To the best of our knowledge, our work constitutes the best current algorithm for generating quantum weight enumerator polynomials as long as a good QL construction for the quantum code is known. Compared to the brute force method, our algorithm provides up to a substantial speed up. 

The enumerators immediately induce a protocol to compute quantum code distances. To the best of our knowledge, it provides the first such protocol for general quantum codes beyond (Pauli) stabilizer codes, which can be exponentially faster than brute force search in many instances. For non-degenerate Pauli stabilizer codes, the complexity scaling is roughly comparable with existing algorithms for classical linear codes under reasonable assumptions, which implies that it scales exponentially with the code distance. For degenerate codes, our method can be exponentially faster  in certain instances compared to known methods based on classical linear codes. 

{\color{black}Finally, we introduce novel applications and new abstract enumerators that have not been discussed in literature. We generalize \cite{Ashikhmin:1999ef} and connect enumerators to logical error probabilities when the code is subjected to any i.i.d. single qudit error channel. We then provide the optimal decoder for any code that admits a known QL construction and propose a more accurate method to compute effective distances and error thresholds. Our arguments hints at a general connection between the hardness of distance calculation, optimal decoding, and the amount of entanglement present in the system.  Because the speedup relies on a tensor network construction of the quantum code or quantum state in question, we also provide a systematic method for building all stabilizer codes and graph states using Quantum Lego. This also includes the cases where the stabilizer group is non-Abelian, such as the Quantum Double models. These breakthroughs lift a longstanding computational barrier for the exact analyses of quantum codes and resource states in measurement-based quantum computation. Additionally, we show that the higher genus weight enumerator of pure states computes the non-stabilizerness of a state, thus providing an another efficient method for the challenging task of computing quantum many-body magic. }

As a proof of principle, we derive weight enumerators, compute (biased) distances, and obtain exact analytical expressions for logical error probabilities under depolarizing and coherent noise for a few well-known stabilizer and subsystem codes that are of order a hundred qubits or so. The novel contributions in this paper are summarized in Fig.~\ref{fig:intro}.
Overall, the tensor-network-based approach to quantum weight enumerators provides both a unified framework and the practical means for studying code properties, decoding, entanglement, and magic in quantum codes and quantum manybody systems at large.

\begin{figure}
    \centering
    \includegraphics[width=\linewidth]{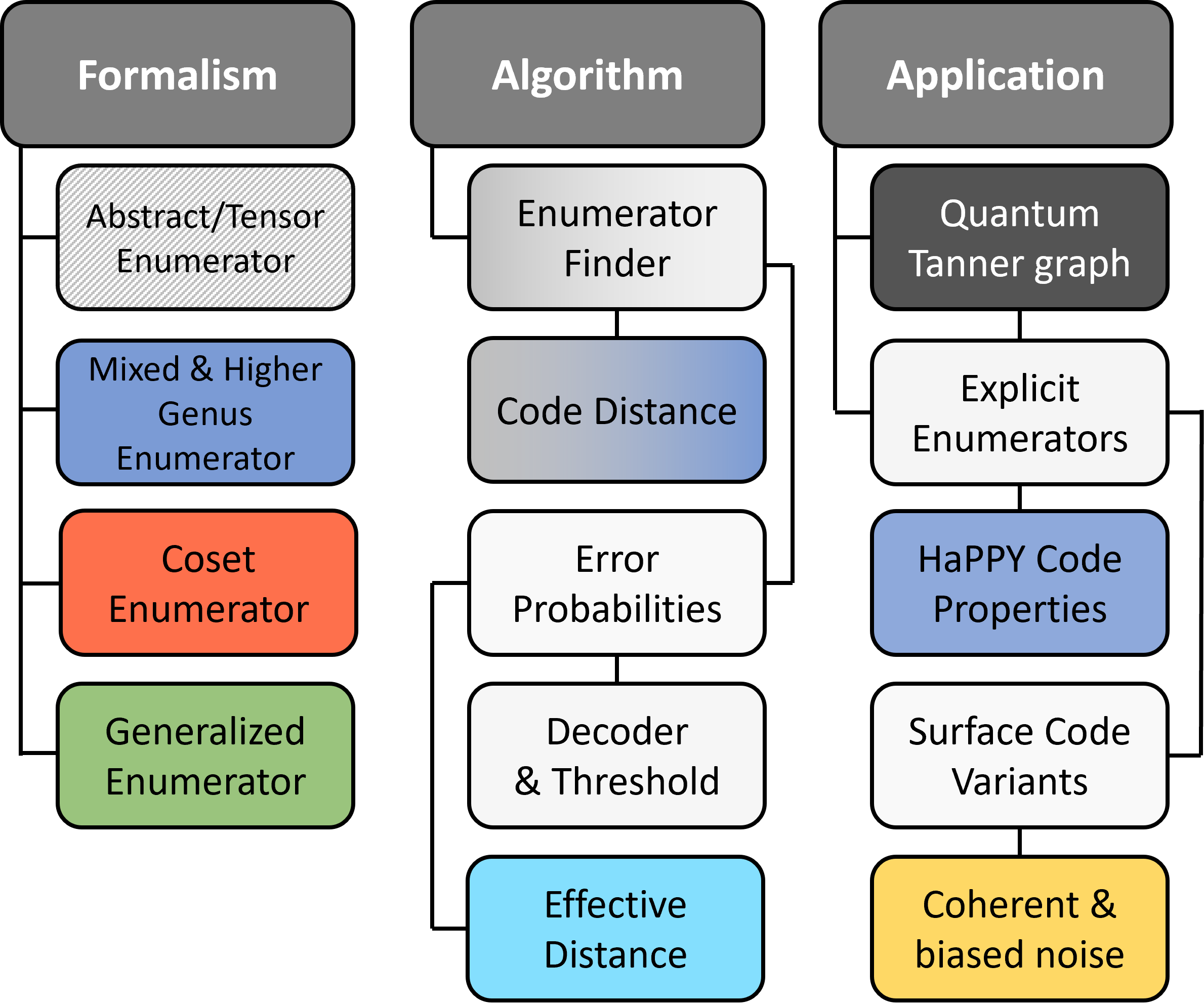}
    \caption{Summary of contributions. Topic dependencies are red-green-blue color coded.  If all  three colored topics in the formalism section are used, then the color is white.  Cyan indicates green and blue topics.  Yellow indicates red and green topics. Black indicates that it does not use any of the new formalism, but is a new tensor network construction. Half shaded grey/blue indicates it uses grey and blue topics.  Half shaded grey/white indicates that it uses all 4 formalism topics.  }
    \label{fig:intro}
\end{figure}

In Sec~\ref{sec:formalism}, we review the basics of weight enumerator polynomials in the most abstract form and introduce their generalizations. In Sec.~\ref{sec:application}, we discuss their existing applications for computing code distance and extend their applications for error detection under general error channels. We present novel constructions such as mixed enumerators, higher genus enumerators and coset enumerators. {\color{black} We introduce new applications of these enumerators in building optimal decoders, in analyzing subsystem codes, and in computing non-stabilizerness in quantum states. We show that the error detection threshold for a large class of codes have a universal value of $1/6$ and} suggest improvements for threshold computations based on existing sampling-based methods when used in conjunction with enumerators. Then we discuss the computational cost of this method and provide some entanglement-based intuition in Sec.~\ref{sec:complexity}. As a proof of principle, and to provide novel analysis of existing codes, we study some common examples and explain their significance in Sec.~\ref{sec:examples}.
In Sec~\ref{subsec:surfacecode} we construct various weight enumerators of the (rotated) surface code and its deformations. We compare their performances under biased noise and coherent error channels. In Sec~\ref{subsec:colorcode} we provide a new tensor network construction of the 2d color code using Steane codes as basic building blocks and compute its enumerators. In Sec~\ref{subsec:holographic} we study different bulk qubits with mixed enumerators in the holographic HaPPY code. We obtain their (biased) distances and performance under (biased) Pauli noise. In Sec~\ref{subsec:bsc} we apply the mixed enumerator technology to the Bacon-Shor code and showcase its computation for subsystem codes.
Finally, we make some summarizing comments in Sec.~\ref{subsec:discussion} and provide insights on the connection with stat mech model and graph states.

We prove the relevant theorems,  discuss technical implementations and clarify practical simplifications in the Appendices. Although not stated explicitly, the distance finding protocol introduced in \cite{LTNC} effectively computes the Shor-Laflamme enumerators for a subset of stabilizer codes known as local tensor network codes. Their approach also shares a number of similarities with our own, which we explain in App.~\ref{subapp:LTNC}. For such stabilizer codes, our protocol generally offers a quadratic speed-up in the form of reduced bond dimensions. In the regime where the stabilizer code has high rate and code words are highly entangled, our method can lead to an exponential advantage using the quantum MacWilliams identities.

\section{General Formalism}\label{sec:formalism}
Throughout the article, we represent multi-indexed objects like vectors and tensors in bold face letters $\mathbf{A},\mathbf{B}$ to avoid clutter of indices. Scalar objects are written in regular fonts like $A,B$.

\subsection{Abstract scalar weight enumerator}
Abstract scalar weight enumerators introduced in \cite{CL2022} include common enumerators discussed in literature \cite{ShorLaflamme,rains2,hu2020weight}. 
Let $\mathcal{E}$ be an error basis on Hilbert space $\mathfrak{H}$ with local dimension $q$. A \emph{weight function} is any function $\mathrm{wt}:\mathcal{E} \to \mathbb{Z}_{\geq 0}^k$. We extend this (without introducing new notation) to $\mathrm{wt}:\mathcal{E}^n \to \mathbb{Z}_{\geq 0}^k$ by
\begin{equation}\mathrm{wt}(E_1 \otimes \cdots \otimes E_n) = \mathrm{wt}(E_1) + \cdots + \mathrm{wt}(E_n).\end{equation}
For a $k$-tuple of indeterminates $\mathbf{u} = (u_1, \dots, u_k)$ we write
\begin{equation}\mathbf{u}^{\mathrm{wt}(E)} = u_1^{\mathrm{wt}(E)_1} \cdots u_k^{\mathrm{wt}(E)_k}.\end{equation}

We can then define abstract enumerators of Hermitian operators $M_1,M_2$ for a weight function $\mathrm{wt}$ as
\begin{align}
    A(\mathbf{u}; M_1, M_2) &= \sum_{E \in \mathcal{E}^n} \Tr(E M_1) \Tr(E^\dagger M_2) \mathbf{u}^{\mathrm{wt}(E)}\\\nonumber
    B(\mathbf{u}; M_1, M_2) &= \sum_{E \in \mathcal{E}^n} \Tr(E M_1 E^\dagger M_2) \mathbf{u}^{\mathrm{wt}(E)}.
\end{align}

These polynomials satisfy a quantum MacWilliams identity. Let us restrict to the case where our error basis satisfies $EFE^\dagger F^\dagger = \omega(E,F) I$ for a phase $\omega(E,F)$. This includes the Pauli basis (of local dimension $q$) as well as general Heisenberg representations. Consider the (polynomial-valued) function $f(E) = \mathbf{u}^{\mathrm{wt}(E)}$ for a weight function $\mathrm{wt}:\mathcal{E} \to \mathbb{Z}_{\geq 0}^k$. Then the discrete Wigner transform of this function is
\begin{equation}\hat{f}(D) = \frac{1}{q}\sum_{E} \omega(E,D) f(E) = \frac{1}{q}\sum_{E} \omega(E,D) \mathbf{u}^{\mathrm{wt}(E)}.\end{equation}

\begin{theorem}
    Suppose there exists an algebraic mapping $\Phi(\mathbf{u}) = (\Phi_1(\mathbf{u}), \dots, \Phi_k(\mathbf{u}))$ such that
    \begin{equation}\Phi(\mathbf{u})^{\mathrm{wt}(D)} = \hat{f}(D) = \frac{1}{q}\sum_{E} \omega(E,D) \mathbf{u}^{\mathrm{wt}(E)}.\end{equation}
    Then for any $M_1,M_2$ we have
   \begin{equation}\label{eqn:abstractMacWilliams}
       B(\mathbf{u};M_1,M_2) = A(\Phi(\mathbf{u});M_1, M_2).
   \end{equation}
\end{theorem}

\begin{proof}
    See \cite{CL2022}.
\end{proof}

The map $\Phi$ is a generalization of the discrete Wigner transform. For the remainder of the work, we take $\mathcal{E}$ to be the Pauli group. By considering different forms of the variable $\mathbf{u}$, abstract weight function $\rm wt$, and transformation $\Phi$, one can recover existing scalar enumerator polynomials and their MacWilliams identities. For completeness, we review a few common enumerators in Appendix~\ref{app:scalarenum} that are used in this work. 

\subsection{Generalized Abstract Weight Enumerators}\label{subsec:genabs}

Slightly extending the form in the previous section, we define a novel generalized weight enumerator. 

\begin{align}%\label{eqn:generalized_enum}
    \bar{A}(\mathbf{u};M_1,M_2) &= \sum_{E,F\in\mathcal{E}^{n}}\Tr[EM_1]\Tr[F^{\dagger}M_2] \mathbf{u}^{wt(E,F)}\\\nonumber
    \bar{B}(\mathbf{u};M_1,M_2) &= \sum_{E,F\in\mathcal{E}^{n}}\Tr[EM_1 F^{\dagger}M_2] \mathbf{u}^{wt(E,F)},
\end{align}
where $\mathrm{wt}(E,F)$ is an abstract function of the operators $E,F$ and $\mathbf{u}$ is a set of variables. It has no obvious classical analogues as far as we know. This type of enumerators are useful in analyzing qudit-wise general error channels. We further elaborate this connection in Sec~\ref{subsec:errdet} for coherent noise and other single qubit errors such as amplitude damping channels. We are not able to identify MacWilliams identities for these types of enumerator polynomials in general.

\subsection{Tensor Weight Enumerators}\label{subsec:twep}

One can generalize the above scalar enumerator formalism to vectors and tensors. The reasons for this extension is two-fold: 1) the novel vector or tensor enumerators can probe code properties unavailable to their scalar counterparts and 2) the cost for computing scalar enumerators is generally expensive and scales exponentially with ${n-k}$. However, by contracting suitable tensor weight enumerators, one can break down the computation of scalar enumerators into manageable pieces and render the process far more efficient. In this section, we briefly review the basic definitions of these vectorial and tensorial enumerators and introduce their graphical representations.

From \cite{CL2022}, we define tensor enumerators

\begin{widetext}
\begin{align}\label{eqn:tensor_enum}
    \mathbf{A}^{(J)}(\mathbf{u}; M_1, M_2) &= \sum_{E,\bar{E}\in\mathcal{E}^m} \sum_{F \in \mathcal{E}^{n-m}} \Tr((E \otimes_J F)M_1) \Tr((\bar{E}^\dagger \otimes_J F^\dagger) M_2) \mathbf{u}^{\mathrm{wt}(F)} e_{E,\bar{E}},\\
    \mathbf{B}^{(J)}(\mathbf{u}; M_1, M_2) &= \sum_{E,\bar{E}\in\mathcal{E}^m}  \sum_{F \in \mathcal{E}^{n-m}}  \Tr((E \otimes_J F)M_1 (\bar{E}^\dagger \otimes_J F^\dagger) M_2) \mathbf{u}^{\mathrm{wt}(F)} e_{E,\bar{E}}
\end{align}
\end{widetext}
where $\{e_{E,\bar{E}}\}$ are orthonormal basis vectors of a $q^4$-dimensional vector space. Here, $\mathrm{wt}(F)$ is an abstract weight function we discussed in the previous section and $\mathbf{u}$ can be an $n$-tuple of variables, and 
$J \subseteq \{1,\dots,n\}$ is a set of $m$ qudits/locations. We write $\otimes_J$ denotes the tensor product of length $m$ Pauli string $E$ interlaced with Pauli string $F$ of length $n-m$ at the positions marked in the set $J$. Later we will also use $\mathcal{E}^{n-m}[d]$ is the set of Pauli operators $F$ on $n-m$ sites that have weight $d$. 

\begin{figure}[t]
    \centering
    \includegraphics[width=0.7\linewidth]{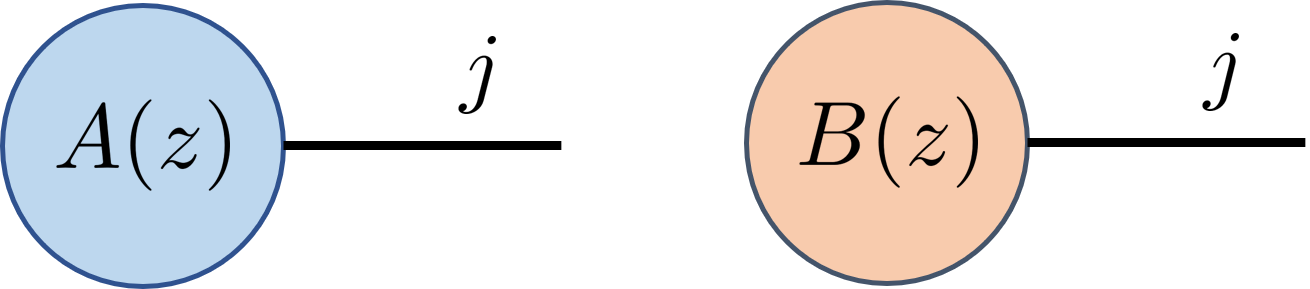}
    \caption{Vector Enumerators}
    \label{fig:vwep}
\end{figure}

To give a more concrete illustration of these objects consider the case of a rank-1 tensors ($m=1$), which we refer to as \emph{vector} enumerators. For simplicity consider the usual (quantum) Hamming weight where $\mathbf{u}=z$ and $\mathrm{wt}(E)$ returns the number of nonidentity tensor factors in the Pauli operator $E$. For $J=\{j\}$ the vector enumerators along leg $j$ read 
\begin{align}
    &\mathbf{A}^{(j)}(z; M_1, M_2)\\ \nonumber
    &= \sum_{E,\bar{E}\in\mathcal{E}} \sum_{d=0}^n A^{(j)}_d(E,\bar{E};M_1,M_2) z^d e_{E,\bar{E}},\\
    &\mathbf{B}^{(j)}(z; M_1, M_2) \\\nonumber
    &= \sum_{E,\bar{E}\in\mathcal{E}} \sum_{d=0}^n B^{(j)}_d(E,\bar{E};M_1,M_2) z^d e_{E,\bar{E}}.
\end{align}
with coefficients (weights) here are defined as 
\begin{align}
    &A^{(j)}_{d}(E,\bar{E};M_1,M_2)\\\nonumber
    &= \sum_{F \in \mathcal{E}^{n-1}[d]} \Tr((E \otimes_j F)M_1) \Tr((\bar{E}^\dagger \otimes_j F^\dagger) M_2),\\
    &B^{(j)}_{d}(E,\bar{E};M_1,M_2)\\\nonumber
    &= \sum_{F \in \mathcal{E}^{n-1}[d]} \Tr((E \otimes_j F)M_1 (\bar{E}^\dagger \otimes_j F^\dagger) M_2).
    %\label{eqn:doubleTE}
\end{align}
The $\mathcal{E}^{n-1}[d]$ here is the set of operators that have weight $d$ on the $n-1$ qubits except the $j$th one, and  $E\otimes_j F$ is a Pauli string that has $E$ inserted on the $j$-th position of the Pauli string:
$$E\otimes_j F= F_1\otimes F_2\otimes \dots F_{j-1} \otimes E_j \otimes F_{j+1}\otimes\dots \otimes F_{n}.$$

Formally, it is also convenient to express these coefficients in coordinates, once we have chosen a standard basis $\{\mathbf{\hat{e}}_j\}$. For example, one can denote
\begin{align}
    &A^{(j)}_d(E, \bar{E};M_1,M_2) \rightarrow A^{j}_d,\\ 
    &B^{(j)}_d(E,\bar{E};M_1,M_2) \rightarrow B^{j}_d
\end{align}
by identifying $j=0,\dots, q^4$ where $E,\bar{E}$ each has $q^2$ distinct values. For simplicity, we abuse notation and use $j$ as an open index that labels the dangling leg that comes from the $j$-th qudit. The corresponding vector enumerator polynomials are $ A^{j}(z;M_1,M_2),  B^{j}(z;M_1,M_2)$, which we represent graphically as a rank-1 tensors in Fig.~\ref{fig:vwep}.

In the same vein, the coefficients for a tensor enumerator of rank $m$ may be written as 
\begin{align}
    &\sum_d A^{(J)}_d(E,E;M_1,M_2)z^d\\\nonumber
    &\quad \rightarrow \sum_d A_d^{j_1\dots j_m} z^d=A^{j_1,j_2,\dots, j_m} (z), \\
    &\sum_d B^{(J)}_d(E,E;M_1,M_2)z^d\\\nonumber
    &\quad \rightarrow \sum_d B_d^{j_1\dots j_m}z^d=B^{j_1,j_2,\dots, j_m} (z),
\end{align}
where each tensor coefficient $A^{j_1,j_2,\dots, j_m} (z)$, $B^{j_1,j_2,\dots, j_m} (z)$ is a scalar enumerator. A graphical representation of $A^{j_1,j_2,\dots, j_m} (z)$ is given below in Figure \ref{fig:twep} (top left). 

In practice, it is often sufficient to consider reduced versions of these enumerators that only keep the diagonal terms with $E=\bar{E}$, which we represent using the same graphical form, but now with reduced bond dimension $j_\ell=1,\dots, q^2$. Such enumerators are known as the \emph{reduced enumerators} and they are sufficient  for studying Pauli errors in stabilizer codes. See \cite{CL2022} and App.~\ref{subapp:reducedenum}. In this work, we use the color blue to denote $A$-type enumerators and orange to denote $B$-type enumerators. We often drop the variable $z$ or $\mathbf{u}$ to avoid clutter, but it should be understood that the tensor components of these objects are polynomials.

One can also easily define other tensor enumerators such as the double and complete enumerators by choosing different expressions for the abstract forms $\mathbf{u}$ and weight functions $\mathrm{wt}(E)$. An extension to the generalized abstract tensor  enumerator is also possible. Details are found in App.~\ref{app:tensorenum}.

\subsection{Tracing tensor enumerators}\label{subsec:tweptn}
Let us define a trace operation $\wedge_{j,k}$ over the tensor enumerators which connects any two legs $j,k$ in the tensor network. Graphically, it is represented by a connected edge in the dual enumerator tensor network. Acting on the basis element $e_{E,\bar{E}}$ we define 
\begin{equation}
    \wedge_{j,k} e_{E,\bar{E}} = e_{E\setminus\{E_j,E_k\},\bar{E}\setminus\{\bar{E}_j,\bar{E}_k\}}
\end{equation}
when $E_j = E_k^*$ and $\bar{E}_j = \bar{E}_k^*$ and zero otherwise.

Each contraction can be understood as tracing together two tensors. However we can also view the two tensors as a single tensor enumerator (using the tensor product) then performing a self-trace, which is necessary and sufficient to build up any tensor network. Informally, the trace of the tensor enumerator is the tensor enumerator of the traced network, which is formally stated as the following.

\begin{theorem}\label{thm:tensortrace1}
    Suppose $j,k \in J \subseteq \{1,\dots,m\}$. Then
    \begin{align*}
        &\wedge_{jk}\mathbf{A}^{(J)}(\mathbf{u};M_1,M_2)\\
        &\quad = \mathbf{A}^{(J\setminus\{j,k\})}(\mathbf{u};\wedge_{j,k}M_1,\wedge_{j,k}M_2),
    \end{align*}
    and similarly for $\mathbf{B}^{(J)}$.
\end{theorem}
\begin{proof}
    See Theorem 7.1 of \cite{CL2022}.
\end{proof}

Theorem \ref{thm:tensortrace1} allows us to compute the weight enumerator of a contracted tensor network by contracting the tensor enumerators of each QL block. For example, to construct a scalar enumerator given the QL representation of an encoding map in Fig.~\ref{fig:wep_422}, we first lay down its ``shadow'' that is the tensor enumerator for each $[[4,2,2]]$ atomic code. Then we trace together these blocks following the same network connectivity.

\begin{figure}[b]
    \centering
    \includegraphics[width=\linewidth]{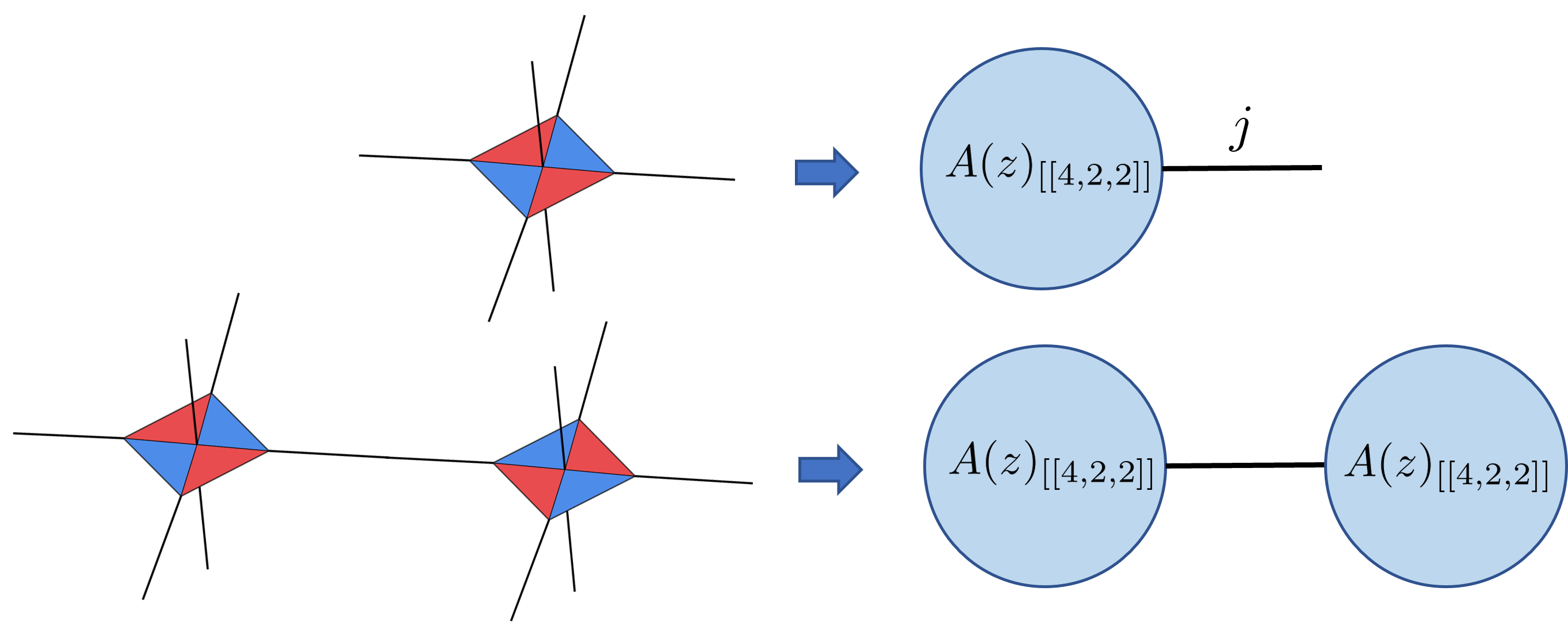}
    \caption{WEP from tracing $[[4,2,2]]$ codes} 
    \label{fig:wep_422}
\end{figure}

The component form of contracting tensor enumerators can be expressed as the conventional sum over indices for a tensor trace. For reduced enumerators at $q=2$ this reads,
\begin{align}
    \label{eqn:weptrace}
    &A^{j_{l+1},\dots,j_m,r_{l+1},\dots,r_k}(\mathbf{u})\\
    =&\sum_{j_1,j_2,\dots, j_l}A^{j_1,j_2,\dots, j_l,\dots j_m}(\mathbf{u})A^{j_1,j_2,\dots j_l, r_{l+1}\dots r_k}(\mathbf{u})\nonumber
\end{align}
and similarly for $\mathbf{B}(\mathbf{u};M_1,M_2)$, where the only difference from a traditional tensor network is the variables $\mathbf{u}$ associated with the polynomial. One can connect these tensors sequentially; at each step an atomic code is glued to the (generically) bigger connected component, Fig.~\ref{fig:twep}. For the full tensor enumerator, or when $q>2$, we need to take more care in raising and lowering the indices to recast them into the proper covariant and contravariant forms before summing over repeated indices.   

While it is natural to use symbolic packages to implement this formalism, we will also elaborate in Appendix~\ref{app:tensoronly} how to implement these objects as the usual multi-linear function without symbolic packages using conventional tensor network methods. 

\begin{figure}[t]
    \centering
    \includegraphics[width=0.7\linewidth]{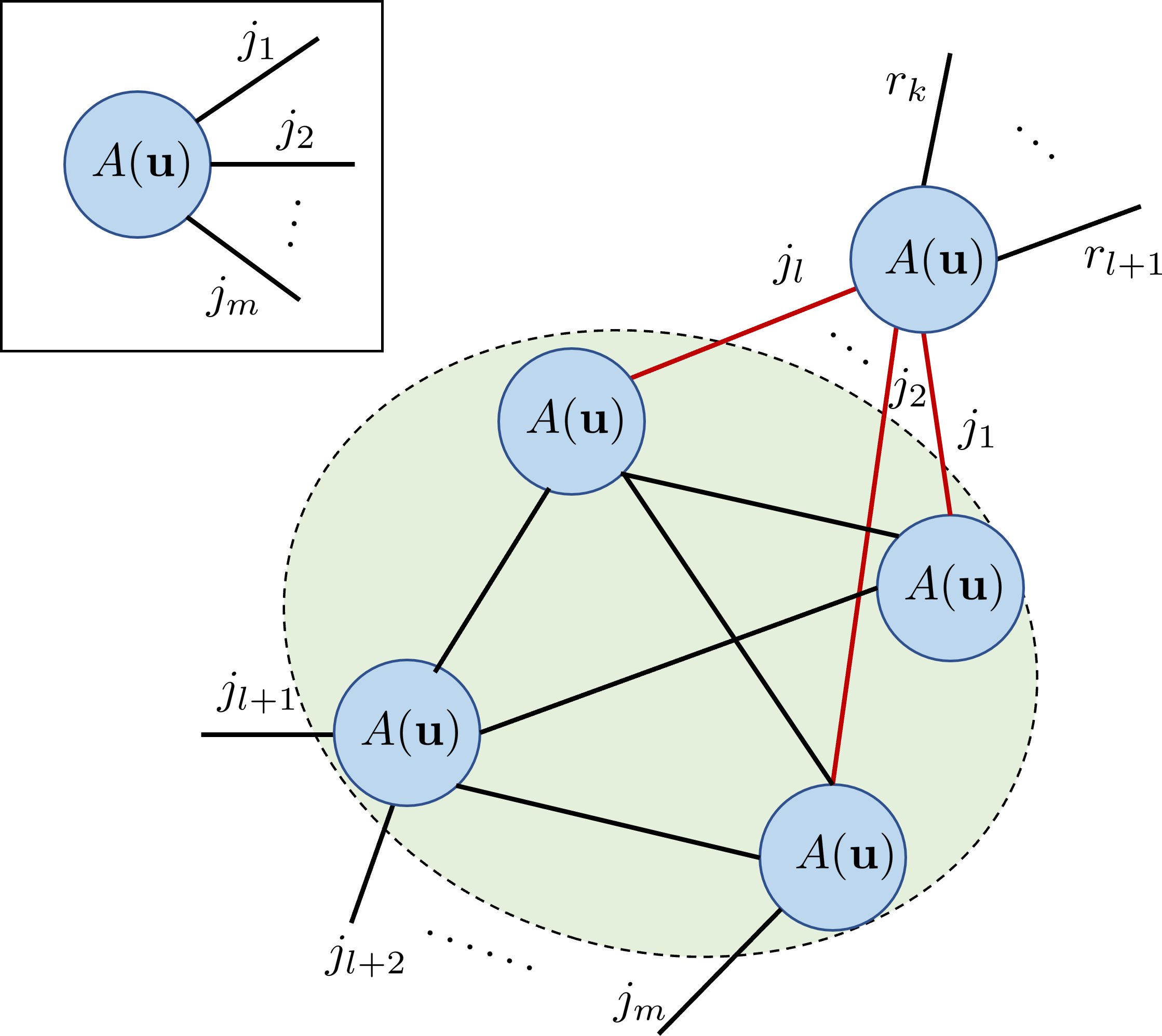}
    \caption{Graphical representation of a type-$A$ tensor enumerator (box). Tracing the type $A$ tensors as in eqn (\ref{eqn:weptrace}). Green region can be seen as $A^{j_1,j_2,\dots,j_l,\dots,j_m}(\mathbf{u})$. Traced legs are red. }
    \label{fig:twep}
\end{figure}

\section{Applications of weight enumerators}\label{sec:application}
\subsection{Code Distance from Enumerators}
The genesis of quantum weight enumerators came from the case $M_1=M_2=\Pi$, the projection onto a stabilizer code, and $\mathbf{u}^{wt(E)}=z^{wt(E)}$. After an appropriate normalization, the enumerators $A^{norm}(z)=A(z)/K^2,\ B^{norm}(z)=B(z)/K$ encode the weight distributions of stabilizers (logical identities) and normalizers (all logical operators) of the code respectively \cite{ShorLaflamme}. The normalized polynomials $A^{norm}(z),\ B^{norm}(z)$ have $B_0=A_0=1$. It follows that $B^{norm}(z)-A^{norm}(z)$ yields the weight distributions of non-trivial logical Pauli operators. Therefore, the smallest $d$ for which $B_d\ne A_d$ is thus the (adversarial) code distance. This observation also generalizes to any quantum code \cite{Ashikhmin:1999ef}. Formally we capture this in the following result for later reference.

\begin{theorem}\label{thm:3}
    Let $\mathcal{C}$ be a quantum code,  $\Pi_{\mathcal{C}}$ be the projection onto its code subspace and 
    \begin{align}
        A(z;\Pi_{\mathcal{C}},\Pi_{\mathcal{C}}) &= \sum_d A_dz^d\\
        B(z;\Pi_{\mathcal{C}},\Pi_{\mathcal{C}}) &= \sum_d B_d z^d
    \end{align}
    be its weight enumerator polynomials properly normalized. Then
    \begin{itemize}
        \item $A_0=B_0=1$,
        \item $B_d\geq A_d\geq 0$ for all $d$, and
        \item the code distance is $t+1$ where $t$ is the largest integer for which $B_i=A_i$ for all $0\leq i\leq t$.
    \end{itemize}
\end{theorem}

A similar version holds for the refined enumerator, as shown by \cite{hu2020weight}, from which one can determine the biased distances for the code (Thm.~\ref{thm:refined_dist}).

As one can read off the distances from the enumerators, our tensor network method provides a straightforward way to compute and verify adversarial distances for all quantum codes whose QL description is known. This provides the first viable method to compute distances for a quantum code that need not be a stabilizer code.

\subsection{Error Detection}\label{subsec:errdet}
With weight enumerators in hand, we can easily obtain the probability for uncorrectable errors \cite{Ashikhmin:1999ef}. For any quantum code $\mathcal{C}$, let $\Pi_{\mathcal{C}}$ be the projector onto the code subspace, and write the orthogonal projector onto $\mathcal{C}^{\perp}$ as $\Pi_{\mathcal{C}}^{\perp}$. We say an error $E$ \emph{uncorrectable} if it cannot be detected, that is
$\Pi_{\mathcal{C}}E\Pi_{\mathcal{C}} \propto \Pi_{\mathcal{C}},$ and is not proportional to the logical identity.
Operationally, one performs a measurement with respect to $(\Pi_{\mathcal{C}},\Pi_{\mathcal{C}}^{\perp})$. An error is detected if the result is contained in $\mathcal{C}^{\perp}$. For stabilizer codes, this corresponds to errors with trivial error syndrome that  perform a non-identity logical operation.

Consider depolarizing channel with unbiased noise which acts identically on any single qubit with
\begin{equation}\label{eqn:depolarizing_channel}
\rho_j\rightarrow (1-3p)\rho_j+p X\rho_j X+p Y\rho_j Y+p Z\rho_j Z,
\end{equation}
where $\rho_j$ is the reduced density matrix on site $j$.
For stabilizer codes, it is easy to check that the probability of the random Pauli errors coinciding with a non-trivial logical operator is nothing but $p_{ne}=(B^{norm}-A^{norm})(z=p,w=1-3p)$ because a Pauli error with weight $d$ occurs with probability $p^d(1-3p)^{n-d}$. As in Theorem~\ref{thm:3}, we have taken the enumerators to be normalized such that $A_0=B_0=1$. In general, \cite{Ashikhmin:1999ef} shows that the error probability for any code with $\dim\mathcal{C}=K$ is 
\begin{equation}p_{ne}=\frac{K}{(K+1)} \big(B^{norm}(p,1-3p) -A^{norm}(p,1-3p)\big).\end{equation}
Note the overall multiplicative factor compared to our initial estimation using stabilizer code because some logical errors takes the initial codeword to a non-orthogonal state, but only the orthogonal component is counted as non-trivial logical error in this construction.

We can extend the argument of \cite{Ashikhmin:1999ef} to more general error models. Suppose the error channel is given by
\begin{equation}\rho_j\rightarrow \sum_{i=1}^{q^2} K_i \rho_j K_i^{\dagger},\end{equation}
which acts identically across all physical qudits, then on the whole system, the errors act as 
\begin{equation}\label{eqn:general_channel}
   \mathcal{E}(\rho)= \sum_{\mathbf{i}} \mathcal{K}_{\mathbf{i}} \rho \mathcal{K}_{\mathbf{i}}^{\dagger}
\end{equation}
where 
\begin{equation}\mathcal{K}_{\mathbf{i}}=K_{i_1}\otimes K_{i_2}\otimes \dots\otimes K_{i_n},\end{equation} and $\mathbf{i}$ is summed over all $q^2$-nary strings of length $n$. It is important to note that for each $\mathbf{i}$, the Kraus operator and its conjugate are the same, there are no cross terms.

\begin{theorem}\label{thm:gen_err}
    The non-detectable error probabilities of any error channel with the form of (\ref{eqn:general_channel}) is given by 
    \begin{align}\label{eqn:ne_err_prob}
    {p}_{nd} &= \frac{K}{K+1}\Big(\frac 1 K \sum_{\mathbf{i}} \Tr[\mathcal{K}_{\mathbf{i}}^{\dagger}\Pi \mathcal{K}_{\mathbf{i}} \Pi]\\\nonumber
    &-\frac{1}{K^2}\sum_{\mathbf{i}} \Tr[\mathcal{K}_{\mathbf{i}}^{\dagger}\Pi]\Tr[\mathcal{K}_{\mathbf{i}}\Pi]\Big)
\end{align}
for a quantum code with dimension $K$ with projector $\Pi$.
\end{theorem}
\begin{proof}
    See Appendix~\ref{app:errorchannel}.
\end{proof}

For instance, in the depolarizing channel~(\ref{eqn:depolarizing_channel}) each $\mathcal{K}_{\mathbf{i}}$ is simply a Pauli string $E\in\mathcal{E}^n$ weighted by $p^{wt(E)}(1-3p)^{n-wt(E)}$. Substituting we find that the two terms in (\ref{eqn:ne_err_prob}) are simply the enumerator polynomials $A$ and $B$ evaluated at $z=p$ and $w=1-3p$ as expected.

\subsubsection{General error channels in the Pauli basis}
For each $K_a$, its Pauli decomposition $K_a=\sum_E c_E^a E$ allows us to re-express the error probability in terms of the generalized weight enumerators in Sec~\ref{subsec:genabs}.  In such cases, we can re-organize the sum over $i$ by Pauli types. Again, let the noise model be single qubit errors that are identical across all physical qubits such that 
\begin{equation}
\label{eqn:krausdecomp}
    \rho\rightarrow \sum_i^{q^2} K_i \rho K_i^{\dagger} = \sum_{P,\bar{P}} k_{P\bar{P}} P\rho \bar{P}^{\dagger}.
\end{equation}
Let us label each $P\bar{P}$ pair as $G$ so that $|\{G\}|=q^4$ and so write $k_G=k_{P\bar{P}}$. For example,  $\{{G}\}=\{II,IX,XI,IZ,ZI,XX,ZZ\dots\}$ (all 16 arrangements) for $q=2$.

Then let $wt_G^n$ be a weight function 
\begin{equation}
    wt_G^n (E\otimes F) = \sum_{i=1}^n wt_{G}(E_i\otimes F_{i}^{\dagger}) 
\end{equation}
where 
\begin{equation}
    wt_{G}(E_i\otimes F_i)=\begin{cases}
    1~\text{if $E_i\otimes F_i=G$}\\
    0~\text{otherwise,}
    \end{cases}
\end{equation}
and $E\otimes F = \bigotimes_i E_i\otimes F_i$. Thus $wt_G^n$ counts the number of times $G=P\otimes \bar{P}$ appears in a string $E\otimes F$ where $E,F$ each has length $n$. The relevant terms can then be expanded in this basis as
\begin{widetext}
\begin{align}\label{eqn:kraus_enum}
    B(\{k_G\};\Pi,\Pi)=\sum_{\mathbf{i}}\Tr[\mathcal{K}_i\Pi \mathcal{K}_i^{\dagger}\Pi ] &= \sum_{E,F\in \mathcal{E}^n}\Tr[E\Pi F^{\dagger}\Pi] \prod_G k_{{G}}^{wt_G^n(E\otimes F)}\\
    A(\{k_G\};\Pi,\Pi)=\sum_{\mathbf{i}}\Tr[\mathcal{K}_i\Pi]\Tr[ \mathcal{K}_i^{\dagger}\Pi ] &= \sum_{E,F\in \mathcal{E}^n}\Tr[E\Pi]\Tr[F^{\dagger}\Pi] \prod_G k_{{G}}^{wt_G^n(E\otimes F)}
\end{align}
\end{widetext}
We can then distill a set of enumerators sufficient in describing the effect of all error channels 

\begin{align}\label{eqn:general_enum}
    &\bar{A}(\mathbf{u}_G; M_1, M_2) \\\nonumber
    &\quad = \sum_{E,F\in \mathcal{E}^n}\Tr[EM_1]\Tr[F^{\dagger}M_2] \mathbf{u}^{wt^n_G(E\otimes F)}\\
    &\bar{B}(\mathbf{u}_G; M_1, M_2)\\\nonumber
    &\quad = \sum_{E,F\in \mathcal{E}^n}\Tr[EM_1 F^{\dagger}M_2] \mathbf{u}_G^{wt_G^n(E\otimes F)},
\end{align}
where 
\begin{align*}
    &\mathbf{u}_G^{wt_G^n(E\otimes F)} \\
    &\quad =\underbrace{ u_{II}^{wt_{II}^n(E\otimes F)} u_{IP_1}^{wt_{IP_1}^n(E\otimes F)} \dots  u_{P_qP_q}^{wt_{P_qP_q}^n(E\otimes F)}}_{\textrm all~q^4~terms}.
\end{align*}
We see this is nothing but a specific form of the generalized enumerator we introduced in Sec~\ref{subsec:genabs}. Note that we only need to compute the relevant enumerators once. The effects of different error models are now completely captured by the polynomials and can be evaluated by inserting the relevant values of $c_G$.

By substituting the proper expressions for Kraus operators, we are now in a position to rephrase all identical single qubit error channels in the form of weight enumerators. 
In practice, computing the generalized enumerator that accommodates arbitrary error channels can be rather expensive. Even for qubits, we would in general require 16 different variables in a polynomial. Fortunately for common channels, the computation simplifies and it is possible to express them with a much smaller set. As the Kraus representations are not unique it may be possible that some representations yield more succinct expressions than others. For pedagogical reasons, let us apply this to a few common error channels on qubits. 

\subsubsection{Biased Pauli Errors}
For a noise model where bit flip ($X$) error and phase ($Z$) error can occur independently on physical qubits with probability $p_x,p_z$ respectively. The error channel is 
\begin{align*}
    \rho\rightarrow &(1-p_x-p_z+p_xp_z)\rho+(p_x-p_xp_z) X\rho X\\
    &+p_xp_z Y\rho Y+(p_z-p_xp_z) Z\rho Z.
\end{align*}
For stabilizer codes, the probability that the Pauli error coincides with a non-trivial logical operation is given by the normalized double weight enumerator of \cite{hu2020weight}:
\begin{equation}\label{eqn:doubleenum}
(D-D^{\perp})(x,y,z,w),\end{equation}
evaluated at $x = 1-p_x$, $y = p_x$, $z = 1-p_z$ and $w=p_z$. Applying Theorem \ref{thm:gen_err}, we see that the actual non-correctable logical error probability is given by (\ref{eqn:doubleenum}) but again modified by multiplicative factor $K/(K+1)$ when taken into account the effect of non-orthogonal states.

Similarly, a channel where all Pauli errors have different independent error probabilities 
\begin{equation*}\rho\rightarrow (1-p_x-p_y-p_z)\rho + p_x X\rho X+p_y Y\rho Y+p_z Z\rho Z\end{equation*}
have non-correctable error probability given by the complete enumerators, 
\begin{align}
    p_{ne}=&\frac{K}{K+1}\Big(F(p_x,p_y,p_z,1-p_x-p_y-p_z)\\\nonumber
    -&E(p_x,p_y,p_z,1-p_x-p_y-p_z)\Big).
\end{align}
For definitions of $D,D^{\perp}, E,F$, see \cite{hu2020weight,CL2022} or App.~\ref{app:scalarenum}.

\subsubsection{Coherent error} 
Pauli errors are in some sense classical; for a coherent quantum device, unitary errors are also relevant. Compared to Pauli errors, studies of the impact of coherent errors are less common \cite{TNrealnoise,2017coherentrepcode,FFcoherent} partly hampered by the computational costs. Nevertheless various methods exist. Here we examine a special case of single qubit coherent error and express it in terms of weight enumerator polynomials. 
Suppose we have single qubit/qudit coherent error applied identically to all physical qubits
\begin{equation}
    \rho_i\rightarrow U_i\rho_i U_i^{\dagger}
\end{equation}
acting on each qubit $i$, where each unitary can be decomposed as 
\begin{equation}
    U_i=aI_i+bX_i+cY_i+dZ_i.
\end{equation}

The logical error probability is

\begin{align}
    \bar{p}_{nd} =& \frac{K}{K+1}\Big(\frac 1 K \Tr[U^{\dagger}\Pi U\Pi]\\\nonumber
    &-\frac{1}{K^2}\Tr[U^{\dagger}\Pi]\Tr[U\Pi])\Big).
\end{align}
Expanding $U=\bigotimes_i U_i$ in the Pauli basis, we have $U^{\dagger}=\sum_E k_E^* E$ and $U=\sum_F k_F F$ where we sum $E,F$ over all length $n$ Pauli strings. As coefficients $k_E,k_F$ only depends on the number of Paulis that appear in $F$

\begin{equation}
    k_F = a^{n-w(F)} b^{w_x(F)}c^{w_y(F)}d^{w_z(F)},
\end{equation}
each term in the overall probability is
\begin{widetext}
\begin{align}
    \Tr[U^{\dagger}\Pi U\Pi] &= \sum_{E, F\in \mathcal{E}^n} \Tr[E\Pi F\Pi] k^*_E k_F \\\nonumber
    &=  \sum_{E, F\in \mathcal{E}^n} \Tr[E\Pi F\Pi] a^{n-w(F)}{b}^{w_x(F)}c^{w_y(F)}d^{w_z(F)} \bar{a}^{n-w(E)} \bar{b}^{w_x(E)}\bar{c}^{w_y(E)}\bar{d}^{w_z(E)}\\
        \Tr[U^{\dagger}\Pi]\Tr[U\Pi] &= \sum_{E, F\in \mathcal{E}^n} \Tr[E\Pi]\Tr[F\Pi] k^*_E k_F \\\nonumber
    &=  \sum_{E, F\in \mathcal{E}^n} \Tr[E\Pi]\Tr[F\Pi] a^{n-w(F)}{b}^{w_x(F)}c^{w_y(F)}d^{w_z(F)} \bar{a}^{n-w(E)} \bar{b}^{w_x(E)}\bar{c}^{w_y(E)}\bar{d}^{w_z(E)}.
\end{align}
\end{widetext}

These are nothing but the generalized versions of the complete weight enumerators 

\begin{align}
    &A(\mathbf{u}_1,\mathbf{u}_2; M_1, M_2)\\\nonumber
    &= \sum_{E,F\in \mathcal{E}^n} \Tr[E^{\dagger}M_1]\Tr[F M_2]\mathbf{u}_1^{wt(F)}\mathbf{u}_2^{wt(E)}\\
    &B(\mathbf{u}_1,\mathbf{u}_2; M_1, M_2) \\\nonumber
    &= \sum_{E,F\in \mathcal{E}^n} \Tr[E^{\dagger}M_1 F M_2]\mathbf{u}_1^{wt(F)}\mathbf{u}_2^{wt(E)}
\end{align}
evaluated at $M_1=M_2=\Pi$, $w_1=a, x_1=b, y_1 =c, z_1=d$ and $w_2, x_2,y_2,z_2$ at their complex conjugates. To simplify the notation, we absorbed each 4-tuple of variables into abstract variables $\mathbf{u}_i$ and weight functions $wt(\cdot)$ such that 

\begin{equation}\mathbf{u}_i^{wt(E)} = x^{wt_{x}(E)}y^{wt_{y}(E)}z^{wt_{z}(E)}w^{n-wt(E)}.\end{equation}

\subsubsection{Amplitude damping and dephasing channels}
Amplitude damping channel is relevant for superconducting qubits. Its has a Kraus representation with operators

\begin{align}
K_0=&
    \begin{pmatrix}
    1 & 0\\
    0 & \sqrt{1-\gamma}
    \end{pmatrix}
    \\\nonumber
    =&\frac 1 2( 1 +\sqrt{1-\gamma})I +\frac 1 2 (1-\sqrt{1-\gamma})Z\\\nonumber
    K_1=&
    \begin{pmatrix}
    0 & \sqrt{\gamma}\\
    0 & 0
    \end{pmatrix}
    =\frac{\sqrt{\gamma}}{2} X + \frac{i\sqrt{\gamma}}{2} Y
\end{align}
In this case, we only need to keep 8 distinct variables $\{u_{II},u_{IZ},u_{ZI},u_{ZZ},u_{XX},u_{XY},u_{YX},u_{YY}\}$ as the remaining coefficients $k_{G}$ are 0. In fact, the nonzero coefficients further satisfy $k_{IZ}=k_{ZI}=\lambda_I \lambda_Z, k_{PP}=|\lambda_P|^2, k_{XY}=\bar{k}_{YX}=\lambda_X\bar{\lambda}_Y$ where $\lambda_P$ are the coefficients in the Pauli expansion of the Kraus operators. Therefore, the end polynomial would only require 4 independent variables $\{\lambda_{I},\lambda_X,\lambda_Y,\lambda_Z\}$. In other words, when summing the polynomial in practice, we only sum over the qudit strings of local dimension 4 where the coefficients for $G$ are non-vanishing. Furthermore, one can rewrite the nonzero coefficients as
\begin{equation}
    \prod_G k_G^{wt_G(E\otimes F)} = \prod_P \lambda_P^{wt_P(E)}\prod_{\bar{P}}\bar{\lambda}_{\bar{P}}^{wt_{\bar{P}}(F)}
\end{equation}
which depends on 4 parameters and is no more complicated than the complete enumerator.

For a dephasing channel, $K_0$ remains the same while
\begin{equation}
    K_1=\sqrt{\gamma}\begin{pmatrix}
    0 & 0\\
    0 & 1
    \end{pmatrix}
    =\frac{\sqrt{\gamma}}{2} (I-Z).
\end{equation}

Expanding and simplifying, we find that it only depends on two nonzero coefficients $c_{II}=(1+\sqrt{1-\gamma})/2$ and $c_{ZZ}=(1-\sqrt{1-\gamma})/2$.
Thus this is even easier than computing the original weight enumerator! 
Furthermore, instead of summing over the full Pauli group, we only need to sum over $E\in\mathcal{Z}^n$ where $\mathcal{Z}$ is the set of Pauli strings that only contains $I$ or $Z$.

\subsection{Effective Distance}
While adversarial distance is a useful measure of the goodness of a code, it is also informative to devise more refined measures like effective distances \cite{VQAQEC,Cliffdef} that serve as useful benchmarks of code performance with respect to different error profiles. 
For example, recall that \cite{Cliffdef} defines an effective distance 
\begin{equation}
    d'= \mathcal{N}^{-1}\log(p_{0}(1-p)^{-n})
\end{equation} 
for codes under depolarizing channel, where $p=p_X+p_Y+p_Z$ and $\mathcal{N}$ is some normalization factor that depends on the physical error probabilities. In the original definition, $p_{0}$ is the probability where the Pauli noise implements the most likely non-trivial logical operator. Using enumerators, we can also produce more precise effective distances under depolarizing noise, where $p_{0}$ is replaced by the probability $p_{ne}$ where Pauli noise implements \emph{any} non-trivial logical operator. Similar measures have been used to quantify effective code performance \cite{VQAQEC,QLRL}. For example, one can define another effective distance for some $c_1, c_2$ 
\begin{equation}
    d_{\rm eff} ={c}_1 \log (p_L)+c_2
\end{equation}
such that $d_{\rm eff}$ is higher for lower error rate $p_L$.

Similar to \cite{QLRL}, we also use the normalized logical error probability \begin{equation}p_L^{\rm norm}= p_L/p_{s=0}\end{equation} as a measure of code performance throughout this work.
Here $p_{s=0}$ is the probability of error non-detection and better protection corresponds to a smaller normalized error rate. This is not a distance measure and it corresponds to the probability of uncorrectable error where the ``error correction'' protocol simply discards the quantum state upon detecting an error. 

\subsection{Subsystem codes and Mixed Enumerator}\label{subsec:mixedenum}
The above applications are general and can be used for any quantum code. Let us now focus on a few more applications that are most closely tied to stabilizer codes and subsystem codes.

Mixed enumerators are made by tracing together tensor enumerator of both $\mathbf{A}(\mathbf{u})$ and $\mathbf{B}(\mathbf{u})$ types. 

\begin{proposition}
Let $M(\mathbf{u})$ be a mixed enumerator polynomial obtained from tracing tensor enumerators of $A$ and $B$ types. MacWilliams transform on $M(\mathbf{u})$ produces a dual polynomial $M^{\perp}(\mathbf{u})$ which, up to normalization, can be built from the same tensor network where we exchange the $A$ and $B$ type tensors. 
\end{proposition}

\begin{proof}
The MacWilliams transform commutes with trace as long as the generalized Wigner transform is its own self-inverse up to a constant multiple. This is clearly the case when the tensor enumerators are diagonal, when the generalized Wigner transform reduces to regular Wigner transform. The same must also hold true for the generalized transform, as the MacWilliams transform commutes with trace when the tensor enumerators are not mixed.
\end{proof}

A key application is finding the distance of subsystem codes where we need to enumerate all gauge-equivalent representations of the logical operators.
It is convenient to think of the subsystem code as a stabilizer code encoding multiple logical qubits where some of them are demoted to gauge qubits. To obtain its distance, we first enumerate all logical operators, which is given by ${B}(\mathbf{u})$ of the stabilizer code. This can again be obtained by ${A}(\mathbf{u})$ and applying the MacWilliams identity. We then need to enumerate all gauge equivalent logical identities of this subsystem code ${I}(\mathbf{u})$. Technical details in obtaining $I(\mathbf{u})$ can depend on the specific tensor network in question. However, it is rather straightforward if all logical legs in the network are independent, i.e., encoding map defined by the QL tensor network has trivial kernel. For example, this is the case for holographic code, but not for the Bacon-Shor code tensor network. 

For encoding tensor networks that have trivial kernel, we can divide the input legs, which we call logical legs in \cite{CL2021}, into two categories: (i) the ones where operator pushing produce logical operators, which we now call logical legs, and (ii) the ones that alter the state of gauge qubits, which we now call gauge legs. Let us first assume that each tensor has only one such an input leg that is either logical or gauge, which is the case for the holographic code. To enumerate the logical identity, we construct a mixed enumerator --- for each tensor in the QL network whose input leg is logical, we contract the tensor enumerator $\mathbf{A}(\mathbf{u})$ of the local atomic code (e.g. the $[[5,1,3]]$ code in HaPPY) on the corresponding vertex in the enumerator tensor network. If the tensor in the QL network has a gauge leg, then we contract the $\mathbf{B}(\mathbf{u})$ tensor enumerator of the local atomic code (Figure~\ref{fig:holographic}). The resulting tensor network enumerates the weights of all $g\bar{I}$ for $g\in \mathcal{G}$. Then the difference $\tilde{C}(\mathbf{u})=B(\mathbf{u})-{I}(\mathbf{u})$ between these enumerators only contains the weights of non-identity logical operators, which informs us about the distance. This is also known as the \emph{word distance} \cite{Harris2020,uberholography}.

Similarly, if we want to compute the distance of a logical qubit in the stabilizer code (i.e. all logical, no gauge qubits), then we only enumerate the stabilizer equivalent logical operations that act non-trivially on that qubit. For this we insert $\mathbf{B}(\mathbf{u})$ on the vertex containing the logical qubit of interest and $\mathbf{A}(\mathbf{u})$ everywhere else. This enumerator now only counts the stabilizer equivalent of that particular logical qubit, instead of all logical qubits, like $B(\mathbf{u})$. This can be quite relevant in the holographic code, where the central bulk qubit can have a distance that scales with system size, whereas the ones on the peripheral have constant distance \cite{Harris2020}.

For instance, in the holographic HaPPY code, Figure~\ref{fig:holographic}, one can treat the system as a stabilizer code. Then the \emph{stabilizer distance} can be determined by counting all non-identity logical operators associated with a particular bulk qubit. In the figure we choose the logical qubit living on the central tile. The stabilizer distance is then the minimum power of $z$ in $C_0(z)$ for which the coefficient is non-zero. If we treat it as a subsystem code, then the distance should instead be counted by including the logical operators of other bulk qubits as gauge qubits using the enumerator $\tilde{C}_0(z)$. 
\begin{figure}
    \centering
    \includegraphics[width=\linewidth]{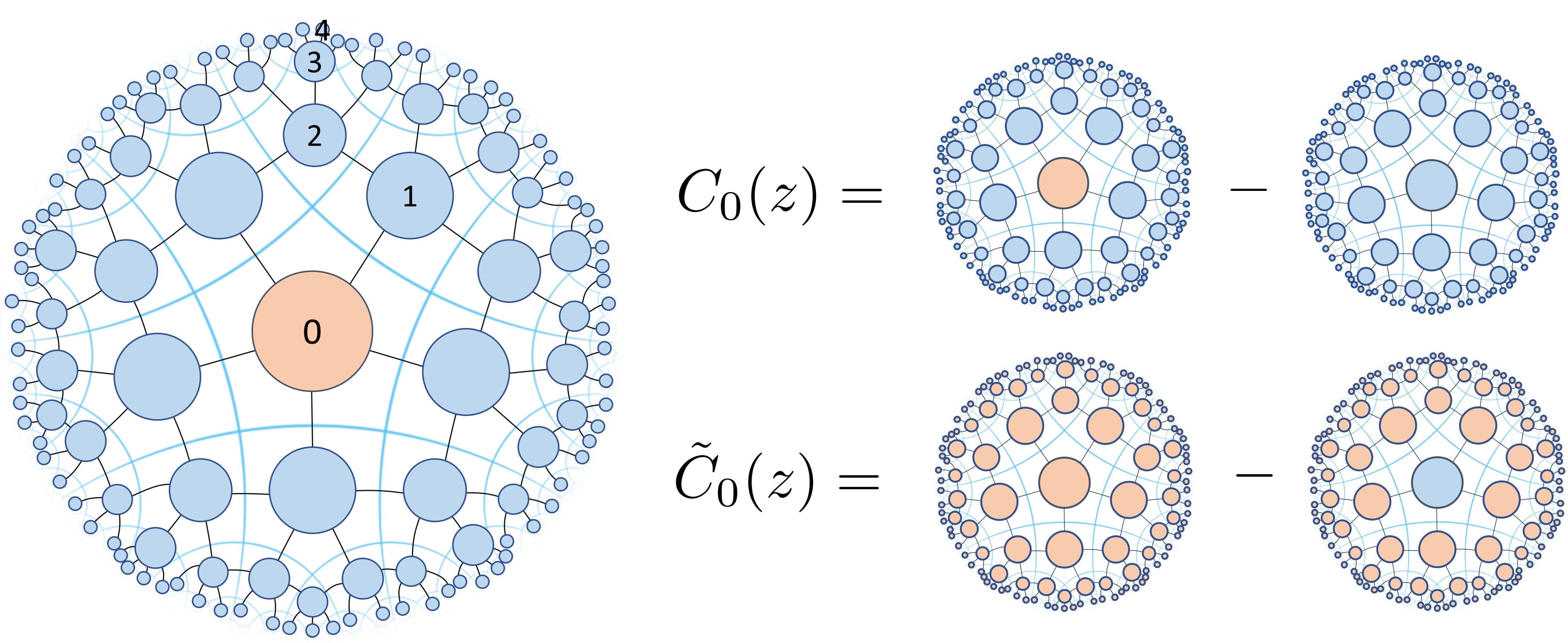}
    \caption{Left: tensor network for a mixed enumerator of the holographic HaPPY pentagon code, where blue indicates insertion of $\mathbf{A}(z)$ and orange $\mathbf{B}(z)$. Right: different tensor networks compute the different distribution of logical operators. The same exercise can be repeated for logical qubits at different distances from the center (labelled 0,1,2,3,4). }
    \label{fig:holographic}
\end{figure}

If each tile has multiple input legs, some of which are gauge and others logical, we then need to make slight modifications to the tensor enumerators used in the above prescription. For a word distance computation, we send $\mathbf{A}(\mathbf{u})\rightarrow \mathbf{A}'(\mathbf{u})$ such that $\mathbf{A}'(\mathbf{u})$ enumerates all logical identity operators of the local atomic code, e.g. the $[[4,2,2]]$ atomic code on even columns of a 2d Bacon-Shor code tensor network (Fig.~\ref{fig:baconshor}). In other words, we enumerate all elements of the non-abelian gauge group $\mathcal{G}$. For Pauli operators, this modification is rather straightforward as we simply count the number of operators that act as identity on the logical legs. More precisely, let 
\begin{equation}
    \Pi'=\frac{1}{|\mathcal{G}|}\sum_{g\in\mathcal{G}} g
\end{equation}
and prepare $\mathbf{A}'(\mathbf{u})=\mathbf{A}^{(J)}(\mathbf{u};M_1=\Pi',M_2=\Pi')$ as a reduced tensor enumerator.  

For stabilizer distance computations, we send $\mathbf{B}(\mathbf{u})\rightarrow\mathbf{B}'(\mathbf{u})$, the latter of which enumerates the number of logical operators that act as the identity on gauge qubits. This can be prepared by a similar reduced tensor enumerator such that $\mathbf{B}'(\mathbf{u}) = \mathbf{A}^{(J)}(\mathbf{u};M_1=\Pi'',M_2=\Pi'')$
where \begin{equation}\Pi''\propto \sum_{g\in \mathcal{G}'}g\end{equation} and $\mathcal{G}'$ is generated by the center of $\mathcal{G}$ and logical operators $\{\mathcal{L}_{\rm logical}\otimes I_{\rm gauge}\}$ that act as identity on the gauge qubits. In other words, we construct a new gauge group $\mathcal{G}'$ where we swapped the roles of the gauge and logical qubits in the original code defined by $\mathcal{G}$.

Finally, for a tensor network whose encoding map has a non-trivial kernel, i.e., the logical legs are inter-dependent, one should take extra care in applying the above recipe for building a useful mixed enumerator. For instance, in the Bacon-Shor code tensor network (Fig.~\ref{fig:baconshor}), multiple input legs are inter-dependent and several of them correspond to the same logical or gauge degree of freedom. One then needs to make sure that the type of the legs (gauge vs logical) is being tracked consistently across different tensors when contracting the tensor network. 

Instead of the mixed enumerators introduced above, we can also directly use tensor enumerators  to study subsystem codes. This is similar to the approach by \cite{LTNC,CL2022}. The recipe for building a relevant tensor enumerator of the code is quite similar. For each tensor that contains the logical leg/qubit, we put down a tensor enumerator $\mathbf{B}(z)$ of the \emph{encoding tensor} at that node (e.g. the tensor of the $[[6,0,4]]$ state in the HaPPY code) in the tensor network, except now that we keep the logical index open in addition to the contracted legs of the tensor. The components of the resulting tensor enumerator $\mathbf{B}(z)$ now contains the weight distribution for each logical Pauli operator. This allows us to read off the distances for each logical operator, after subtracting off the part that enumerates the logical identity by fixing certain tensor indices to $0$. This can be performed efficiently if the number of open logical legs is not too many although the number of gauge qubits can still be high\footnote{In practice, it may be more efficient to compute the tensor enumerator $\mathbf{A}(z)$ of the entire code, then perform a MacWilliams transform on the tensor enumerator.}.

\subsection{Higher Genus Enumerator}
{\color{black} There is a well understood link between multiple weight enumerators of classical codes and modular forms in number theory, wherein the degree of the modular form is the number of codewords whose weight is being enumerated \cite{duke1993codes, runge1996codes}. Consequently, that number is now typically referred to as the ``genus'' of the enumerator. Just as in the classical case, we can extend this to higher genus quantum weight enumerators by introducing weight functions that count the number of factors where tuples of error operators realize specific error patterns. We can also study subsystem codes using genus two enumerators.} 

For concreteness, consider genus $g = 2$. We introduce $q^4$ variables $\mathbf{u} = (u_{G_1,G_2} \::\: G_1,G_2 \in \mathcal{E})$, and weight function $\mathrm{wt}:\mathcal{E}^n\to \mathbb{Z}^{q^4}$ that counts factors
\begin{equation}
    \mathrm{wt}(E,F) = \#\{ i \::\: E_i = G_1,\ F_i = G_2,\ G_1,G_2\in\mathcal{E}\}.
\end{equation}
The genus-$2$ weight enumerators of Hermitian operators $M_1,M_2$ on $\mathfrak{H}\otimes\mathfrak{H}$ are\\
\begin{widetext}
\begin{align}
    A^{(2)}(\mathbf{u};M_1,M_2) &= \sum_{E,F\in\mathcal{E}^n} \Tr((E\otimes F)M_1)\Tr((E\otimes F)^\dagger M_2) \mathbf{u}^{\mathrm{wt}(E,F)}\\
    B^{(2)}(\mathbf{u};M_1,M_2) &= \sum_{E,F\in\mathcal{E}^n} \Tr((E\otimes F)M_1(E\otimes F)^\dagger M_2) \mathbf{u}^{\mathrm{wt}(E,F)}.
\end{align}
\end{widetext}

Notice the coefficients of these enumerators are just what would use for a systems with $2n$ factors. The interesting addition is the additional variables to track correlations in the weights of $E$ and $F$. Indeed if we were to ignore these correlations and evaluate $u_{G_1,G_2} = u_{G_1}u_{G_2}$ then we recover the ordinary enumerators:

\begin{align}
    &A^{(2)}(\{u_{G_1}u_{G_2}\};M_1\otimes M_1',M_2\otimes M_2')\\\nonumber
    &\quad = A(\{u_G\};M_1,M_2)\cdot A(\{u_G\};M_1',M_2')\\\nonumber
    &\quad = A(\{u_G\};M_1\otimes M_1',M_2\otimes M_2').
\end{align}
and similarly for $B^{(2)}$.

To capture new information in the higher genus enumerators, we evaluate their variables in interesting ways. For example, consider the case where $M_1 = M_2 = \Pi_1\otimes \Pi_2$ where $\Pi_1$ and $\Pi_2$ are projections that need not commute. Evaluating
\begin{equation}u_{G_1,G_2} = \left\{\begin{array}{cl} u_{G_1} & \text{if $G_1 = G_2$}\\ 0 & \text{if $G_1 \not= G_2$,}\end{array}\right.\end{equation}
we have
\begin{align}
    \mathbf{u}^{\mathrm{wt}(E,F)} &= \prod_{G_1,G_2} u_{G_1,G_2}^{\mathrm{wt}_{G_1,G_2}(E,F)}\\\nonumber
    &= \left\{\begin{array}{cl} \prod_G u_G^{\mathrm{wt}_G(E)} & \text{if $E = F$}\\ 0 & \text{if $E \not= F$.}\end{array}\right.
    \label{eqn:weightfcng2}
\end{align}
Thus
\begin{align}
    &A^{(2)}(\mathbf{u};\Pi_1\otimes \Pi_2)\nonumber\\ &= \sum_{E\in\mathcal{E}^n} \Tr(E\Pi_1)^2\Tr(E\Pi_2)^2 \mathbf{u}^{\mathrm{wt}(E)}\\
    &B^{(2)}(\mathbf{u};\Pi_1\otimes \Pi_2)\nonumber\\ &= \sum_{E\in\mathcal{E}^n} \Tr(E\Pi_1 E\Pi_1)\Tr(E\Pi_2 E\Pi_2)\mathbf{u}^{\mathrm{wt}(E)}.
    \label{eqn:genus2}
\end{align}

In particular, consider a subsystem code whose gauge group decomposes as $\mathcal{G} = \mathcal{G}_1 \cup \mathcal{G}_2$ where each of $\mathcal{G}_1$ and $\mathcal{G}_2$ are maximal Abelian subgroups and $\mathcal{C}(\mathcal{G}) = \mathcal{G}_1 \cap \mathcal{G}_2$. This is the case for generalized Bacon-Shor codes \cite{BravyiGBSC} where $\mathcal{G}_1$ consists of the $X$-type generators of $\mathcal{G}$ and the row operators, while $\mathcal{G}_2$ is the $Z$-type generators and the column operators.\footnote{In fact this is true of every subsystem code: using the usual symplectic formalism of stabilizer groups, the gauge group becomes a subspace, and a Darboux basis for this subspace provides the two isotropic subspaces that characterize $ \mathcal{G}_1$ and $\mathcal{G}_2$.} Each of $\mathcal{G}_1$ and $\mathcal{G}_2$ could be considered a stabilizer in its own right, however the weight enumerators of these have little to do with subsystem code of $\mathcal{G}$.

Nonetheless, consider them as stabilizers of codes and write the projections onto their code subspaces as $\Pi_1 = \frac{1}{2^{n-k_1}} \sum_{S\in\mathcal{G}_1} S$ and $\Pi_2 = \frac{1}{2^{n-k_2}} \sum_{S\in\mathcal{G}_2} S$ where $k_1,k_2$ are the dimensions of these codes. Then
\be\Tr(E\Pi_1)^2 = \left\{\begin{array}{cl} 4^{k_1} & \text{if $E \in \mathcal{G}_1$} \\ 0 & \text{otherwise,}\end{array}\right.\ee
and similarly for $\Tr(E\Pi_2)^2$, and therefore
\begin{align}
    &A^{(2)}(\mathbf{u};\Pi_1\otimes \Pi_2) = 4^{k_1+k_2} \sum_{E\in\mathcal{G}_1\cap\mathcal{G}_2} \mathbf{u}^{\mathrm{wt}(E)}\\\nonumber
    &\quad = 4^{k_1+k_2} \sum_{d=0}^n \#(\mathcal{E}^n[d] \cap \mathcal{C}(\mathcal{G})) w^{n-d} z^d
\end{align}
is the enumerator of the stabilizer of the subsystem code of $\mathcal{G}$. Also
\be\Tr(E\Pi_1E\Pi_1) = \left\{\begin{array}{cl} 2^{k_1} & \text{if $E \in \mathcal{N}(\mathcal{G}_1)$} \\ 0 & \text{otherwise,}\end{array}\right.\ee
and similarly for $\Tr(E\Pi_2E\Pi_2)$. Hence
\begin{align}
    &B^{(2)}(\mathbf{u};\Pi_1\otimes \Pi_2) = 2^{k_1+k_2} \sum_{E\in\mathcal{N}(\mathcal{G}_1)\cap\mathcal{N}(\mathcal{G}_2)} \mathbf{u}^{\mathrm{wt}(E)}\\\nonumber
    &\quad = 4^{k_1+k_2} \sum_{d=0}^n \#(\mathcal{E}^n[d] \cap \mathcal{N}(\mathcal{G})) w^{n-d} z^d
\end{align}
is the enumerator for the logical operators of the subsystem code.

{\color{black}Notice that up to taking a logarithm and a constant offset, (\ref{eqn:genus2}) is also the stabilizer Renyi-2 entropy when $q=2$, $\mathbf{u}^{wt(E)}\rightarrow \mathrm{constant}$ for $\Pi=|\psi\rangle\langle\psi|$. The stabilizer Renyi entropy is a computable measure of non-stabilizerness.
\begin{definition}
The $\alpha$-stabilizer Renyi entropy\cite{Leone_2022} is given by
   \begin{equation}
    M_{\alpha}(|\psi\rangle) = (1-\alpha)^{-1} \log \sum_{E\in \mathcal{P}^n}\Xi^{\alpha}_E(|\psi\rangle)  -n\log 2  
\end{equation}
where  
\begin{equation}
    \Xi_E(|\psi\rangle) = \frac{1}{2^n} \langle \psi|E|\psi\rangle^2 = \frac{1}{2^n}\Tr[E\Pi]^2.
    \end{equation}
\end{definition}

We see that $\Xi_E(|\psi\rangle)^{\alpha}$ is precisely the coefficient of the type $A$ genus-$\alpha$ weight enumerator (\ref{eqn:weightfcng2}). 

As a corollary, because the higher genus enumerators can be computed using the tensor network method, this provides a method for computing the stabilizer Renyi entropy also. While there are existing methods based on matrix product state \cite{haug_efficient_2023,PauliMPS} that computes stabilizer Renyi entropy, the enumerator tensor network provides an alternative and more general method for computing the magic of a quantum manybody system.

Physically, this also establishes a new connection between weight enumerator and quantum manybody magic, in addition to the entanglement angle that has been explored in the form of sector length. 

}

\subsection{Coset Enumerator and errors with non-trivial syndrome}
Until now, we have been working with particular instances of weight enumerator polynomials, that is, $M_1=M_2=\Pi$ or related operators. For stabilizer codes, they recover the weight distributions of error operators that have trivial syndrome. However, for the purpose of decoding, it is also useful to learn the probability $P(E|s)$ for any error syndrome $s$.

Let $E$ be a Pauli error that gives syndrome $\sigma(E)=s$. We consider the probability $P(E\bar{L})$ of errors that are stabilizer equivalent to $E\bar{L}$, where $\bar{L}$ is any logical operator. If we have this distribution, then we can construct a maximum likelihood decoder by undoing the $E\bar{L}$ with the maximal probability of $P(E\bar{L})$ given syndrome $s$. Similarly, one could apply a Bayesian decoder where $E\bar{L}$ is applied with the probability $p(E\bar{L}|\sigma(E))$ for error correction.

\begin{definition}
A coset weight enumerator for a stabilizer code is given by $A^s(\mathbf{u};E_s,\Pi) = A(\mathbf{u};M_1,M_2)$ where $M_1=M_2^{\dagger}=E_s \Pi$ for some (Pauli) operator $E_s$ with syndrome $s$. Its ``dual'' enumerator is $B^s(\mathbf{u};E_s,\Pi)= B(\mathbf{u};M_1,M_2)$ where $M_1=\Pi$, $M_2=E_s\Pi E_s^{\dagger}$. Their tensorial versions are similarly defined with $M_1,M_2$ taking on these specific values. The same definition applies for the generalized enumerators $\bar{A}^s, \bar{B}^s$.
\end{definition}
Note that $M_1,M_2$ here are no longer hermitian for $A$ and the operators used for $A, B$ are different. As a result, the ``dual'' enumerator is very different from its usual form. We do not use or prove a MacWilliams identity in this work, though it may be interesting to see if an analogous relation exists. More generally, the coset enumerators can be defined for $E_s$ that are not Pauli operators so long as $\Pi_s= E_s\Pi E_s^{\dagger}$ projects onto the error subspace. As the definition of such subspaces and syndromes can be highly code-dependent, here we focus on the instance where $E_s$ are Pauli's.

\begin{proposition}
Up to an overall normalization, the coefficients of the coset enumerator $A^s(\mathbf{u};E_s,\Pi)$ counts the number of coset elements in $E_s\mathcal{S}$ while $B^s(\mathbf{u}; E_s,\Pi)$ enumerates the number of elements $E_s\mathcal{N}(\mathcal{S})$.
\end{proposition}
\begin{proof}

Let 
\be\Pi_s=\frac{1}{|\mathcal{S}|}\sum_{D\in E_s\mathcal{S}} D\ee then for any $E\in \mathcal{E}^n$

\be \Tr[E\Pi_s]\Tr[E^{\dagger}\Pi_s^{\dagger}] = |\Tr[E\Pi_s]|^2 = q^{2n}/|\mathcal{S}|^2\ee if $E\in E_s \mathcal{S}$ and zero otherwise.
Hence, up to a constant normalization factor, the coefficient of the coset enumerator counts the number of coset elements of a particular weight. As we do not track signs in the distribution, no generality is lost by choosing the left vs right coset.

The $B$ type enumerators have coefficients
\be\Tr[E\Pi E^{\dagger}\Pi_s] = \Tr[\Pi_E \Pi_s]\ee
where $\Pi_s=E_s\Pi E_s^{\dagger}, \Pi_E=E\Pi E^{\dagger}$. As these projectors are orthogonal for different syndromes in a stabilizer codes, this coefficient is only non-trivial when $E\in E_s\mathcal{N}(\mathcal{S})$, i.e., when $E$ is any logical error with syndrome $s$. Therefore, up to normalization, we again obtain an enumerator that captures the weight distribution of $E_s\mathcal{N}(\mathcal{S})$. 
\end{proof}

Practically, the process of preparing this enumerator using tensor network is the same as before except we modify the values of $M_1$ and $M_2$.
First we identify the physical qubits on which $E_s$ has support.
Suppose $E_s$ acts on a particular atomic code non-trivially with $E_s^T$, then we prepare the $A$-type tensor coset enumerator of this code with $M_1=M_2^{\dagger}=E_s^T\Pi^T$ where $\Pi^T$ is the projection onto the code subspace of the local QL. 
Such a tensor enumerator counts elements in the coset $E_s^T\mathcal{S}^T$. 
We then repeat this for all such tensors. For ones that $E_s$ does not have support, we compute their tensor enumerator with $M_1=M_2=\Pi^T$ as usual. Then we contract these tensor enumerators in the same way as we did for building $A(\mathbf{u};M_1,M_2)$, e.g. Figure~\ref{fig:surface_coset}. The resulting enumerator polynomial is the desired $A_s(\mathbf{u};E_s\Pi)$. 
Also note that $M_1, M_2$ take on a special form that satisfy Proposition \ref{prop:stab_dt}, hence we can compute it more efficiently using a tensor network with reduced bond dimension, much akin to its weight enumerator counterparts. 

With these weight distributions, it is obvious that we can then compute $P(E|s)$.
For example, suppose we are given the coset enumerator $A^s(z,w;E_s,\Pi)$ for a code space defined by $\Pi$, then under symmetric depolarizing channel with physical error rate $p$, 
\be p_s= B^s(z=p,w=1-3p)/K\ee is the probability of returning an error syndrome $s$ with noiseless checks and $A^s(z=p,w=1-3p)/K^2$ is the probability of errors that are stabilizer equivalent to $E_s$. Indeed, this also extends trivially to double and complete enumerators by evaluating the polynomial at the respective parameters we used for the trivial syndrome examples in Sec.~\ref{subsec:errdet}. 

In fact, such kind of error probabilities generalize to any error channel. Similar to the non-detectable errors we have analyzed in the previous section, it is possible to compute $p(\bar{L}|s)$ using generalized enumerators as long as we replace $M_1,M_2$ by the appropriate values used in the coset enumerators. 

\begin{theorem}\label{thm:coset}
Consider a stabilizer code where $\Pi$ is the projection onto its code subspace of dimension $K$ and let $E_s$ be an error operator with syndrome $s$. Let the error channel be given by the Kraus form $\mathcal{E}(\cdot) = \sum_{\mathbf{i}}\mathcal{K}_{\mathbf{i}}\cdot \mathcal{K}_{\mathbf{i}}^{\dagger}$. Then 
\begin{align}
    p(E_s\mathcal{S}\cap s) &= \frac{1}{K(K+1)}\Big(\sum_{\mathbf{i}}\Tr[\mathcal{K}_{\mathbf{i}} \Pi \mathcal{K}_{\mathbf{i}}^{\dagger} \Pi_s] \\\nonumber
    & + \sum_{\mathbf{i}}\Tr[\mathcal{K}_{\mathbf{i}}\Pi{E}_s^{\dagger} ] \Tr[\mathcal{K}_{\mathbf{i}}^{\dagger}{E}_s\Pi]\Big)\\
    p_s&= \frac 1 K \sum_{\mathbf{i}}\Tr[\Pi \mathcal{K}_{\mathbf{i}}^{\dagger} \Pi_s \mathcal{K}_{\mathbf{i}}],
\end{align}
where $\Pi_s=E_s\Pi E_s^{\dagger}$.
\end{theorem}
Note that $E_s$ need not be a Pauli operator but can take on a general form $P_s\bar{L}$ where $P_s$ is a Pauli error with syndrome $s$ and $\bar{L}$ is any unitary logical operation. For proof and generalization where the logical error is a general quantum channel, see Appendix~\ref{subapp:gen_LEC}. We see that these terms in the expression share a great deal of similarities with Theorem \ref{thm:gen_err} which computes the logical error probability for trivial syndromes. Indeed, by expanding the Kraus operators in the Pauli basis, we see that these expressions can again be written as generalized weight enumerators (\ref{eqn:general_enum}) that we used to compute the uncorrectable error rates. To obtain these error probabilities, we follow the identical recipe by decomposing the Kraus operators in the Pauli basis $\sum_i K_i\cdot K_i^{\dagger}=\sum_{P\bar{P}} k_{P\bar{P}} P\cdot \bar{P}$ and evaluate $\mathbf{u}(k_{P\bar{P}})$ at the appropriate values based on that decomposition (c.f.~ eqn \ref{eqn:krausdecomp}). Finally, we  set $M_1=\Pi, M_2=E_s\Pi E_s^{\dagger}$ for the $B$ type enumerator and $M_1=\Pi E_s^{\dagger}, M_2=M_1^{\dagger}$ for the $A$ type.

Formally, 
\begin{align}
    p(E_s\mathcal{S}\cap s) &= \frac{1}{K(K+1)}\Big( B(\mathbf{k}; \Pi, \Pi_s) \\\nonumber&\quad + A(\mathbf{k}; \Pi E_s^{\dagger},  E_s\Pi) \Big)\\
    p_s &=\frac 1 K B(\mathbf{k}; \Pi, \Pi_s)
\end{align}
where $\mathbf{k}=\{k_{P\bar{P}}\}$ are the coefficients from the Pauli expansion.
With these coset enumerators in hand, we are now ready to discuss optimal decoders for general noise channels.

\subsection{Decoders from weight enumerators}\label{subsec:decoder}
We see that one can express the probability 

\begin{equation}
    p(E_S\mathcal{S}|s) = p(E_s\mathcal{S}\cap s)/p_s
\end{equation}
entirely in terms of weight enumerators. Suppose $E_s=P_s\bar{L}$ where $P_s$ is any Pauli error with syndrome $s$, which can be obtained by solving a set of $n-k$ linear equations, and $\bar{L}$ is a logical operator, then the set of probabilities $\mathcal{P}_s=\{p(P_s\bar{L}\mathcal{S}|s)\}$, as $\bar{L}$ runs over logical operators, is sufficient for us to perform error correction. It is customary to generate the set $\mathcal{P}_s$ for the set of $\bar{L}$ that are logical Pauli operations as they form an operator basis for the code subalgebra and are thus sufficient to generate the conditional probabilities for any unitary logical operators. 

\subsubsection{Maximum likelihood and Bayesian decoders}
It is straightforward to implement a maximum likelihood decoder where we identify the logical operator $\bar{L}_m$ for which $p(P_s\bar{L}_m\mathcal{S}|s)=\max \mathcal{P}_s$. Then error correction is performed by acting $\bar{L}$ and $P_s$ following the syndrome measurements. In this case, it is sufficient to compute just the $A$ enumerator because the $B$ enumerators are independent of $\bar{L}$ and only add to an overall normalization that not impact our choice of the maximum element. When multiple global maxima exist, then we choose one at random.

 One can also correct errors based on the probability distribution of $p(P_s\bar{L}\mathcal{S}|s)$ where we act on the state using operator $P_s\bar{L}$ with the selfsame probability. We call this the Bayesian decoder. As we require that $\sum_{\bar{L}}p(P_s\bar{L}\mathcal{S}|s)=1$ when summing over all Paulis, it is again sufficient to only compute the type $A$ enumerator as the constant from $B$ can be obtained by solving the normalization condition.

For each syndrome $s$, the complexity in implementing these decoders is therefore the complexity $\mathcal{C}(A^s)$ of computing $A^s$ from tensor contractions. For some codes this can be performed efficiently, which we further elaborate in Sec.~\ref{sec:complexity}. Nevertheless, even if each such contraction is efficient, we would still have to compute $q^{2k}$ number of enumerators as there are $q^{2k}$ number of distinct logical Pauli operators. Therefore the overall complexity estimate for such a decoder is $O(\mathcal{C}(A^s)q^{2k})$.

\subsubsection{Marginals}
For a code where $k$ is large, building the above maximum likelihood decoder remains challenging.  However, it is possible to compute the ``marginals'' efficiently. Let us write any logical Pauli operator $\bar{L}$ as
\begin{equation}
    \bar{L}(\mathbf{a},\mathbf{b})\propto \bigotimes_{i=1}^k X_i^{a_i}Z_i^{b_i},~ a_i,b_i=0,\dots q-1,
\end{equation}
where $\mathbf{a},\mathbf{b}$ are $k$-tuples with $\mathbf{a}=(a_1,a_2,\dots,a_k)$ and the same for $\mathbf{b}$.
Let $\bar{L}_i(a_i,b_i)$ be the logical Pauli acting on the $i$th logical qubit.  Consider the marginal error probability 

\begin{equation}
    p(E\bar{L}_i)=\sum_{a_j,b_j;\forall j\ne i} p(E\bar{L}(\mathbf{a},\mathbf{b})).
\end{equation}

This can be computed using the mixed enumerator. Recall that one can compute $p(\bar{I}_i|s=0)$ by inserting $B$-types tensor enumerators for all atomic blocks whose logical leg correspond to qubits $j\ne i$, and $A$ type for blocks whose logical legs correspond to qubit $i$. This enumerator, which we called $I_i(z)$, records the weight distribution of  logical operators $\mathcal{N}_i\subset \mathcal{N}(S)$, which consists of all Pauli logical operators in the code that act as the identity on the $i$-th logical qubit. If we treat the other qubits $\ne i$ as gauge qubits, we can think of it as recording all gauge equivalent representations of the identity operator. For example, this is what we have done for the Bacon-Shor code and the holographic code. For general cosets of operator $E_s$, we now compute mixed coset enumerator for operator $E_s$ such that $\sigma(E_s)=s$. Let us construct the mixed enumerators
\begin{align}
    M(\mathbf{u};E_s,\Pi)=\wedge_{J,J_j} &\Big[A^{J}_{(i)}(\mathbf{u};E_s^A,\Pi) \\\nonumber
    &\bigotimes_j B^{J_j}_{\ne i}(\mathbf{u};E_s^B,\Pi)\Big]
\end{align}
where $E_s=E_s^A\otimes E_s^B$ and $E_s^{A,B}$ are the Pauli substrings that only have support on physical legs of the atomic codes that are mapped to type $A$ or $B$ tensor enumerators respectively. We take $\wedge_{J,J_j}$ to be tracing over the appropriate legs required by the tensor network. This now enumerates the weights of $E_s\mathcal{N}_i$. We can repeat this a number of times for different $E_s=\bar{L}_iP_s$s, and the resulting enumerators would provide the requisite error probabilities $p(P_s\bar{L}_i|\sigma(P_s))$. 

For example, under symmetric depolarizing noise with probability $p$, 

\begin{equation}
    p(P_s\bar{L}_i|\sigma(P_s)) =\\ M(z=p,w=1-3p;P_s\bar{L}_i,\Pi).
\end{equation}
For other error models, we again select the appropriate parameters for the abstract n-tuple $\mathbf{k}$ and weight function. See Sec.~\ref{subsec:errdet} and App.~\ref{app:scalarenum}.

A decoder can then choose an operator with the highest probability then correct the error by acting $E\bar{L}_i$ on the system. In the case where no other logical qubits are present, this reduces to the maximum likelihood decoder.

It is easy to generalize this such that $\bar{L}_i$ can include multiple qubits logical qubits in some set $\kappa$ such that 
\begin{align}
    &M(\mathbf{u};E_s,\Pi)=\\\nonumber
    &\wedge_{J_i,J_j} \bigotimes_{i\in \kappa}  A^{J_i}_{i}(\mathbf{u};E_s^A,\Pi) \bigotimes_j B^{J_j}_{\not\in \kappa}(\mathbf{u};E_s^B,\Pi).
\end{align}
However, in general, if we include $|\kappa|$ qudits in the mixed enumerator such that we only integrate out $j\not \in \kappa$, then we need to check $q^{2|\kappa|}$ terms to find the error operator with the highest marginal probability.

\subsection{Logical Error Rates}
\subsubsection{Exact Computations} 
We have seen previously how one can compute the trivial syndrome enumerators, which yield the uncorrectable error probability. We can interpret the value $p_D=1-A/B$ as a logical error rate for a decoder with perfect syndrome measurements such that one discards the state whenever a non-trivial syndrome is measured. For such processes, one can define an error detection threshold $p_{th}$ such that $p_D$ is suppressed as a function of $d$ for error rates below the threshold. One such example is shown in Fig.~\ref{fig:threshold} for the surface code and 2d color code.

\begin{remark}\label{rmk:threshold}
If a class of quantum codes has an error detection threshold under i.i.d. depolarizing error, then the threshold is $p_{\rm th}=1/6$.
\end{remark}

\begin{proof}
    Let $A^*(z), B^*(z)$ be the enumerators with normalization such that $A_0^*, B_0^*=1$. Then for a quantum code with dimension $K$, $A(z) = K^2 A^*(z)$ and $B(z) = K B^*(z)$. Thus
\be\frac{B^*(z) - A^*(z)}{B^*(z)} = 1 - \frac 1 K \frac{A(z)}{B(z)}.\ee
Now, homogenizing, the MacWilliams identity has $B(w,z) = A((w+3z)/2, (w-z)/2)$, and using $z = p = 1/6$ and $w = (1-3p) = 1/2$, we see $A(1/2, 1/6) = B(1/2, 1/6)$ for every quantum code. Therefore all curves $p_L(p,1-3p)$ cross at $(1/6,1-1/K)$. 
\end{proof}

 We may similarly ask whether the current tensor network method can efficiently compute the exact logical error rate under other decoders. We do not provide such a method in this work, though it may be an interesting direction. A simple application of the current method fails to be efficient in the following examples. 
 
The exact logical error rate with maximum likelihood decoder can be expressed as 

\begin{equation}
    p_L = \sum_s [B^s(\mathbf{u};P_s,\Pi)-\max_{\forall\bar{L}\in \mathcal{L}}\{A^s(\mathbf{u};P_s\bar{L},\Pi)\}]
\end{equation}
and the error rate for a Bayesian decoder is 

\begin{equation}
    p_L =1- \sum_s \frac{1}{B^s(\mathbf{u};P_s,\Pi)}\sum_{\bar{L}}  A^s(\mathbf{u};P_s\bar{L},\Pi)^2.
\end{equation}

We see that both of them involve non-linear functions of the weight enumerators, which makes it difficult to compute efficiently through a tensor network method. It would appear that one has to sum over exponentially many syndromes even if each enumerator can be produced efficiently. 

This does not mean that enumerators cannot improve the computation of logical error rates. In practical decoding, it is far more relevant to consider a sample with only polynomially many distinct syndromes after running the decoder for a reasonable amount of time. It is also the case for all sampling-based simulations that are currently used for error and threshold computations.

\subsubsection{Error rate estimation} 
In addition to computing exact error probabilities given syndrome $s$, one can also use enumerators to provide more accurate estimates for logical error rates in conjunction with sampling-based methods. 

Conventional sampling methods generate errors $E$ based on particular noise models. Once the error is generated, its associated syndromes $\sigma(E)$ are determined. Note that for noiseless syndrome measurements, $\sigma(E)$ always outputs a syndrome $s$ deterministically. However, for more realistic models with faulty measurements, the outcome $s$ can depend also on the noisy measurement process. A decoder $\mathcal{D}(\sigma(E))$ then takes the syndrome and suggests a recovery operator $R$ with probability $p_{\mathcal{D}}(R|s), \sum_R p(R|s)=1$. If $RE\sim \bar{L}$ is equivalent to a non-identity logical operator, then a logical error has occured and this adds to the error probability $p_L$. This process is repeated until a large enough sample size has been established such that the overall $p_L$ estimate is believed to have sufficiently converged.

We can improve up this method, especially those derived from rare events/syndromes with enumerators. 
Given an error model (e.g. depolarizing noise with fixed error probability $p$) a set of errors are generated using existing sampling methods. Subsequent syndrome measurements (either noiseless or noisy) lead to a sampled syndrome distribution $\mathcal{P}(s)$ such that $\sum_s \mathcal{P}(s)=1$ and only has support over polynomially many distinct syndromes. 
In our case, we assume that we are given the distribution $\mathcal{P}(s)$, the error correcting code (along with its tensor network construction), the error model in question, and a decoder $\mathcal{D}$ of the user's choice. 

The logical error rate estimates are thus given by
\begin{equation}\label{eqn:rateest}
    \bar{p}_L(\mathbf{k}) = \sum_s \mathcal{P}(s) \sum_R p_{\mathcal{D}}(R|s) \Big( 1-\frac{A_s(\mathbf{k};R,\Pi)}{B_s(\mathbf{k};\Pi,\Pi_s)}\Big)
\end{equation}
where $A_s(R)/B_s$ is precisely the expected probability where the decoder's choice of $R$ successfully corrects the error based on the syndrome. For a maximum likelihood decoder, $p_{\mathcal{D}}(R|s)$ is also trivial except for one $R$. 
For a pure sampling based method, the probability $A_s(R)/B_s$ would usually require a large number of events before the estimate converges to its true value. Therefore, its estimate for rare syndromes can be wildly inaccurate.
Here with the enumerator method, we can compute these quantities exactly, thereby improving the accuracy for $\bar{p}_L$. It is also useful sometimes to further sort the logical error rate by operator types. This can be done by excluding certain terms in (\ref{eqn:rateest}) from the summation over $R$. We do not provide its explicit forms here as the extension is somewhat trivial and situation-dependent. 

In scenarios where the computation cost of enumerators are relatively expensive, one can complement, for instance, the Monte Carlo method, where only error rates associated with rare syndromes are computed using weight enumerators.

\textit{Faulty measurements:} With logical error probabilities in hand, we can compute error thresholds in the usual way by repeating such calculations or estimations for a class of codes with different distances. Note that the use of enumerators above is compatible with any error model composed of identical single qubit error channels. The computation also fully accommodates different models of noisy syndrome measurements, as they only affect the distribution $\mathcal{P}(s)$. Furthermore, the impact of each decoder can be independently evaluated to produce the conditional probability $p(R|s)$. We hasten to point out that the choice of decoder here is completely arbitrary and not limited to the decoders we constructed in Sec~\ref{subsec:decoder} based on weight enumerators. 

Since the contributions from the error channel, noisy measurements, decoders, and enumerators can be separated into independent modules, one can prepare them separately. For example, one can prepare a syndrome distribution $\mathcal{P}_0(s)$ with noiseless measurements. If the measurements are noisy, they are given by some set of transition probabilities $p(s_f|s_i)$ which depend solely on the noise model associated with the measurement. Composing these probabilities we get \be\mathcal{P}(s)=\sum_{s_i}\mathcal{P}_0(s_i)p(s|s_i).\ee Once the set of relevant syndromes have been established, which we take to be $poly(n-k)$, we create the decoding table from which $p_{\mathcal{D}}(R|s)$ can be obtained. At the same time, the enumerators that depend on $s$ and $R$ may be prepared in parallel, if needed. In many cases, exact contractions may not be needed as we may not require the same level of accuracy for distance verification. In such cases, approximate but efficient contraction algorithms maybe sufficient. 

{\color{black}
\section{Tensor networks for Codes}
As our primary tool for speed up comes from the QL description of the code, to make use of these methods, one also needs a modular construction for the codes. For codes created from QL, this is automatically true. Here we also provide a general method for decomposing known codes and states into smaller ``quantum lego blocks''.
We give two approaches for performing this decomposition based on how the code or state is prepared. As the enumerator for graph states\cite{SLD} is of interest, we also provide an explicit QL decomposition of all graph states. 
\subsection{Quantum codes from quantum Tanner graph}
For any stabilizer code with Abelian and non-Abelian stabilizer group (such as XP stabilizer codes\cite{XS,XP,shen2023quantum}), the codewords can be defined by the simultaneous $+1$ eigenspace of the stabilizer elements. 

Any such code can be rewritten with a QL decomposition where the tensor network is isomorphic to the Tanner graph of the code (Fig.~\ref{fig:qtanner}a). 
\begin{figure}
    \centering
    \includegraphics[width=\linewidth]{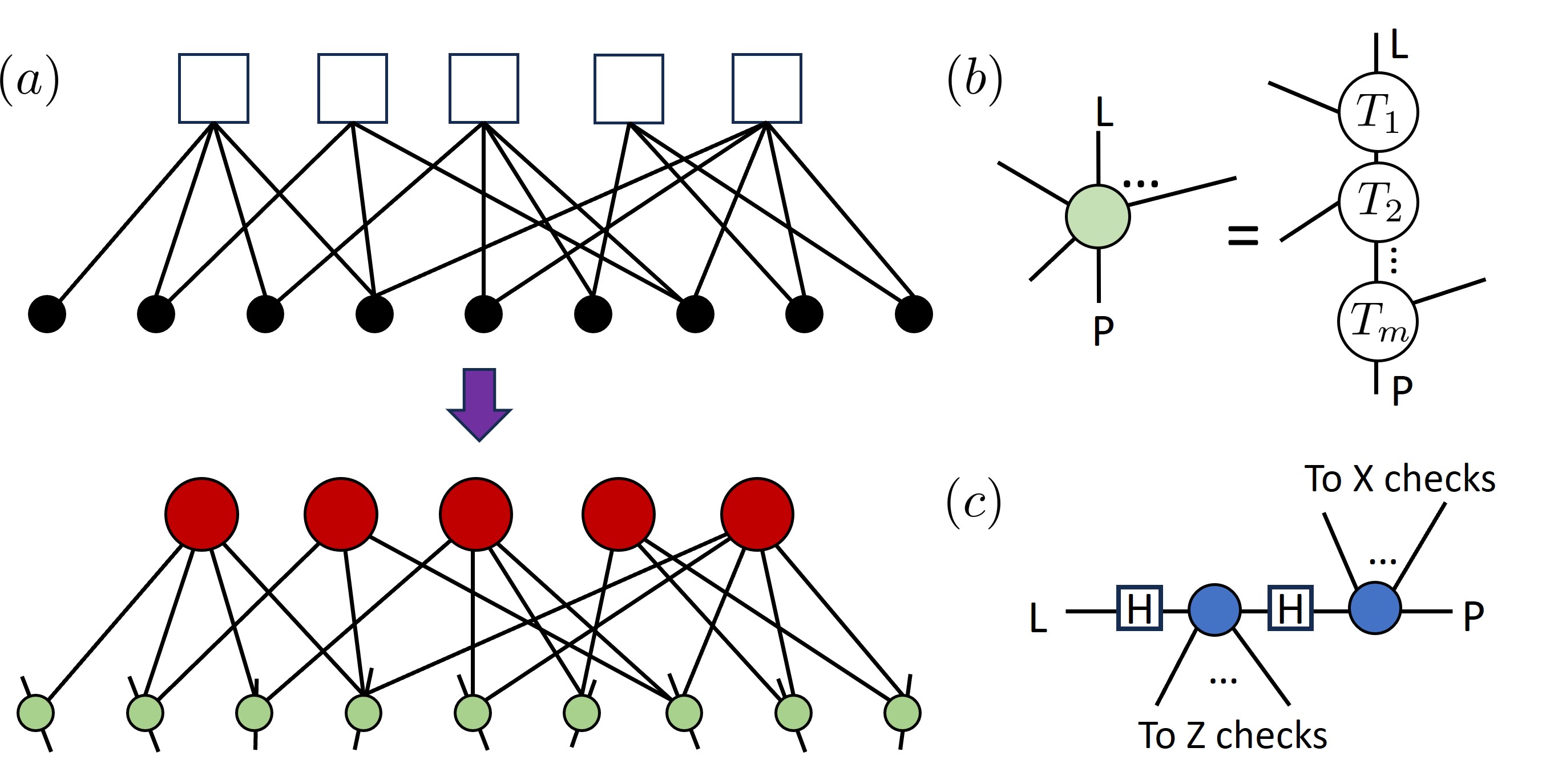}
    \caption{(a) A quantum code written as a Tanner graph where qubits are nodes and checks are squares. A tensor network isomorphic to this graph can be constructed from $Z$-spiders (red) and local QL codes (green). (b) Each local code has two dangling legs where $L$ marks the logical leg and $P$ marks the physical leg. The remaining legs are contracted with the check node tensors. (c) For the special case of a CSS code, the local code simplifies to the contraction of Hadamard tensors and X-spiders (blue).}
    \label{fig:qtanner}
\end{figure}
For each check node connected to $\ell$ qubits/qudits, we place a degree GHZ tensor which is the encoding tensor of a repetition code. The GHZ tensor is also known as a Z-spider in ZX-calculus \cite{Coecke_2011}. A similar repetition code with $H$ applied to each leg is known as an X-spider. For each physical node checked by $m$ checks with operators $\{g_i, i=1,\dots,m\}$ acting on the physical node, we place a tensor (green) with degree $m+2$ shown in Fig.~\ref{fig:qtanner}b where each of the $m$ contracted legs is connected to the corresponding check node tensor the qubit is checked by. The remaining two dangling legs represent the logical and physical degrees of freedom associated with the atomic code at each physical node. The specific $T_i$ tensor of degree 3 is completely determined by the type of operator $g_i$ that is present in the stabilizer check. If the code is a Calderbank-Shor-Steane (CSS) code, then the local tensor takes the simple form of Fig.~\ref{fig:qtanner}c, where it is a contraction between X-spiders (blue), which are repetition codes in a different basis, and Hadamard tensors. With this construction, one can easily build up tensor networks for existing codes with known stabilizer checks. When applied to topological codes, the quantum Tanner graphs are structurally similar to existing constructions with similar bond dimension, e.g. \cite{Guetal} for the toric code, which is CSS, and for the twisted quantum double model\cite{wenlevin,Vidal}, which has a non-abelian stabilizer group.  Details of this construction are given in Appendix~\ref{app:qtanner}.
The number of seed codes from the quantum Tanner graph are essentially optimal where we have $2n-k$ atomic legos for an $[[n,k]]$ code, identical to the usual Tanner graph description. If the code is a Low-Density Parity-Check (LDPC) code, then each tensor node also has bounded degree.

\subsection{Circuit-based tensor network}
%\CC{text below also slightly changed}
The quantum Tanner graph tensor network covers the vast majority of the existing quantum codes. For codes and states with an encoding circuit, then one can also convert the circuit into a tensor network\footnote{In the context of stabilizer codes, its Clifford encoding circuit is also easily obtainable \cite{Aaronson}.} as it is simply the contraction of tensors of unitary gates and product $|0\rangle$ states. For Clifford gates, the states dual to these tensors are stabilizer states. For instance, concatenated codes can naturally yield a log-depth tree tensor network. In general, the connectivity of a subregion of the network can scale linearly as the number of gates/tensors inside the region. For 1d (spatially) local circuit, it is in principle possible to cut through the network in the time direction. The edge cuts are upper bounded by the circuit depth $T$, and hence each contraction is no costlier than $O(\exp(T))$. For log-depth, this clearly yields exponential speed up. For $d$ spatial dimensions, the number of edge cuts is upper bounded by the surface area of the spacetime region $\ell^{d-1}T$ where $\ell^d \lesssim n$. Therefore the upper bound for the cost for each contraction is $O(\exp(n^{1-1/d}T))$. This can still lead to a sub-exponential speedup as long as $T\lesssim n^{1/d-\epsilon}$ asymptotically. %Very generally, we do not expect the tensor networks constructed from circuits to be optimal, i.e., minimizes the edge cuts for some subregion. Therefore it is still of interest to identify an efficient recipe for building such networks for stabilizer codes, potentially in conjunction with compilation tools such as ZX calculus.

%\subsection{Example: QL decomposition of graph states}
A simple procedure of such a circuit-to-tensor-network conversion is constructed in Fig.~\ref{fig:graph_state} for all graph states. Each CZ gate can be converted into the fusion of 3 tensors which can then be individually simplified through local recombinations (Fig.~\ref{fig:graph_state}a). This produces a tensor network with the same graph connectivity as the graph state with GHZ tensors (Z spiders) on the nodes and Hadamard tensors inserted on the edges (Fig.~\ref{fig:graph_state}b).

\begin{figure}
    \centering
    \begin{subfigure}[t]{0.45\textwidth}
        \centering
        \includegraphics[width=\linewidth]{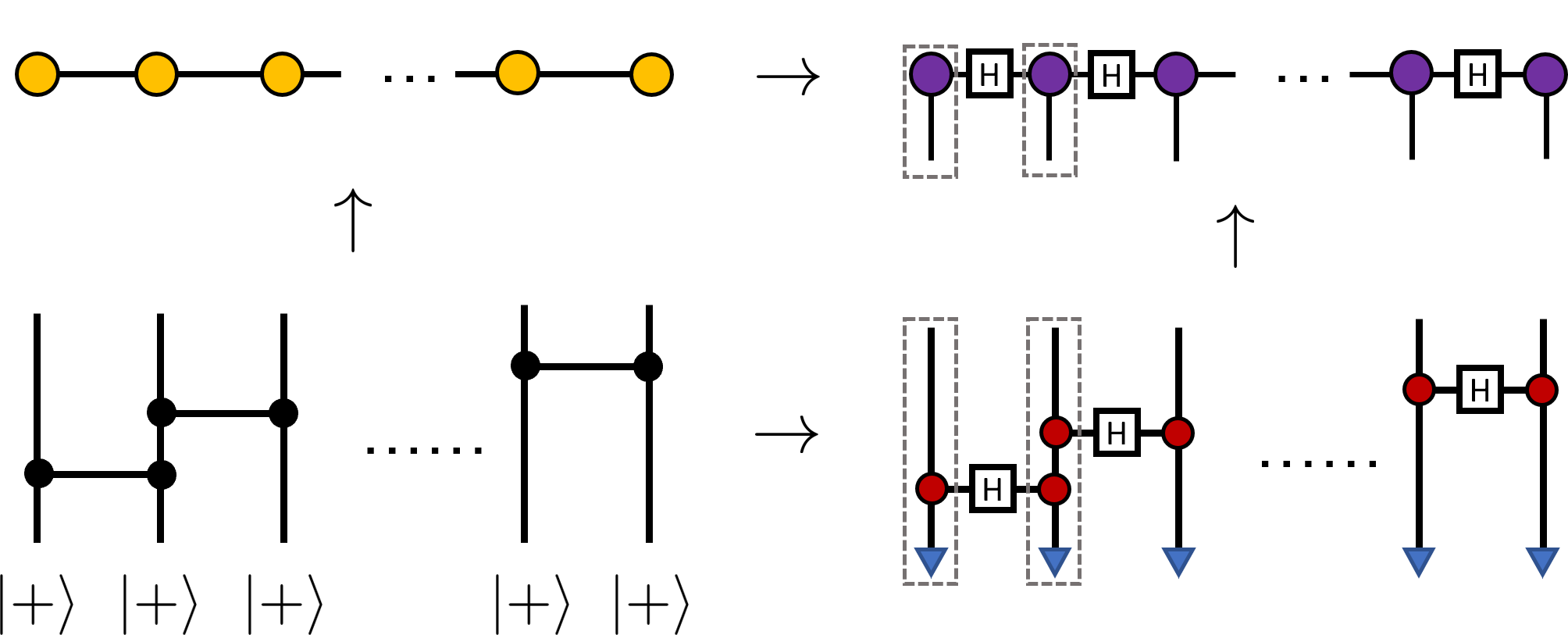} 
        \caption{Example: a 1d cluster state tensor network from encoding circuit. } \label{fig:1dcluster}
    \end{subfigure}
    \hfill
    \begin{subfigure}[t]{0.45\textwidth}
    \centering
    \includegraphics[width=\linewidth]{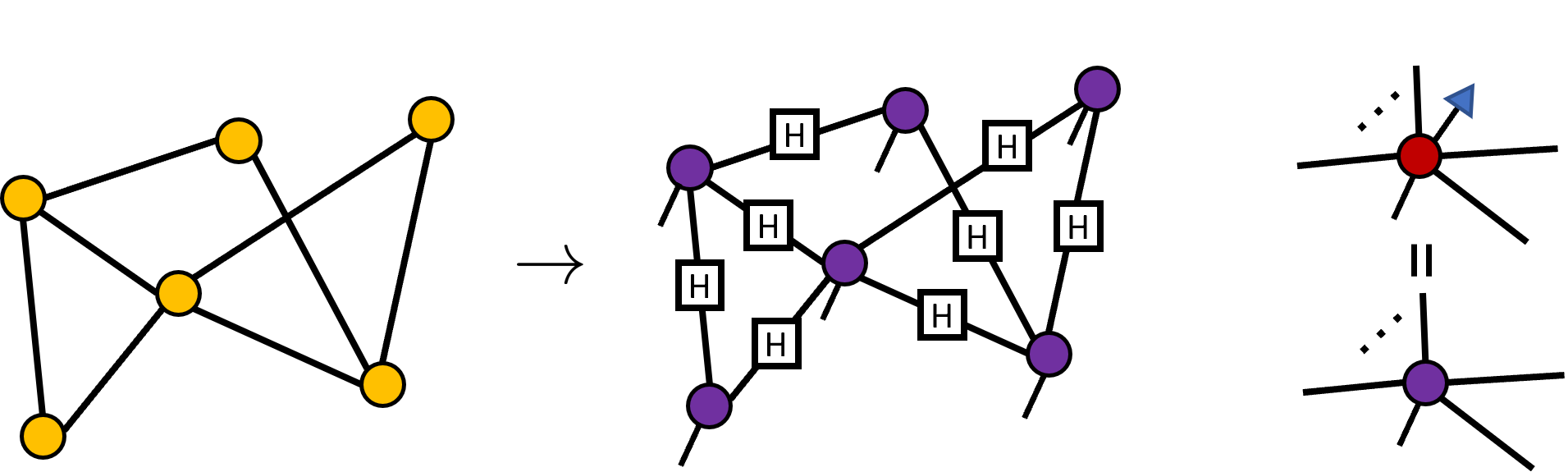}
    \caption{Any graph state can be converted to such a tensor network using the procedure above. A purple tensor is the GHZ tensor contracted with a $|+\rangle$ tensor. }
    \end{subfigure}
    \caption{Any graph state (yellow) can be converted to a tensor network using its encoding circuit constructed from $CZ$s acting on $|+\rangle$s. The tensor network consists of the GHZ tensors (red), Hadamard tensors (H) and contraction with $|+\rangle$ (blue triangles). Note that multiple GHZ tensors, which are also Z-spiders, can be merged to create a larger Z-spider.}
    \label{fig:graph_state}
\end{figure}

%The specific application of these two constructions can be case-dependent. %A high-depth unitary circuit description of a code can produce a high width tensor network that is difficult to contract. 
It is important to note that given the level of generality of our method, not all tensor networks will be exactly contractible in polynomial time as we see in the next section. This is because finding enumerators is NP-hard and exponential time for general tensor network contraction is unavoidable.} %However, even for topological codes where exact contractions are subexponential, the tensor networks obtained from such quantum tanner graphs are similar both in structure and bond dimension to popular constructions by \cite{}, we therefore expect competitive performance. 

\section{Computational Complexity}\label{sec:complexity}

\subsection{General Comments}

\subsubsection{Brute Force Method} 
For a generic stabilizer code, the construction of its weight enumerator polynomial is at least NP-hard. We thus expect the same for a generic quantum code. Indeed, as we see that constructing enumerators solve the optimal decoding problem \cite{IyverPoulin}, such tasks must be at least $\#$P-complete. A simple brute force algorithm is exponential in the system size. For stabilizer codes, one can enumerate all of its stabilizer or normalizer elements, which is of $O(q^{n-k})$ and $O(q^{n+k})$ respectively. This extracts the relevant coefficients $A_d,B_d$. For a general quantum code, each coefficient $A_d,B_d$ is already hard, as it involves $q^n\times q^n$ matrix multiplications. One then has to repeat this $O(q^{2n})$ times for each error basis element. Therefore the complexity for the brute force method is $O(q^{O(n)})$ for general quantum codes of local dimension $q$. A slightly better strategy computes only the coefficients of $A_d$ and then perform a MacWilliams transform, which is polynomial in $n$. Therefore, for complexity estimates, it is sufficient that we provide the estimate for computing $A(\mathbf{u})$.

\subsubsection{Tensor network method} 
Now we analyze how our method improves this picture assuming the QL constructions are known. 

\paragraph{Tensor preparation overhead.}
Let us first revisit the encoding tensor network of an $[[n,k]]$ stabilizer code with local dimension $q$ where each tensor is obtained from a small stabilizer code. We assume that the degree of each tensor (including dangling legs) is bounded by some constant $c$. This is to ensure that the complexity in constructing the tensor enumerator of each node is upper bounded by a constant overhead\footnote{For stabilizer codes, if there are $k_v$ logical legs on a tensor on a node $v$, then building $\mathbf{A}_v(z)$ is upper bounded by complexity $O(q^{c-2k_v})$ and is less expensive compared to that of $\mathbf{B}_v(z)$.For general quantum codes where one uses the full tensor enumerator, preparing the coefficients of $\mathbf{A}_v(z)$ requires a worst case of $O(q^{(5c-4k_v)})$ operations.}. Then consider the graph $G=(V,E)$ produced from the tensor network by removing all dangling legs such that the tensors are vertices and contracted legs are edges.
Suppose the tensor network representation is one such that $|V|\leq C(n+k)$ for some constant $C$, then preparation of the atomic blocks has worst case complexity $O((n+k)q^{5c})$.
 In fact, many tensor networks consist of only a few types of tensors, e.g. recall that any QL structure is constructible from a constant number distinct blocks, making even $O(q^{5c})$ sufficient.  Therefore the overhead for tensor preparation is usually constant while a generous upper bound is at most linear in the system size. Here we assume that the tensor network does not contain an overwhelming number of tensors that have no dangling legs, e.g. a deep quantum circuit. This assumption can always be satisfied (e.g. MPS). 

\paragraph{Tensor Contraction.}
We now contract these tensors to build up the tensor network. Recall that each tensor contraction may be construed as a matrix multiplication. Suppose we have two tensors of $p\leq m$ legs respectively, while we contract $n\leq p$ legs. For the most general quantum code, we need to use the full tensor enumerators as building blocks, which have bond dimension $\chi=q^4$ and can be reshaped as a multiplication of two matrices of size $\chi^{(p-n)}\times \chi^{n}$ and $\chi^{n}\times \chi^{(m-n)}$.  Hence each contraction step with the same parameters above has worst case $O(\chi^{(p+m-n)})$. For codes that only needed reduced enumerators, this can be done with $\chi=q^2$.  
 For stabilizer codes, these matrices are especially sparse and have at most $q^p, q^m$ nonzero elements, thus each such contraction is loosely upper bounded by $O(q^{p+m+\min(p,m)})$. 
 Therefore, the computational complexity scales exponentially with the number of uncontracted legs during tensor contraction. 
 
To incorporate the symbolic functions, additional degrees of freedoms are often needed. The specifics can depend on the implementation. One method is to introduce a separate index with bond dimension $(n+1)^{\ell}$ to track the power of the polynomial (App.~\ref{app:tensoronly}). This adds another factor of $n^{\ell}$ to the complexity counting above. The power $\ell$ depends on the number of independent variables one needs to track. For Shor-Laflamme enumerators $\ell=1$, but $\ell>1$ for the refined enumerators. As this cost can vary depending on the treatment of symbolic objects, we do not include their contributions in the following estimates. One can easily restore them when needed.

\paragraph{Fully contracted tensor network.} 
Aside from minor corrections related to symbolic manipulations and those associated with storing and manipulating for large numbers, the computational complexity would be determined by the contractibility of the tensor network, which is ultimately dominated by the cost of multiplying large matrices. Heuristically, the cost of tensor contraction scales linearly with the bond dimension of the uncontracted indices, or exponentially with the number of minimal edge cuts in the tensor network. 

In the ensuing the discussion we will use base $e$ exponential for complexity. For a tensor network with bond dimension $\chi$, we can generally set $e\rightarrow \chi$ to obtain the worst case complexity estimate. As we discussed earlier, the general rule of thumb for bond dimension is $\chi=q^4$ for the full tensor enumerator, $\chi=q^2$ for codes that only requires reduced enumerators. However, using a sparsity argument in stabilizer codes, the effective bond dimension needed in an efficient representation can even be as low as $q$.

Let us represent a sequence $\mathcal{S}_G$ of tensor contractions by a sequence of induced subgraphs $H_i=(V^H_i, E^H_i)$ where $V^H_i\subset V$, $V^H_{i+1}=V^H_i\cup \{v_{i+1}\in V\setminus V^H_{i}\}$, and $V^H_0=\{v_0 : v_0\in V\}$.   In other words, we construct a sequence of subgraphs by adding one additional vertex at a time. The sequence terminates at $i=|V|-1$, when the subgraph contains $G$. 
Let $E_c(W,W')= \{e=\{v_a,v_b\}\in E: v_a\in W, v_b\in W'\}$ denote the set of edges connecting any two sets of vertices $W,W'$ and let $M_{i+1}$ be the connected component of $H_{i+1}$ containing $v_{i+1}$.

Then the complexity for the $i$th step of contraction is
\begin{widetext}
\begin{align}
    \mathcal{C}_i\lesssim \exp(|E_c(V_{H_i}\cap V_{M_{i+1}},V\setminus V_{H_i})|+deg(v_{i+1})-|E_c(V_{H_i}\cap V_{M_{i+1}},\{v_{i+1}\})|)\lesssim O(\exp(|C_{\rm max}|))
\end{align}
\end{widetext}
where $|C_{\rm max}|=\max_i|E_c(V_{H_i},V\setminus V_{H_i})|$ is the largest possible cut through the tensor network during contraction.
Then we see that the number of computations needed for calculating the final tensor enumerator of the tensor network is given by

\begin{equation}
    \mathcal{C}=\sum_{i=0}^{|V|-1} \mathcal{C}_i\lesssim O(|V|\exp(|C_{\rm max}|)).
\end{equation}

The upper bound is a pretty drastic overcounting especially if $H_i$ contains many disconnected components, as many do not enter the contraction. In other words, as long as each connected component of the induced subgraph has only $\log|V|$ connectivity with its complement throughout the sequence $\mathcal{S}_G$, then the complexity is polynomial in $|V|$.

\subsection{Cost for common codes}
\begin{table*}[t]
    \centering
    \begin{tabular}{|c|c|c|}
    \hline
         Network architecture & TN cost& code examples  \\
         \hline
         Tree  & $O(\log n)$ & concatenated (symmetric)\\
         %\hline
         Tree, 1d area law & $O(n)$& concatenated (general), convolutional \\
         %\hline
         2d hyperbolic & $O(n^{\alpha+1}), \alpha>0$& holographic, surface code (hyperbolic) \\
         %\hline
         (hyper)cubic& $O(n\exp(n^{1-1/D}))$ & topological (Euclidean), Bacon-Shor \\
         %\hline
         (hyper)cubic (bounded $L$)& $O(n\exp(L^{D-1}))$ & rectangular surface code\\
         %\hline
         $\delta$-volume law & $O(n\exp(\delta n)), \delta<1$ & non-degenerate code, random code\\
         %\hline
         generic encoding circuit & $O(n^2\exp(n)/\log n)$ & generic stabilizer code\\ 
         \hline
    \end{tabular}
    \caption{tabulates the computational cost for enumerator preparation from tensor network contractions. There are additional complexity associated with the symbolic manipulation of the polynomial, storage of large numbers, and MacWilliams transforms, which can also contribution an additional cost that can be superlinear.}
    \label{tab:costs}
\end{table*}

\paragraph{Tree tensor network.}
Tree tensor networks can be used to describe concatenated codes over $n$ qubits (leaves). It is also known that these tensor networks can be contracted with polynomial complexity. A contraction algorithm would start from the leaves of the tree  and contract into $O(n)$ disconnected components of the graph. 
Each piece in this first layer of contraction has at most $\mathcal{E}\sim O(c)$ open legs where $c$ is the maximum degree or branching factor in the tree. Then at each iteration, we join the $\leq c-1$ branches with another tensor. The maximum number of open legs on each connected component is always bounded by $c$, therefore the complexity for each contraction is at most $O(e^{2c})$. For a tree with $n$ leaves, the overall complexity is $O(n e^{2c})$ for tensors of bounded degree, Fig.~\ref{fig:ttn}. If the codes on each node are identical, then we only have to perform a separate contraction at each layer, yielding a complexity $O(\log n)$, Table~\ref{tab:costs} (general and symmetric). The latter would be doubly exponentially faster than brute force enumeration.
\begin{figure}
    \centering
    \includegraphics[width=\linewidth]{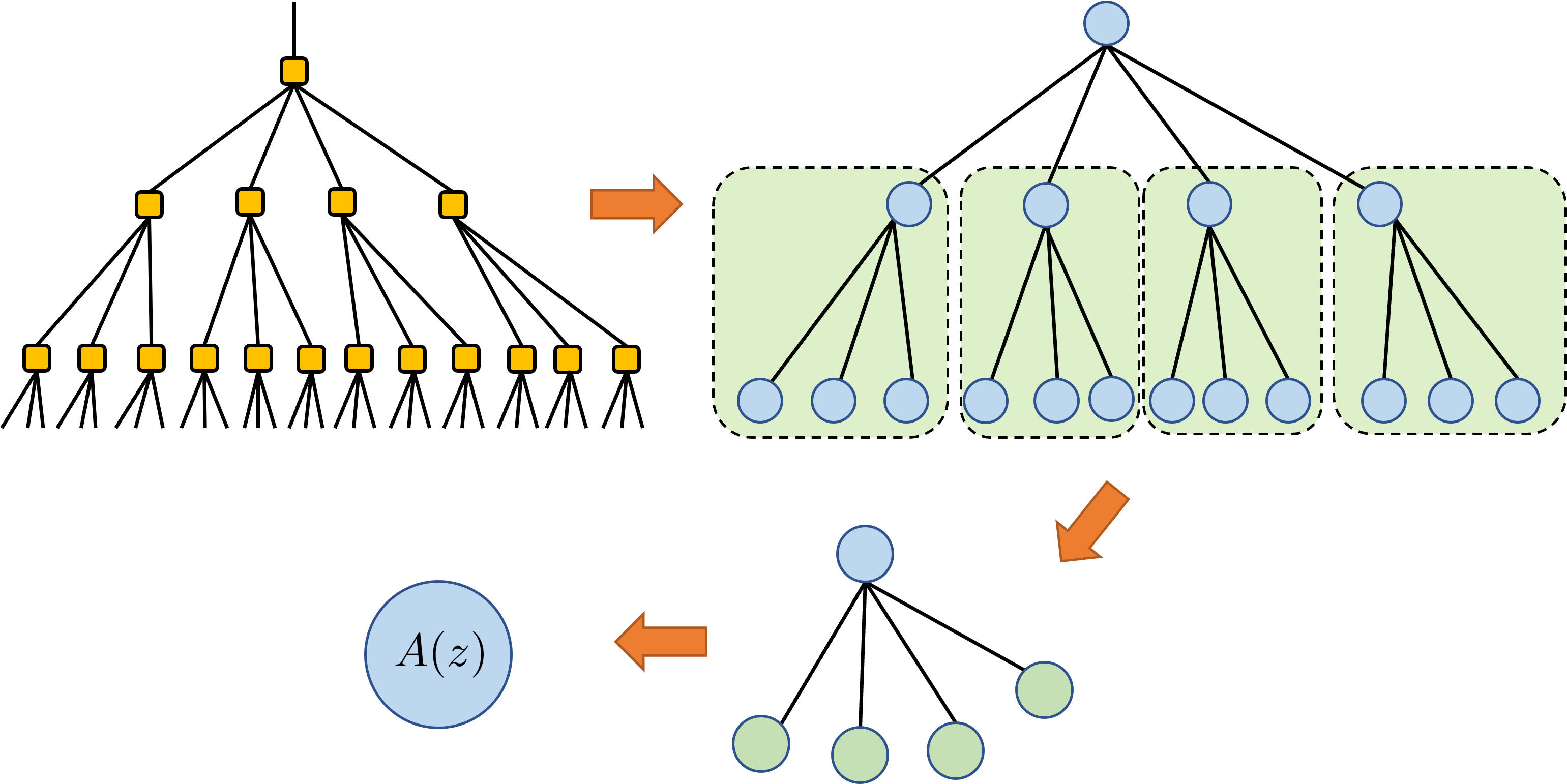}
    \caption{Tree tensor network for concatenated code. It is efficiently contractible from the bottom up and can be parallelized.}
    \label{fig:ttn}
\end{figure}

\paragraph{Holographic code.}
For tensor networks of holographic codes \cite{HaPPY,holosteane,ABSC,HTN}, the network is taken from a tessellation of the hyperbolic disk. This is slightly more connected than the tree tensor network (TTN) as it contains loops. The contraction strategy is similar to that of the TTN, except now minimum cuts depend on the system size such that each connected component has at most $O(\alpha \log n)$ open legs during the contraction. The parameter $\alpha$ depends on the tessellation. Then
\begin{align}
    \mathcal{C}&\sim \sum_{m=1}^n \exp(\alpha \log m+2c) \leq \exp({2c}) n^{\alpha+1} \\\nonumber
    &\sim O(n^{\alpha+1}).
\end{align}

A similar counting argument holds for the hyperbolic surface code, where minimal cuts remain logarithmic in the system size.

\paragraph{Codes with shallow local circuits.}
If the encoding circuit of a code is known (e.g. stabilizer code once the check matrices are given), then we can easily convert the circuit into a tensor network. If these circuits are shallow, say, of constant or $\log n$ depth, then one can contract the circuit induced tensor network in the space-like direction where the minimal number of edge cuts would be given by the circuit depth. Thus the enumerators of such codes can be prepared in $poly(n)$ time. 

\paragraph{Codes on a flat geometry.}
These are codes on an Euclidean geometry of dimension $D$ such as ones where the code words may be described by a PEPS. Some examples include the 2d color code, the surface code, Haah code \cite{haahcode}, etc. Constructions like the Bacon-Shor code also fall under this category. Note that the worst case complexity holds for any such tensor network regardless of the specific tensor construction or its symmetries.

For codes whose discrete geometry are embeddable in the $D$ dimensional Euclidean space, we simply ``foliate'' the lattice with co-dimension 1 objects. 
Each such object can be built up from $O(n^{1-1/D})$ contractions where each contraction retains at most $O(n^{1-1/D})$ open legs. Then $\mathcal{C}\sim O(n\exp(n^{1-1/D}))$. Compared to the brute force method, this permits a sub-exponential speed up. 

If the geometry of the network allows for fewer open edges during tensor contraction, then it is possible to get further speedups. Note the above counting assumes $n\sim L^D$ for a system that has similar lengths in different directions. If all but one direction have bounded length $L$ then we obtain an exponential speed up. For example, consider a rectangular surface code of size $L\times n/L$ on a long strip where $L$ is bounded, then each contraction along its shorter side is only $O(\exp(L))$. 

Note that the hardness of evaluating the weight enumerator polynomial here is directly tied to the hardness of the tensor network contraction. It was shown in \cite{pepshard} that contraction of PEPS is average case $\#P$-complete. Therefore there is strong reason to believe that an exponential speedup of this process is unlikely for both classical and quantum algorithmic approaches using tensor networks if one disallows post-selection and choose the tensors in a Gaussian random fashion. However, we also note that often the tensors are strictly derived from stabilizer codes. Therefore it is not impossible that these added structures in the discrete symmetry and contractible 2D tensor networks may permit further speedups. 

\paragraph{Codes with volume law entanglement.}
For states that have volume law entanglement for any subsystem, let us assume that the number of edges connected to vertices in a subregion is proportional to the number of vertices in that region, i.e. $\eta |V|$. For simplicity, let us also assume that the number of tensors and qubits are roughly equal. In general, $\eta$ need not be less than one. This is because each node may be connected to multiple nodes in the complementary region, while the entanglement captured in each bond is not maximal. However, if a carefully crafted tensor network is efficiently capturing the entanglement of the state, such that each bond is roughly maximally entangled, then we could expect the number of bonds cut to be less than or equal to the total number of qubits in the region for large enough subregions. Then the cost for each contraction is $O(\eta |V|)$. For $\eta<1$, this provides a polynomial speed up. If the number of bonds cut $\leq d$ for any subsystem and the code distance $d=\delta n, \delta <1$, which is the case for random codes, then the overall complexity would be

\begin{equation}
   \mathcal{C}\sim O(n\exp(\delta n)),
\end{equation}
which again admits a polynomial speed up.

However, if the number of bonds cut for a subsystem is $\geq n$, then we do not get any speed up. This would be case for all-to-all connected graphs where the edge cuts can be of size $(n/2)^2$, our algorithm at $O(\exp(n^2/4))$ will actually be slower than the brute force algorithm. For another example, consider the encoding circuit of any stabilizer code has $n^2/\log n$ complexity, which can be thought of as a tensor network. Suppose we simply contract the circuit tensor network timeslice by time slice, then we expect $|C_{\rm max}|\sim n$ because each time slice would correspond to a tensor network with $O(n)$ legs and the worst case complexity scales as $\sim O(n^2 \exp(n)/\log n)$. This is fully expected, as we should not be able to solve a $\#$P-complete problem in polynomial time. Therefore, in this regime, even if its tensor network description is optimal and minimizes the number of edge cuts for any subregion, the tensor network method would still only provide a polynomial speed up at best.

\subsection{Entanglement and Cost} In this work, we say that a tensor network representation is \emph{good} if its graph connectivity reflects the entanglement structure of the underlying state. In other words, the entanglement entropy of any subsystem can be reasonably well approximated by the number of edge cuts when bipartitioning the graph into the subsystem and its complement. This definition does not require the network to be efficiently contractible \cite{pepshard, GeEisert}. If we use the tensor network connectivity interchangeably with subsystem entanglement then we see that the complexity for computing the weight enumerator can be connected with the amount of entanglement present in the codewords. For more highly entangled codewords/states, its tensor network will be more connected, and hence the number of edge cuts for each subsystem will be higher. This provides us a heuristic where the general expectation of its weight enumerator computation should scale as $\sim \exp(S)$ where $S$ is roughly the maximum amount of entanglement for subsystems we generate during tensor tracing. We see that this is indeed the case for our examples --- the complexity is polynomial for codes whose code words are weakly entangled, i.e., $S\lesssim \log n$ and generally subexponential for states that satisfy an area law $S\sim n^{1-1/D}$ for systems with $D$-dimensional Euclidean geometry. 

For \emph{non-degenerate} quantum codes, the $(d-1)$-site subsystem are maximally mixed, hence $d\sim S$. Therefore, up to polynomial factor corrections, we expect the complexity lower bound for computing the enumerator polynomial to be comparable to that of finding the minimal distance in classical linear codes \cite{CRSS,distance2}, i.e.,
\begin{equation}
    \mathcal{C}\sim \exp(O(S))\sim\exp( O(d)).
\end{equation}
For this high level analysis, we will neglect other subleading terms and the dependence on rate $R=k/n$. Because stabilizer codes can be identified with classical linear codes over $GF(4)$ \cite{CRSS}, it means that the tensor network method should have comparable complexity scaling with existing algorithms for non-degenerate stabilizer codes.

In \emph{degenerate codes}, however, there exist subsystems where $S\ll d$. For example, a gauge fixed Bacon-Shor code can be constructed from a TTN (Sec.~\ref{subsec:bsc}). Although certain subsystems are highly entangled, its much weaker entanglement for some other subsystems allows one to engineer the network such that it is written in an efficiently contractible form, such that each step of the contraction is bounded by a constant. Depending on the gauge, we can get away with an enumerator with as few as $2\sqrt{n}$ such contractions. Although the code has overall distance $d\sim \sqrt{n}$, the cost in preparing its enumerator is only $O(\sqrt{n})$ time, compared to a naive distance scaling of $O(\exp(\sqrt{n}))$ (Figure~\ref{fig:xxbsc}). Therefore, we expect some degenerate codes to have $\mathcal{C}\ll\exp(O(d))$, which is a substantial speedup compared to known methods.

\begin{figure}[t]
    \centering
    \includegraphics[width=\linewidth]{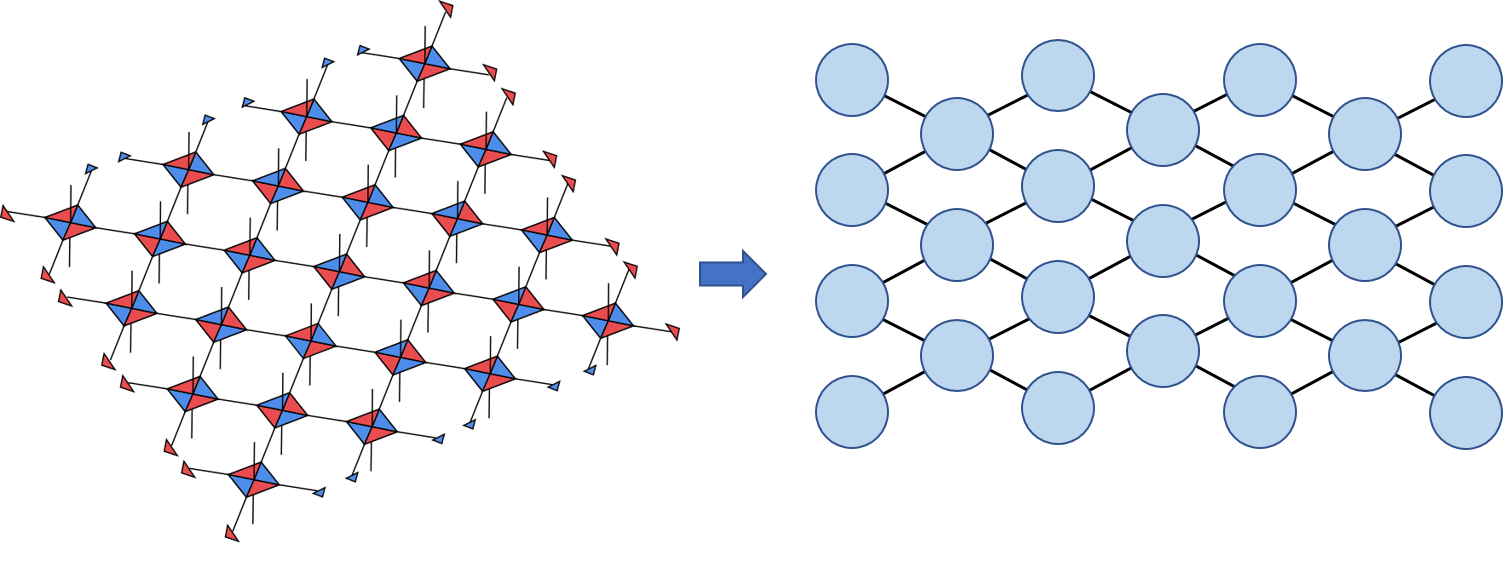}
    \caption{A surface code and the tensor network of its weight enumerator.}
    \label{fig:surface_code}
\end{figure}

\section{Examples}\label{sec:examples}
Now we examine a few examples by computing the enumerators for codes that have order a hundred qubits or so. These analyses are to showcase the tensor enumerator method; they are not meant to be exhaustive nor do they represent the largest possible codes one can study with this method.

\subsection{Surface code}\label{subsec:surfacecode}
\paragraph{Kitaev's Surface code.} 
Recall from \cite{Kitaev03} that the tensor network for the surface code encoding map, Fig.~\ref{fig:surface_code} (left), is one  where each tensor is a $[[5,1,2]]$ code and the boundaries are contracted with $|0\rangle, |+\rangle$ states (red and blue triangles). The upward pointing dangling legs denote the logical inputs and downward pointing legs denote physical qubits, therefore the encoding map has a non-trivial kernel and a physical qubit sits on each node. For each atomic block, we construct its tensor enumerator and contract them column by column to generate the entire network, Fig.~\ref{fig:surface_code} (right). For example, the quantum weight enumerators of a $[[181,1,10]]$ surface code are
\begin{align}
    A(z)&=1+36 z^3+180z^4+136z^5+1344z^6\\\nonumber
    &\quad +7084z^7+24001z^8+60432z^9\\\nonumber
    &\quad +286748z^{10}+\dots\\
    B(z)&=1+36 z^3+180z^4+136z^5+1344z^6\\\nonumber
    &\quad +7084z^7+24001z^8+60432z^9\\\nonumber
    &\quad +286768 z^{10}+\dots,
\end{align}
where we count only 20 representations of non-trivial logical operators at weight $10$.

Using a similar network, we can also find the coset weight distribution. Suppose a Pauli error acts on physical qubits in the form of Fig.~\ref{fig:surface_coset} (left). Note that we do not contract the Pauli errors into the encoding tensor network when defining the encoding map; if we actually contract the Pauli errors onto the physical legs in the tensor network construction then obtain enumerators from those networks, it would correspond to finding the stabilizer weight distribution for surface codes that have extra minus signs on certain generators. To build the coset enumerator, we swap out the original tensors in Fig.~\ref{fig:surface_code} (right) for the proper coset tensor weight enumerator of each error node (red). The modified tensor network then computes the weight distribution of coset elements. For example, the coset distribution for a single $X$ error at the bottom left corner for a $[[113,1,8]]$ surface code, is 
\begin{align}
    &A^{s_{Xbl}}(z) = z+z^2+2z^3+31z^4+146z^5+284z^6\\\nonumber
    &+1258z^7+5180z^8+17627z^9+\dots
\end{align}

\begin{figure}[b]
    \centering
    \includegraphics[width=\linewidth]{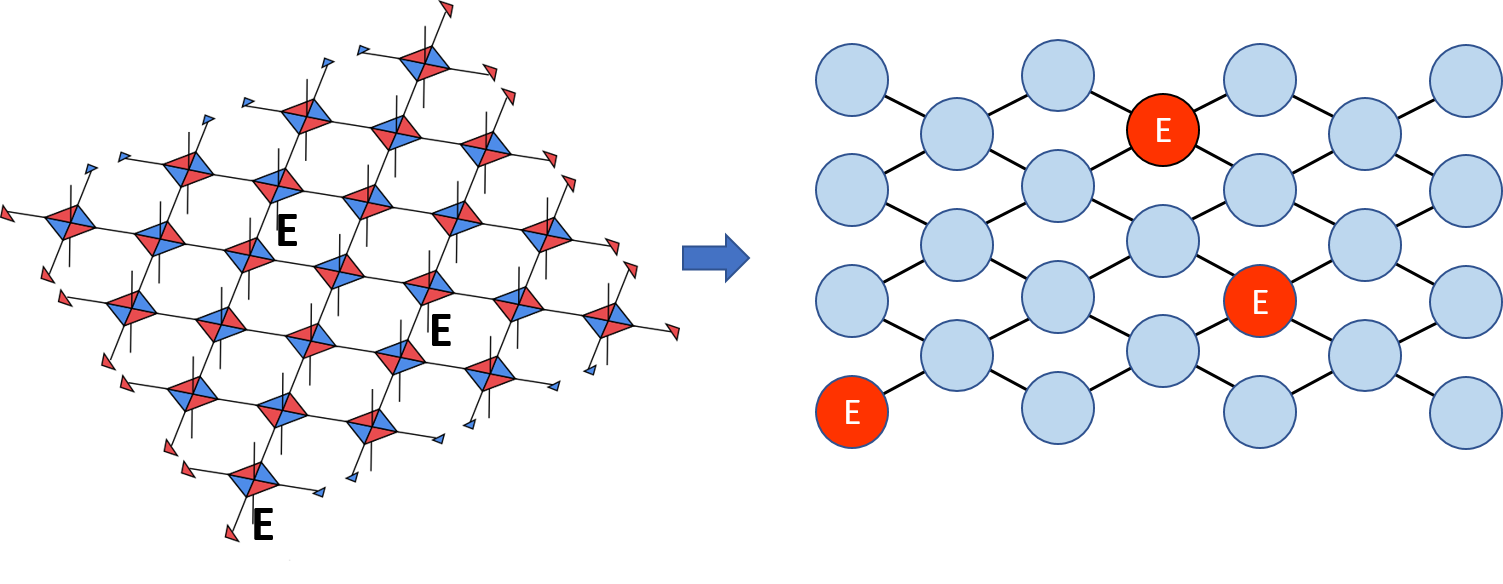}
    \caption{The coset enumerator of a particular error string that acts trivially on some qubits.}
    \label{fig:surface_coset}
\end{figure}

These exercises can be easily repeated for the double and complete weight enumerators where the weights are counted differently. For example, see Fig. 3 of \cite{CL2022} and Fig.~\ref{fig:surfaceCode_double}.

\begin{comment}
\begin{figure}
    \centering
    \includegraphics[width=0.7\linewidth]{3by150_surfaceCode_doubleEnum.png}
    \caption{ Double enumerator for the stabilizer weight distribution for a 3-by-150 surface code at $n=748$. plotting log of the weight distributions.}
    \label{fig:3by150}
\end{figure}
\end{comment}
%\CC{2d plot and 1d plot for surface code enum}

%\begin{figure}%
%    \centering
%    \subfloat{{\includegraphics[width=0.45\linewidth]{4by8surfacecode_jet.png} }}%
%    \qquad
%    \subfloat{{\includegraphics[width=0.45\linewidth]{3by150_surfaceCode_doubleEnum.png} }}%
%    \caption{Left:double enumerator of non-trivial logical weight for a 4 by 8 surface code at $n=53$. Right: double enumerator for the stabilizer weight distribution for a 3 by 150 surface code at $n=748$. Both plots are log of the weight distributions.}%
%    \label{fig:example}%
%\end{figure}

%\begin{figure}
%    \centering
 %   \includegraphics[width=0.7\linewidth]{surfacecode_scalarEnum.png}
  %  \caption{Stabilizer weight distribution for  surface codes at $n=181$ and $n=145$.}
  %  \label{fig:surfaceCode_scalar}
%\end{figure}

\begin{figure*}[]
    \centering
    \includegraphics[width=0.9\linewidth]{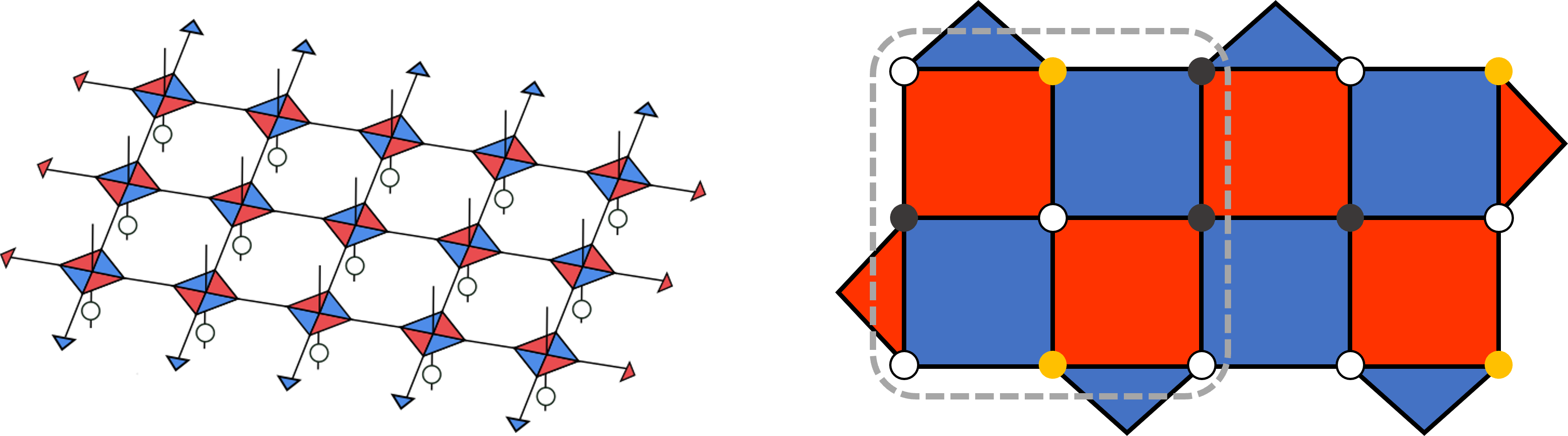}
    \caption{Tensor network of a rotated surface code where the atomic codes are identical to those of the surface code. Only the boundary conditions are modified. One can also modify each tensor by contracting some other single qubit gate/tensor. The checks are given on the right where qubits (vertices) adjacent to red regions indicate $Z$ checks and blue indicate $X$ checks. For the derandomized local Clifford deformed code \cite{arpit}, white and yellow dots indicate local $HSH$ and $H$ deformations respectively.}
    \label{fig:rot_surf}
\end{figure*}

\paragraph{Rotated Surface code.} 
In practice, it is easier to deal with rotated surface code as the distance scaling is better by a constant factor for a similar value $n$, Fig.~\ref{fig:rot_surf}. Note that one only has to modify the boundary conditions compared to the original surface code.  
The rotated surface code tensor network is also easier to contract exactly. For reference, the enumerator for the $[[256,1,16]]$ rotated surface code at $d=16$ can be computed on a laptop with a run time of $\approx 20$ minutes. The weight enumerators for this code are
\begin{widetext}
\begin{align}
    A(z)&=1+30z^2+776z^4+15538z^6+276801z^8+4431408z^{10}+65676619z^{12}\\\nonumber
    &+912021486z^{14}
    +12003931907z^{16}+150911390280z^{18}+\dots\\
    B(z)&=1+30z^2+776z^4+15538z^6+276801z^8+4431408z^{10}+65676619z^{12}\\\nonumber
    &+912021486z^{14}
    +12004980483z^{16}+150970896992z^{18}+\dots
\end{align}
\end{widetext}
Indeed, we see that the two coefficients start deviating at $d=16$.

One can also obtain an error detection threshold by assuming a decoder that performs no active error correction, but discards all instances that return a non-trivial syndrome assuming perfect measurements. Recall (Remark~\ref{rmk:threshold}) that this threshold is at $p=1/6\approx 16.67\%$, which is quite similar to the code-capacity thresholds \cite{eczoo_surface} across various decoders under depolarizing noise.

\begin{figure}[h]
    \centering
    \includegraphics[width=\linewidth]{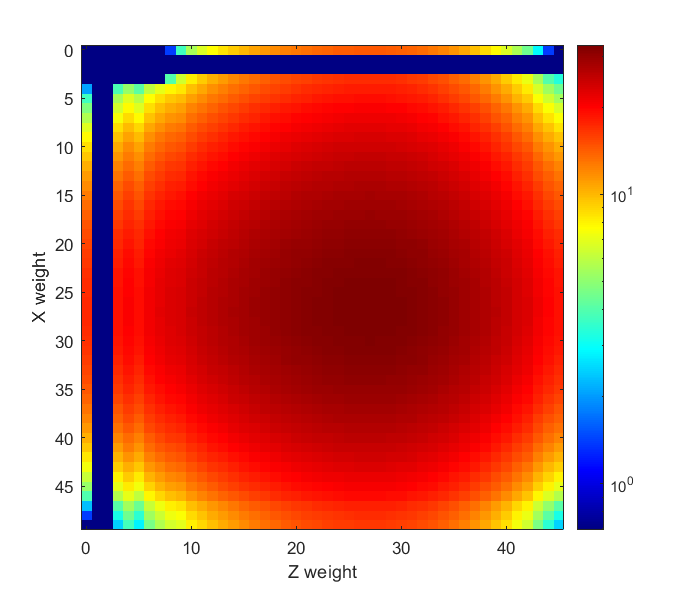}
    \caption{Double enumerator of a 4 by 8 surface code at $n=53$. Plotting log of operator weight distribution for non-trivial logical operators. A relatively small code is chosen for clarity in the figure. }
    \label{fig:surfaceCode_double}
\end{figure}

\paragraph{Local Clifford deformations.}
We can perform local modifications \cite{CL2021} on each tensor to perturb the (rotated) surface code. These are represented by the circle tensors that act on each qubit. For the vanilla surface code, these tensors are trivial (identity). However, we may choose them at will. For instance, if they are random single qubit Clifford operators, then the tensor network reproduces the Clifford-deformed surface codes \cite{Cliffdef}. Similarly, if choosing every other tensor to be a Hadamard, then one arrives at the XZZX code \cite{XZZX}.

Because the Shor-Laflamme enumerator is invariant under local unitary deformations, it is clear that the logical error probabilities of such locally deformed codes would be identical under unbiased noise. However, this local unitary invariance is broken when we consider more general enumerators with other weight functions, which indicate that their performances under biased noise differ. In Fig.~\ref{fig:asymdist}, we see that the de-randomized Clifford deformed code (right) has fewer logical operators that have low $Z$ weight, which is to be contrasted with the rotated surface code (left) and the XZZX code (middle). We use a de-randomized Clifford deformed code like the one shown in Fig.~\ref{fig:rot_surf} (right) where yellow and white dots indicate local $HSH$ and $H$ rotations \cite{arpit}. More general dimensions of the code follow from repeating the local patterns on the $3\times 3$ blocks (enclosed by dashed lines) periodically.

\begin{figure*}[t]
    \centering
    \includegraphics[width=\linewidth]{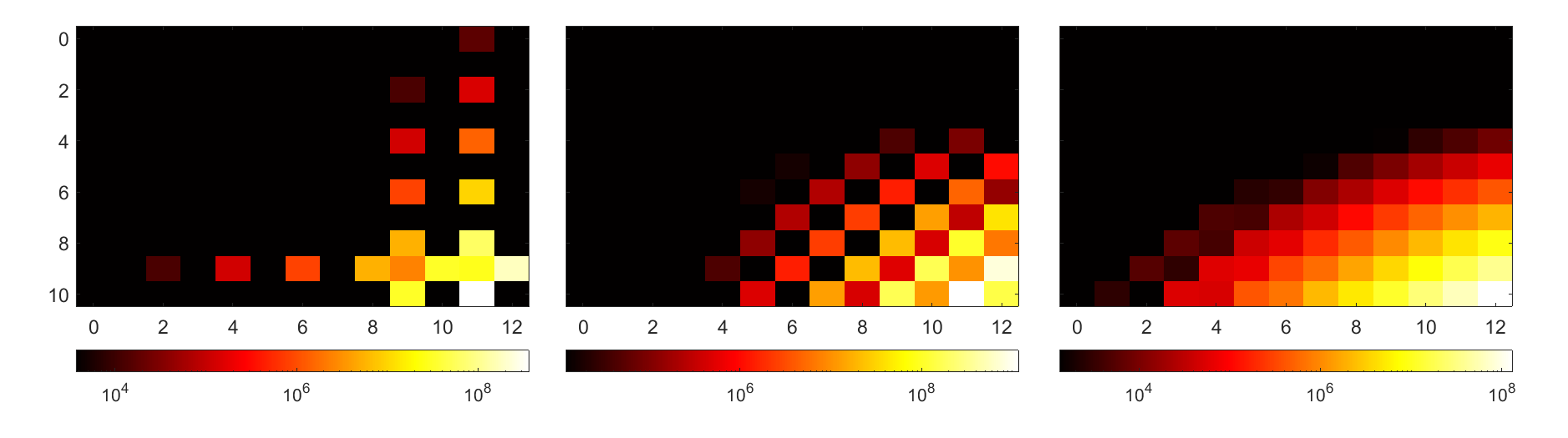}
    \caption{Truncated $X,Z$ weight distribution of non-trivial logical operators for the $9\times 9$ surface code (left), XZZX code (middle), and the Clifford deformed code (right). Horizontal axis: $X$-weight, vertical axis: $Z$-weight. Note that  non-zero weights are invisible in this scale.}
    \label{fig:asymdist}
\end{figure*}

For example, using the double enumerators, we contrast the performance of the XZZX code and the derandomized Clifford deformed code, Fig.~\ref{fig:biased_ratio}, under biased noise with $p=p_X+p_Y+p_Z$ and $p_X=p_Y=p_Z/(2\eta)$. It is clear from the normalized uncorrectable error rate (and hence effective distances) that the Clifford deformed construction vastly outperforms the XZZX at high bias and small $p$. Note that the weight function for these double enumerators is slightly different from the one used in App.~\ref{app:scalarenum} or \cite{hu2020weight} because it enumerates the $X,Y$ weight separately from the $Z$ weights.
\begin{figure}
    \centering
    \includegraphics[width=\linewidth]{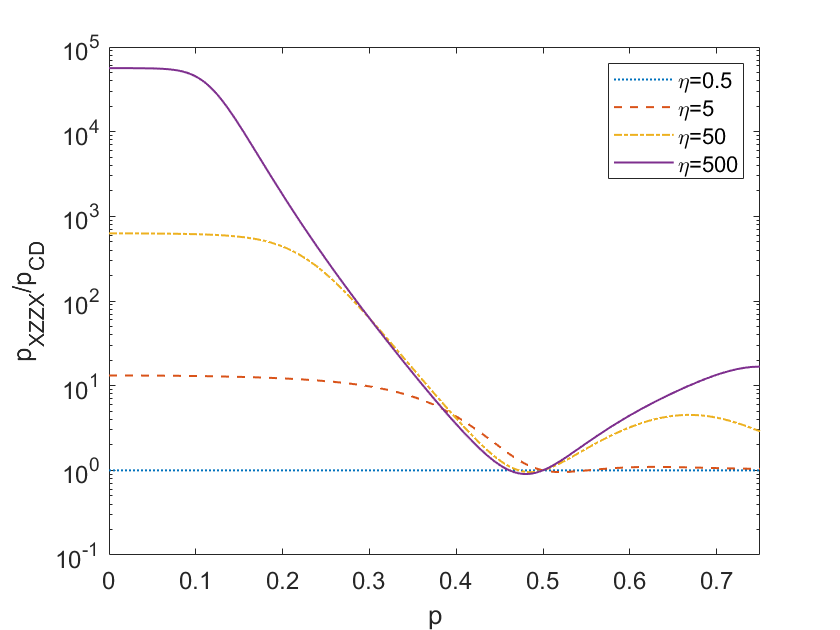}
    \caption{The ratio between the XZZX code normalized uncorrectable error rate $p_{XZZX}$ and that of the Clifford deformed code $p_{CD}$ as a function of physical error parameter $p$ at different biases $\eta$ for $d=7$. }
    \label{fig:biased_ratio}
\end{figure}

\paragraph{Coherent error.} 
General quantum errors are not limited to random Pauli noise, which are somewhat ``classical''. Here we compute the coherent error probability of the rotated surface code using techniques introduced earlier.

Efficient methods for computing unitary rotations along $X$ or $Z$ have been introduced by \cite{FFcoherent} using a Majorana fermion mapping. Here we instead consider i.i.d.~coherent error of the form $U=\exp(it Y) = \cos (t) I + i \sin(t) Y$. Note that the normalized logical error rate differs for codes with even or odd $X$ and $Z$ distances because the abundance of $Y$-only operators differ for these codes, Fig.~\ref{fig:coherent} (left).

When $d_x,d_z$ are odd, the normalized logical error rate under coherent noise with rotation angle $t$ coincides with that under the $Y$-only Pauli noise with probability $p_Y = \sin^2(t)$. This is because at odd distances, the only $Y$ type logical operator acts globally on the system. When at least one of $d_X$ or $d_Z$ is even, then the coherent noise yields slightly higher logical error probability, Fig.~\ref{fig:coherent} (right). However this only incurs a small correction with a similar order of magnitude, consistent with earlier results but in different settings \cite{FFcoherent}. A similar result holds for the XZZX code with coherent noise of $Y$-only rotations because up to a phase, $Y$ is invariant under Hadamard conjugation.

\begin{figure*}[]
    \centering
\includegraphics[width=0.8\linewidth]{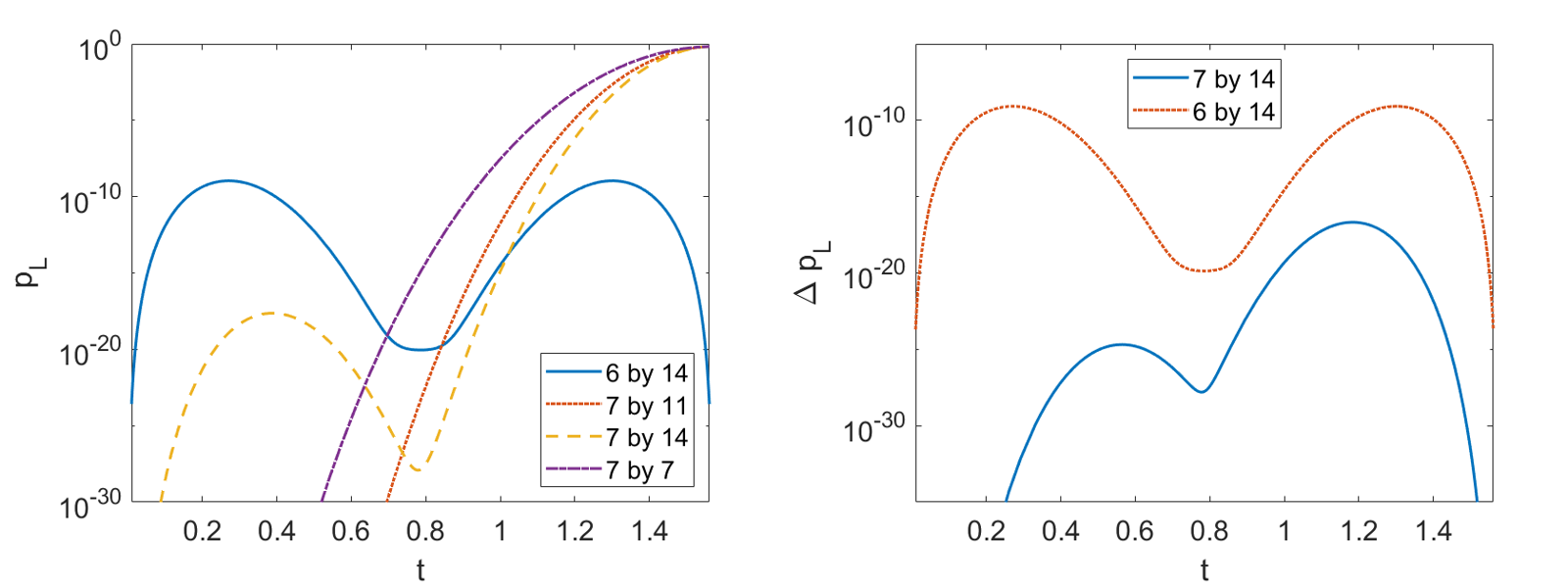}
    \caption{Left: normalized logical error rate as a function of the rotation angle $t$ for codes with size $n=d_x\times d_z$. Right: differences in normalized logical error rates $\Delta p_L = p_L^{\rm coherent} - p_L^{\rm Y~only}$. }
    \label{fig:coherent}
\end{figure*}

Although the impact of coherent noise with $Z$ or $X$ only rotations produce very different logical error profiles than those produced by the $Z$ or $X$-only Pauli noise in the rotated surface code, there exist XZZX codes where their impact are identical. For instance, for the system sizes tested, the effect of such coherent errors and Pauli errors coincide when we have a square lattice where the width is equal to height. It also holds for some rectangular lattices, though not all. The reason is similar as before, where there is a sole logical operator consists of only $I$ and $X$ (or $Z$), but the operator need not act globally. This may be due to special symmetries of the XZZX code, which indicates that local deformations can be tuned to reduce the impact of coherent noise. Though it is also likely that such symmetries are restricted to the $s=0$ sector. We leave a more systematic characterization of such behaviours to future work.

\subsection{2D color code}\label{subsec:colorcode}
We first provide a novel tensor network construction for the hexagonal 2d color code, which is a self-dual CSS code constucted entirely from Steane codes, Fig.~\ref{fig:2dcolor_code}. The class of such tensor networks constructs a family of $[[3\ell(\ell+1)+1,1,2\ell]]$ codes. Similar color codes with hexagonal plaquettes can also be constructed by following the same contraction pattern in the bulk and imposing different boundary conditions. Just like the surface code construction, this tensor network represents an encoding map with a non-trivial kernel\footnote{A previous tensor network construction of the $[[19,1,5]]$ color code can be found in \cite{LTNC}, which requires both the $[[7,1,3]]$ codes and $[[9,0,3]]$ stabilizer states as building blocks. However, the protocol does not generalize to $d>5$ due to concavity of the polygonal region.}. One can similarly construct a codeword of this code, e.g. $|\bar{0}\rangle$ by contracting all the dangling logical legs with $|0\rangle$. Recall that each Steane code can be built from contracting two $[[4,2,2]]$ atomic codes, which was used to construct the surface code. As such, this tensor network can indeed be construed as a double copy \cite{colordoublecopy} of the surface code  in some sense.

Each tensor in the left figure is a Steane code where the logical leg is suppressed. For the remaining 7 physical legs, 6 are drawn in-plane while the remaining one is represented as a dot that corresponds to a physical qubit in the color code. Each stabilizer generator that acts on the plaquette of the $[[7,1,3]]$ code is mapped to a stabilizer that acts on the four physical legs adjacent to a colored quadrilateral in the tensor description. Given this QL construction, its enumerator can be computed using the same method. For example, the enumerators for a $[[91, 1, 11]]$ code are 

\begin{widetext}
\begin{align}
    A(z)&=1+54z^4+297z^6+2889z^8+24258z^{10}+197493z^{12}
    +1629738z^{14}
    +13287999z^{16}\\\nonumber
    &+108647952z^{18}+\dots\\
    B(z)&=1+54z^4+297z^6+2889z^8+24258z^{10}+4176z^{11}+197493z^{12}+67242z^{13}
    +1629738z^{14}\\\nonumber
    &+1066740z^{15}+13287999z^{16}+14401674z^{17}+108647952z^{18}+\dots
\end{align}
\end{widetext}
We see that the two coefficients start deviating at $d=11$, thus verifying its adversarial distance. 
The computation time is only tens of seconds, but a better encoding is needed to avoid unnecessary allocation of memory space for 0s in the sparse matrix. Also note that the cancellation at even weights between $A$ and $B$.

\begin{figure*}
    \centering
    \includegraphics[width=0.65\linewidth]{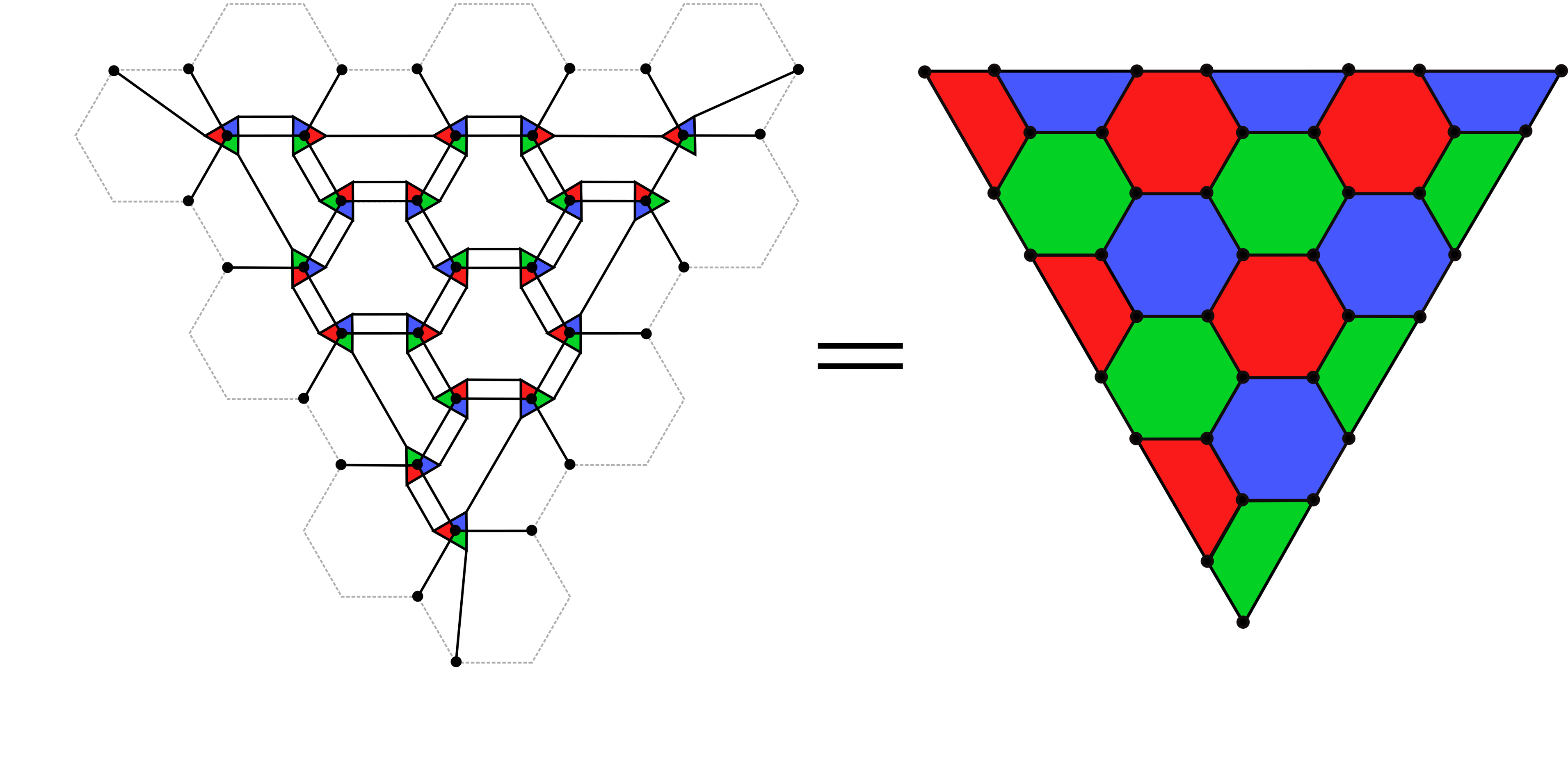}
    \caption{A $[[37,1,7]]$ 2d color code (left) tensor network construction where (right) its stabilizer generators are all $X$ or all $Z$ operators acting on the vertices of each colored plaquette.}
    \label{fig:2dcolor_code}
\end{figure*}

As we discussed in Remark~\ref{rmk:threshold}, these codes admit a common error detection threshold at $p=1/6$ (Fig.~\ref{fig:threshold}) thanks to the MacWilliams identity, and is close to the known code capacity threshold. 

\begin{figure}
    \centering
    \includegraphics[width=\linewidth]{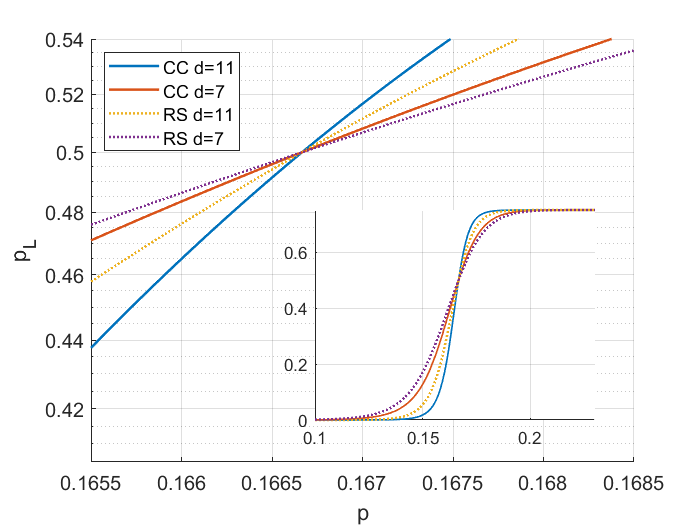}
    \caption{Error-detection thresholds coincide for the 2d color code (CC) and the surface code (RS). Zoomed out plot on the corner shows the error probability in a greater range. Only two distinct distances are shown in the plot, since other distances cross at the same value. }
    \label{fig:threshold}
\end{figure}

\begin{comment}
\begin{figure}%
    \centering
    \subfloat[\centering Rotated surface code 1]{{\includegraphics[width=0.45\linewidth]{RS_threshold.png} }}%
    \qquad
    \subfloat[\centering 2d Color code on hexagonal lattice]{{\includegraphics[width=0.45\linewidth]{ColorCode_threshold.png} }}%
    \caption{Pseudo-code capacity thresholds}%
    \label{fig:threshold}%
\end{figure}
\end{comment}

\subsection{Holographic code}\label{subsec:holographic}
To demonstrate the usefulness of mixed enumerators, we now look at a class of finite rate holographic code \cite{HaPPY} also known as the HaPPY (pentagon) code,  originally conceived as a toy model of the AdS/CFT correspondence. Different versions of this code have been proposed in various contexts \cite{Harris2020,ABSC} where preliminary studies have examined some of its behaviours under erasure errors and symmetric depolarizing noise. However, the application of such codes in quantum error correction is far less understood compared to the surface code. Here we analyze the HaPPY code as a useful benchmark using our mixed weight enumerator technology and present some novel results. 

This code can be constructed from purely $[[5,1,3]]$ atomic codes. It is known that, as a stabilizer code, it has an adversarial distance 3 regardless of $n$ because of the bulk qubits that are close to the boundary.  However, from AdS/CFT, we expect the logical qubits deeper in the bulk to be better protected and hence having different ``distances''. We can analyze the distances of these bulk qubits in different ways.

First as a stabilizer code, we define the \textit{stabilizer distance} $d_S$ of each bulk qubit as the minimal weight of all stabilizer equivalent non-identity logical operator that acts on a bulk leg/qubit \cite{Harris2020}.  To enumerate such operators, we can build a mixed enumerator by contracting a $B$-type tensor enumerator associated with the bulk tile that contains the logical qubit for which we compute the distance, with $A$-type tensor enumerators on the other tiles.  Subtracting the enumerator polynomial $A(\mathbf{u})$ of the stabilizers, we then obtain a distribution for all the non-identity logical operators acting on that bulk qubit Fig.~\ref{fig:holographic} (top right).

One can also define the word distance of this code, as in \cite{Harris2020}, where it is simply the distance of the resulting subsystem code if we isolate one bulk qubit as the logical qubit and the rest as gauge qubits. To compute the word distance, we construct an $\tilde{A}(\mathbf{u})$ enumerator by contracting $A$-type tensor enumerator on the central tile with $B$-type tensor enumerator on the rest of the network. This enumerates the logical identities in the gauge code. Then subtracting it from the scalar $B(\mathbf{u})$ enumerator of the whole code yields the distribution of all gauge equivalent non-trivial logical operators, Fig.~\ref{fig:holographic} (bottom right).

For each code of a fixed size $n$, we then repeat this for bulk qubits at different radii from the center of the graph measured in graph distance. An explicit labelling of the qubits we study is shown in Figure~\ref{fig:holographic} left\footnote{Note that this radius is different from that in \cite{Harris2020} where on the central bulk qubit is singled out and its distances are computed with respect to codes of different $n$s.}. We give a summary for $n=25$ and $n=85$ in Table~\ref{tab:holo_dist}, where $\mathcal{N}_S,\mathcal{N}_W$ denote the number of minimal weight stabilizer or gauge equivalent representations of the non-identity logical operators.

\begin{figure}
    \centering
    \includegraphics[width=\linewidth]{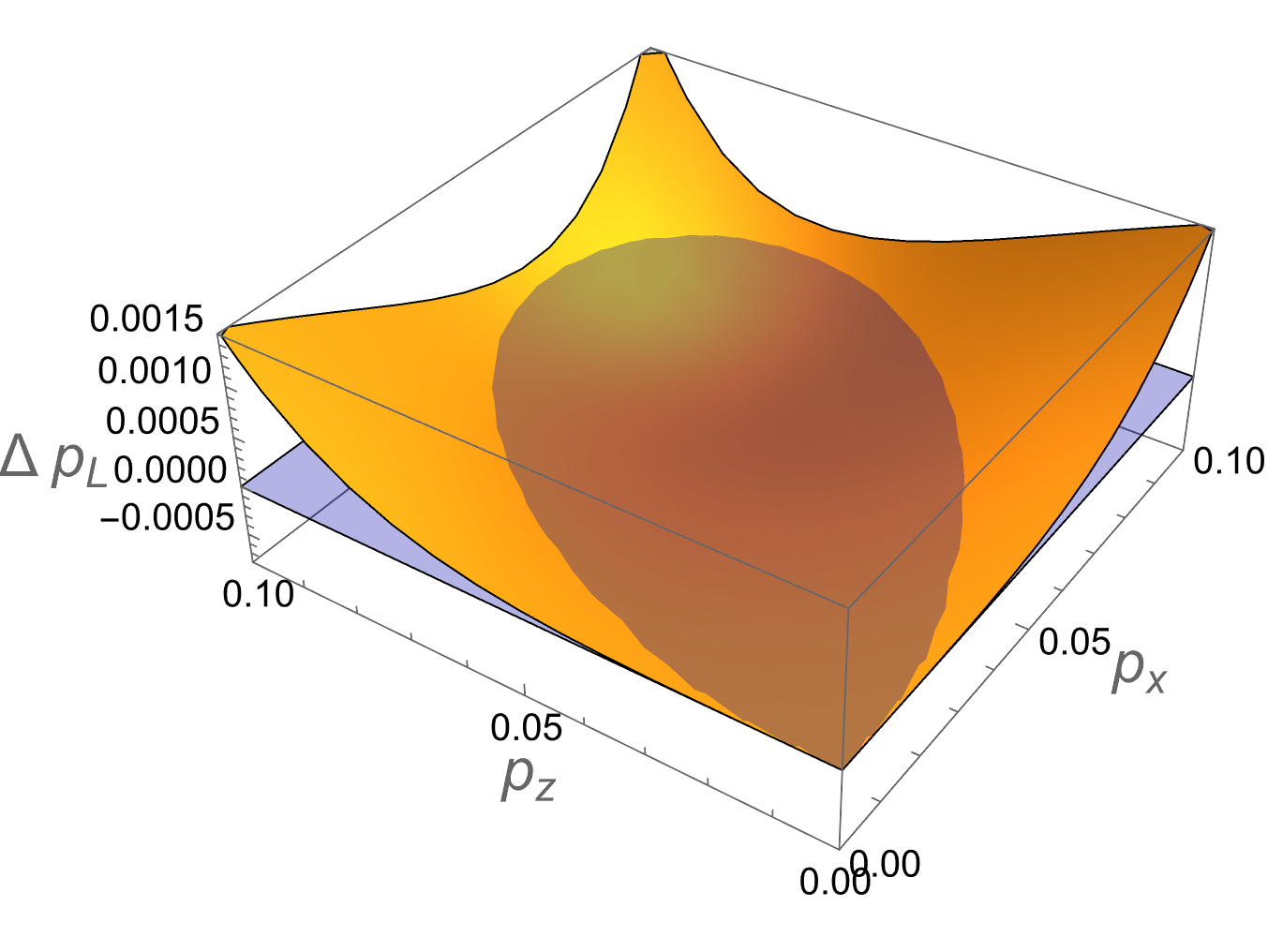}
    \caption{$\Delta p_L= p_L^{r=0}-p_L^{r=3}$ as a function of $p_x,p_z$ the bit flip and phase error probabilities. The blue translucent plane marks $\Delta p_L=0$, below which the bulk qubit at $r=0$ provides better protection.}
    \label{fig:holoxz}
\end{figure}

\begin{figure*}[t!]
    \centering
    \subfloat{{\includegraphics[width=0.45\linewidth]{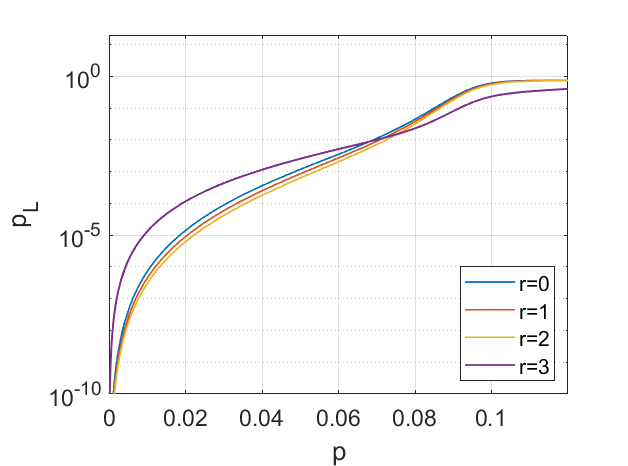} }}%
    \qquad
    \subfloat{{\includegraphics[width=0.45\linewidth]{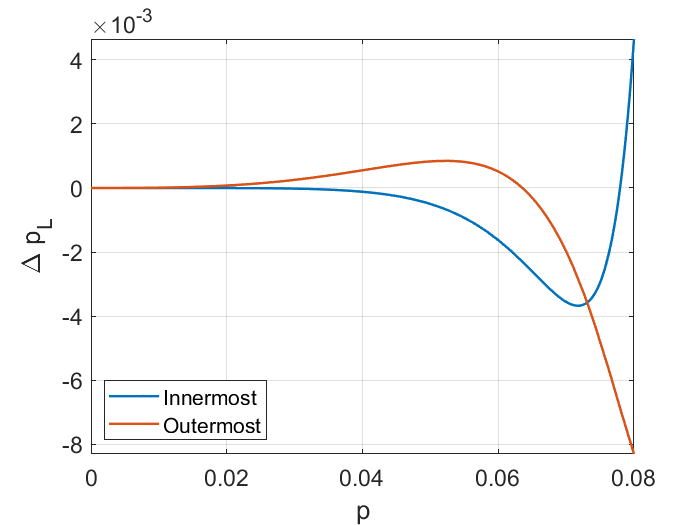} }}%
    \caption{Left: logical error probability of bulk qubits at different radii for a $[[85,41,3]]$ HaPPY code. Right: The difference between logical error rates for two HaPPY codes of radii 3 and 2. At higher $n$, the innermost bulk qubit has lower logical error rate while that for the outermost is higher for sufficiently low physical error rate $p$. The opposite is true at higher $p$.}%
    \label{fig:happy_err}%
\end{figure*}

Although the stabilizer distance decreases as a function of radius, the word distance is more or less constant with respect to the radius. This is a particular consequence of the tiling and the atomic codes, such that erasure of 4 certain boundary qubits can lead to the erasure of the inner most bulk qubit \cite{HaPPY}. Under depolarizing noise with probability $p$, the normalized uncorrectable error probability $p_L$ is shown in Fig.~\ref{fig:happy_err} (left). We see that the central bulk qubit in fact suffers from more errors because it has a greater number of minimal weight equivalent representations despite having the same word distance as most other bulk qubits. We see a crossing because the outermost bulk qubit has a slightly lower distance compared to the rest.

Despite the constant word distance as a function of system size for logical qubits that are deep in the bulk, and presumably the lack of erasure threshold for the central bulk qubit\footnote{This result assumes a particular decoder applied to small sized systems using Monte Carlo methods. It is possible that a different asymptotic behaviour can emerge with larger codes and greater accuracy. }, a larger $n$ does hint at a greater degree of error suppression. Let $\Delta p_L = p_L(n=85)-p_L(n=25)$, we see that the error rate difference for the inner most bulk qubit has a slight suppression at small $p$ while the outermost bulk qubit is the opposite. Intuitively, this is expected for general holographic codes as its construction is a slight generalization of code concatenation. As such, a crossing is expected, where adding more layers of code would generally lead to noisier bulk qubits in the deep IR when the physical error rates are sufficiently large while the opposite happens for the logical qubits in the UV. A more in-depth analysis of other holographic codes with varying word distances can be interesting as future work.

Let us also briefly examine its properties under biased noise using the double enumerator. The asymmetric distances $d^X/d^Z$ are recorded in Table~\ref{tab:holo_dist}. The XZ weight distribution is not symmetric, but the normalized logical error probability is fairly symmetric with respect to $p_X,p_Z$. Here we compare the logical error probability $\Delta p_L= p_L^{r=0}-p_L^{r=3}$ for the $n=85$ code, Fig.~\ref{fig:holoxz}. Like the symmetric depolarizing noise, the bulk qubit deeper in the bulk provides slightly better protection for the encoded information, but becomes noisier at higher physical error rates. However, the bulk qubit at $r=0$ does not provide better protection compared to the bulk qubits close to the boundary for any noise parameter in the heavily biased regime.

\begin{table*}[t]
    \centering
    \begin{tabular}{|c|c|c|c|c|c|c|c|c|c|c|c|c|}
    \hline
    ~  & \multicolumn{6}{c|}{$[[25,11,3]]$} & \multicolumn{6}{c|}{$[[85,41,3]]$}\\
    \hline
   r & $d_S$ & $\mathcal{N}_S $& $d_W$ & $\mathcal{N}_W$ & $d_S^X/d_S^Z$ & $d_W^X/d_W^Z$ & $d_S$ & $\mathcal{N}_S$& $d_W$ & $\mathcal{N}_W$ & $d_S^X/d_S^Z$ & $d_W^X/d_W^Z$\\
    \hline
    0 & 9 & 30 & 4 & 60 & 5/5 & 2/2 & 23 & 240 & 4 & 60 &13/13 &2/2\\
    \hline
    1 & 5 & 6 & 4 & 54 & 3/3 & 2/2 & 13 & 48 & 4 & 36 &7/7 &2/2\\
    \hline
    2 & 3 & 3 & 3 & 3 &1/2 & 1/2 & 9 & 12 & 4 & 24 & 5/5 &2/2\\
    \hline
    3 & n/a & n/a & n/a & n/a  & n/a & n/a& 3 & 12 & 3 &12 &1/2 &1/2 \\
    \hline
    \end{tabular}
    \caption{Tabulated stabilizer distances $d_S$ and word distances $d_W$ for two HaPPY pentagon codes at different sizes. $\mathcal{N}_S, \mathcal{N}_W$ denote the number of minimal weight stabilizer equivalent and gauge equivalent representations of non-trivial logical operators, respectively. We also provide the corresponding asymmetric stabilizer and word distances sorted by $X$ and $Z$ weights. Radial distance $r$ is the graph distance of the bulk qubit from the central tile for a code of fixed $n$. The qubits we studied are labelled according to Figure~\ref{fig:holographic}. }
    \label{tab:holo_dist}
\end{table*}

\subsection{2d Bacon-Shor code}\label{subsec:bsc}
For another example of the subsystem code, we study the 2d Bacon-Shor code. The tensor network for this code is identical to that of the surface code except we designate the physical legs every other row as gauge legs; see Appendix G.4 of \cite{CL2021}. It is conceptually convenient to think of these blocks as $[[4,2,2]]$ stabilizer codes or $[[4,1,2]]$ Bacon-Shor codes. As a subsystem code, it is most relevant to obtain its word distance. To that end we construct its mixed enumerator $I(z)$ for the logical identity. The enumerator for the non-trivial logical operators (non-identity logical operators multiplying any element of the gauge group) is $C(z)=B(z)-I(z)$. It is most convenient to express these enumerators graphically, Fig.~\ref{fig:baconshor}.

\begin{figure*}
    \centering
    \includegraphics[width=0.7\linewidth]{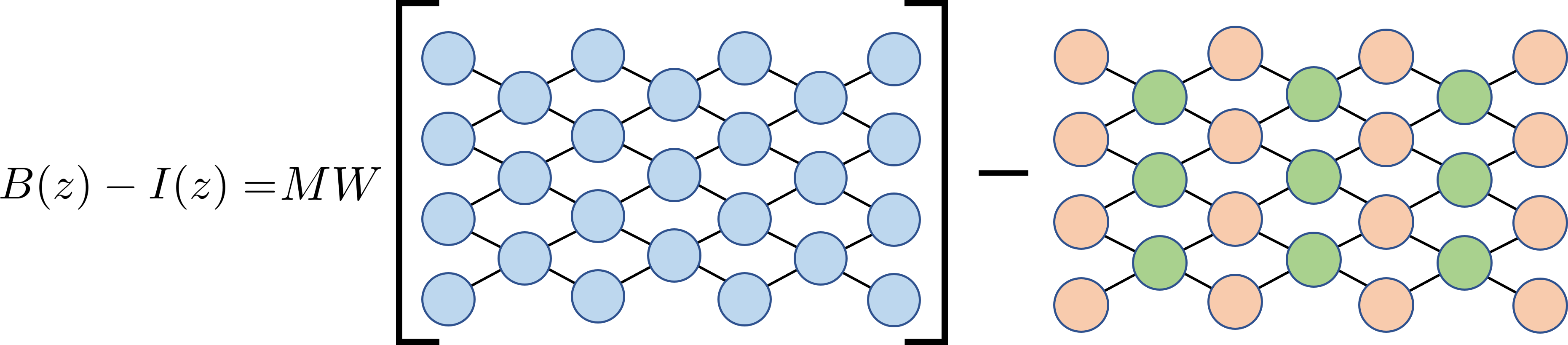}
    \caption{Distribution of non-identity logical operators in the 2d Bacon-Shor code, where blue tensors indicate $\mathbf{A}(z)$ of the $[[5,1,2]]$ code (odd columns) and $[[4,2,2]]$ codes (even columns). Green tensors are $\mathbf{A}'(z)$ of the $[[4,1,2]]$ subsystem code while orange tensors are $\mathbf{B}(z)$ of the $[[5,1,2]]$ codes. }
    \label{fig:baconshor}
\end{figure*}

Computing $B(z)$ is relatively straightforward, as we build it by contracting all $\mathbf{B}(z)$ of the $[[5,1,2]]$ and $[[4,2,2]]$ codes in the tensor network and then renormalize $B_0$ to $1$. Practically, we compute $A(z)$ by contracting all the $\mathbf{A}(z)$ of these tensors then perform a MacWilliams transform. However $I(z)$ requires extra care as we need to place $\mathbf{B}(z)$ on the odd number rows for the regular $[[5,1,2]]$ codes and $\mathbf{A}'(z)$ for the $[[4,1,2]]$ Bacon-Shor codes on even rows.  
Although these tensors in the encoding map are identical, the downward pointing legs in the $[[5,1,2]]$ code now maps to a gauge leg in the $[[4,1,2]]$ code. Therefore we must account for its weight distributions appropriately. It can be checked that the logical legs on the odd rows and columns are correlated with the logical legs on the even rows and columns. Therefore, they only contribute to an overall normalization. 

Above computations can also be easily generalized to double and complete enumerators for the Bacon-Shor code. For example, the $X,Z$ weight distributions of all non-trivial logical Pauli operator representations in this subsystem code is  shown in Fig.~\ref{fig:Bacon-ShorXZ} for the 2d Bacon-Shor code of different sizes.

\begin{figure}
    \centering
%    \begin{subfigure}[t]{0.45\textwidth}
 %       \centering
 %       \includegraphics[width=\linewidth]{6by6Bacon-Shor.png} 
  %      \caption{6 by 6 Bacon-Shor code} \label{fig:timing1}
  %  \end{subfigure}
 %   \hfill
  %  \begin{subfigure}[t]{0.45\textwidth}
  %      \centering
  %      \includegraphics[width=\linewidth]{6by7Bacon-Shor.png} 
  %      \caption{6 by 7 Bacon-Shor code} \label{fig:timing2}
  %  \end{subfigure}

  %  \vspace{1cm}
    \begin{subfigure}[t]{0.45\textwidth}
        \centering
        \includegraphics[width=\linewidth]{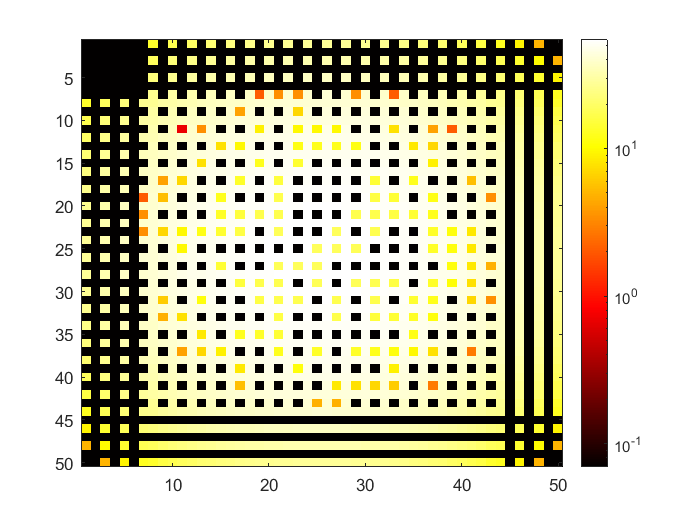} 
        \caption{7 by 7 Bacon-Shor code} \label{fig:timing3}
    \end{subfigure}
    \hfill
    \begin{subfigure}[t]{0.45\textwidth}
    \centering
    \includegraphics[width=\linewidth]{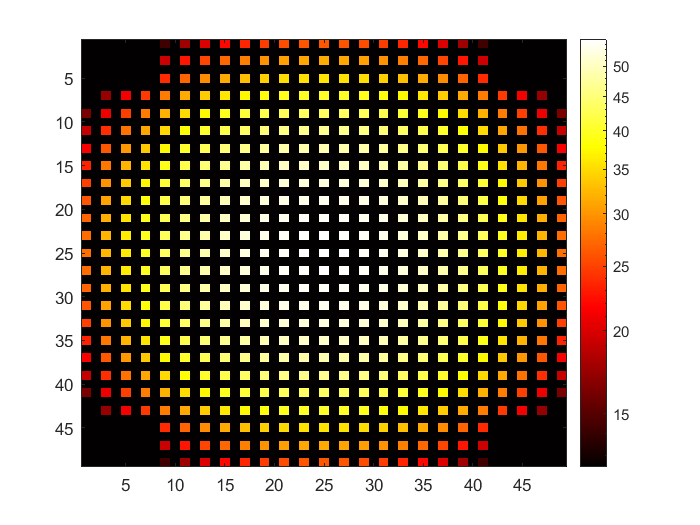}
    \caption{6 by 8 Bacon-Shor code}
    \end{subfigure}
    \caption{Plotting $\log(C_{w_x,w_z})$ in log scale, where $X$ and $Z$ weights are labelled by the vertical and horizontal axes respectively.}
    \label{fig:Bacon-ShorXZ}
\end{figure}

Note that it has a very different structure from the surface code operator weight distribution, a likely consequence of the even weight gauge generators. 

\subsubsection{2d Compass code}
Now we examine different instances of gauge fixed Bacon-Shor codes. For an $\ell\times \ell'$ Bacon-Shor code, let us fix the XX gauge by promoting $(\ell-1)(\ell'-1)$ weight-2 $X$ type gauge operators to stabilizer generators. This yields a stabilizer group with $(\ell-1)+(\ell'-1)+(\ell-1)(\ell'-1)=\ell+\ell'-2+\ell\ell'-\ell-\ell'+1=\ell\ell'-1$ generators, which is a $[[\ell\ell',1,\min(\ell,\ell')]]$ stabilizer code. 

The tensor network for this gauge can be built from the tensor of two different repetition codes 
\begin{align}
    W_R&=|00\rangle\langle 0|+|11\rangle\langle 1|\\
    W_B&=(|00\rangle+|11\rangle)\langle 0|/\sqrt{2}+(|10\rangle+|01\rangle)\langle 1|/\sqrt{2}.\nonumber
\end{align}
The code defined by $W_R$ has stabilizer $ZZ$, and $\bar{X}=XX,\bar{Z}=IZ$ and the one with $W_B$ has $X\leftrightarrow Z$ with stabilizer $XX$, $\bar{X}=IX,\bar{Z}=ZZ$. Their tensors are colored red and blue respectively. The output legs of the tensors are connected into a ring while leaving the inputs dangling. This constructs tensors in the Tree Tensor Network, Fig.~\ref{fig:xxbsc}. Each of the bigger red nodes corresponds to a stabilizer state with stabilizer group $\langle \mathrm{all~}X, \mathrm{even~weight}~ZZ\rangle$. The same holds for the big blue nodes but with $X\leftrightarrow Z$. 

\begin{figure*}[]
    \centering
    \includegraphics[width=0.7\linewidth]{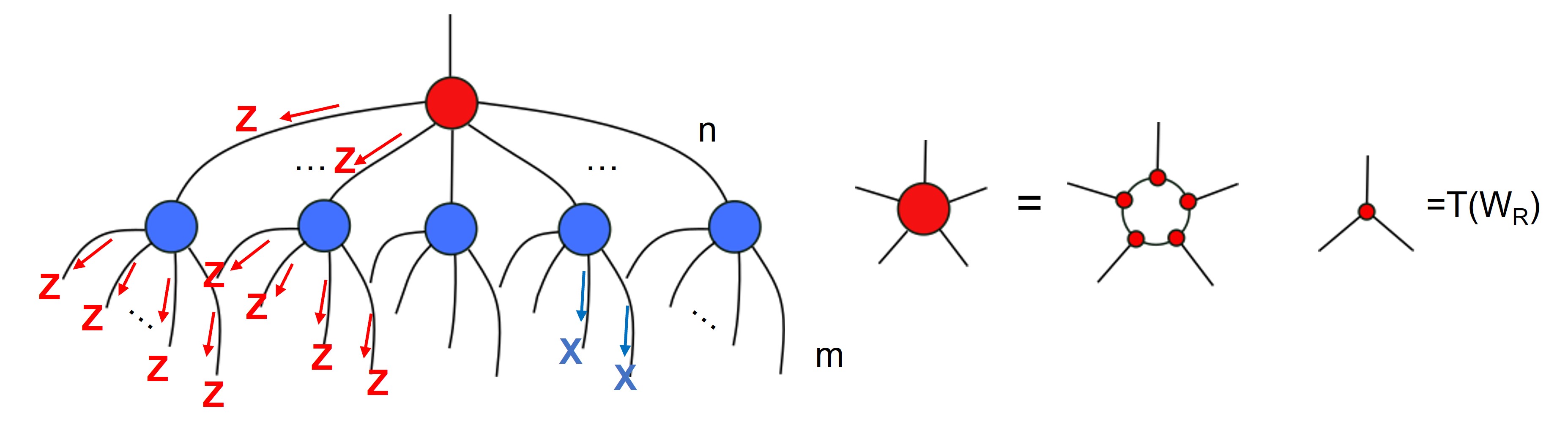}
    \caption{Tree tensor network for a $m\times n$ Bacon-Shor code in the XX gauge. Some stabilizers are shown via operator pushing. The tensors are obtained from repetition code encoding maps.}
    \label{fig:xxbsc}
\end{figure*}

Although the code has $d\sim\sqrt{n}$, the entanglement for some subsystems of size $\sim d$ can be much weaker. This allows us to write down a more efficiently contractible tensor network by taking advantage of these low entanglement cuts\footnote{Note that this speed up would not be possible for non-degenerate codes as all subsystems of size $d$ has the same entanglement $\sim d$.
}. The total time complexity for obtaining the enumerator is thus $O(\ell+\ell')\sim O(d)$ if $\ell\approx \ell'$. 

By fixing the gauges in other ways, one produces a class of codes known as the 2d compass codes \cite{compasscode} which includes a gauge that reproduces the surface code and the $XX$ (or $ZZ$) gauge we examined. Coincidentally, these are also two gauges of the Bacon-Shor code with the highest ($O(n\exp(\sqrt{n}))$) and lowest ($O(\sqrt{n})$) computational cost respectively. The entanglement structure of the underlying quantum state generally depends on the gauge choice. While this speed up is not surprising, as the example can be built from code concatenation, we can estimate how cost would scale for other patterns of gauge fixings that are everywhere-in-between provided we have tensor networks whose connectivity captures the entanglement feature. Intuitively, we can roughly understand the speed up as a statement about entangled clusters. When the code is fixed in the pure XX gauge for instance, there is little entanglement across the columns or rows of the code. If we now introduce gauge fixing such that $ZZ$ stabilizers can occur with some non-vanishing fraction, this introduces more entanglement across these clusters and the resulting tensor network minimal cut now has to cut through these additional bridges of entanglement. Generally, we then expect the complexity to scale exponentially as the width of these bridges, or the minimal cuts that separates these clusters. In the extreme case of the surface code, the bridges are of $\sqrt{n}$, and in the pure XX or ZZ gauge, the bridge is of $O(1)$.  By slowly deforming from the $XX$ or $ZZ$ gauge, one may also explore the intermediate regime of complexities $O(\sqrt{n})\rightarrow poly(n)\rightarrow O(\exp(\sqrt{n}))$\footnote{In this way, the $\sim\exp(d)$ cost of computing the Bacon-Shor weight enumerator is not surprising as the unfixed tensor network encompasses all 2d compass code configurations.}. A more comprehensive study of this complexity transition and gauge fixing can be an interesting subject for future exploration.

\section{Discussion}\label{subsec:discussion}
In this work, we generalize the existing weight enumerator formalism to study cosets, subsystem codes, and all single qubit error channels. In conjunction with tensor networks, we extend their applications in quantum error correction. We show that weight enumerators can be computed more efficiently using tensor network methods once a QL construction of the code is known. The complexity can vary depending on the tensor network connectivity, and is dominated by the cost of tensor contractions. For a QL construction that faithfully reflects the entanglement structure of the code words, the cost for finding their enumerator is $\sim O(\exp(d))$ for non-degenerate codes and up to exponentially faster for degenerate codes. As a novel distance-finding protocol, our proposal constitutes the only and the best current algorithm for finding the distance beyond stabilizer codes. In the case of Pauli stabilizer codes, this provides a comparable performance for non-degenerate codes, and up to exponential speed up for degenerate codes. 

Using the generalized coset enumerators, we also construct (optimal) decoders for all codes using weight enumerators for all i.i.d.~single qubit error channels. As a corollary, it improves the simulation accuracy when estimating fault-tolerant thresholds if used in conjunction with existing methods. Since QL includes all quantum codes, and thus stabilizer codes, the enumerator method can also be understood as a generalization of tensor network decoders. Finally we applied our method numerically to codes with sizes of order ${100}$ to $200$ qubits, showing that it is practical to study codes of relevant sizes in near-to-intermediate term devices. We also provide novel analysis of the surface code, color code, holographic code and the Bacon-Shor code using exact analytical expressions. These include their full operator weight distributions and certain code performance under coherent or biased noise. For the holographic code, we also present new results on asymmetric distances and the varied behaviour of different bulk qubits under (biased) Pauli error. 

{\color{black}This advance also has a wide range of applications in the context of measurement-based quantum computation quantum manybody physics. We have shown that higher genus weight enumerators computes the stabilizer Renyi entropy, or magic, of a quantum state. It is also known that Shor-Laflamme enumerator, or sector length, is a powerful tool to study the entanglement structure of cluster states. With these novel connections between coding-theoretic objects and quantum resource-theoretic quantities, our method provides a far more efficient method to characterize entanglement and magic in quantum manybody systems compared to brute force evaluation. For the case of graph states, existing methods to compute sector lengths have been limited to order 30 qubits. Our numerics from 2d tensor network suggests that this may be pushed to about ten times higher with modest hardware requirements on 2d cluster states. Numerical computation of quantum manybody magic is also widely recognized as a challenging problem. Here we provide an alternative method that is readily implementable for a variety of tensor network architectures.

We also provided a systematic method for building QL decomposition of existing quantum codes, filling the void left by our previous work \cite{CL2021}. In particular, our novel tensor network construction applied to quantum LDPC codes provides a simple algorithmic method to analyze such codes from the perspective of quantum manybody systems using bounded-degree tensor networks.
}

\subsection{Connection with stat mech mapping}
We comment on the connection between optimal decoding and distance from the point of view of the statistical mechanical mapping and weight enumerators.
Recall that the coset weight enumerator polynomial $A(\bar{E},\mathbf{u})$ of $E$ captures the weight distribution of all operators that are stabilizer equivalent to $E$. By plugging in the corresponding coefficients $\mathbf{k}$ from decomposing the error channels, one obtains the probability of incurring any errors that are equivalent to $E$. This is nothing but the partition function $Z_E$ by solving the stat mech mapping \cite{ChubbFlammia} associated with a noise model that satisfies the Nishimori condition for all parameters $\beta J_i$ where $\beta$ is the inverse temperature and $\{J_i\}$ are coupling strengths of the model. 

Conversely, if the error probability from the stat mech model can be obtained exactly, then it must agree with $A(\bar{E},\mathbf{u})$ in some domain that is a connected region near the origin. Since if two polynomials $f,g$ agree in infinite number of points, $f-g$ must have an infinite number of roots. This cannot happen for any non-trivial $f-g$ because the degree $deg(f-g)\leq \max (deg(f),deg(g))$ is bounded. This implies that the solution $Z_E=A(\bar{E},\mathbf{u})$ must be unique. Therefore by solving the stat mech model and obtaining its partition function for different values of $\beta J_i$, we must also have sufficient information to uniquely fix the enumerator polynomial. For example for symmetric Pauli noise, one can in principle fix the coefficients of the polynomial by computing the values of $Z_E(\beta)$ at different temperatures. As there are only finitely many coefficients for $A^s$, one can solve an overconstrained system of equations with integer solutions. 

In practice, however, the expression for $Z_E=\Pr(\bar{E})$ in the stat mech model is often obtained numerically. Therefore, unless $P=NP$ (or $NP=RP$) the reverse process going from the stat mech output to the enumerator can only be trusted to produce the correct results only when the values of $\Pr(\bar{E})$ hold to exponential accuracy generally. This is expected, because otherwise one can solve the minimal distance problem approximately \cite{DumerApprox} using the stat mech model in polynomial time with approximate tensor contraction.
In instances where the partition functions can be (or have been) obtained with relatively high accuracy such that the cost is less expensive compared to the current enumerator method, one can also acquire polynomially many values of the partition functions at different coupling strengths. One can then fit the coefficients of the enumerator polynomial to these data points. This allows us to derive (an approximation of) the enumerator and thus also extrapolat the error probabilities to other regimes instead of evaluating those points individually using the stat mech model.

\subsection{Future directions}

Recently, it was shown that asymptotically good quantum Low-Density-Parity-Check (LDPC) codes like \cite{NLTS} have a circuit depth lower bound that is $\log n$. Since these codes are highly degenerate and some may sustain linear distances even with a much lower entanglement along some cuts, it is possible that a good tensor network description may lead to a more efficient distance verification protocol for codes whose code words saturate the entanglement lower bound. However, we also note that small sized examples, their tensor network descriptions, and a tight entanglement lower bound are still open problems as of the time of writing, the advantage our method provides only remains a theoretical possibility\footnote{The expansion property of these codes may naively indicate the edge cut to scale with volume. However, as they are not the corresponding tensor networks, and that edge cuts are only upper bounds of entanglement, it may be possible that a sparser tensor network can be found that permits fewer cuts.}. Therefore, a general QL recipe for building qLDPC codes would be useful.

As weight enumerators are applicable for non-(Pauli)-stabilizer codes, they can be used to study or search for such codes while providing crucial information on their distances. This would extend the examples in this work beyond stabilizer codes and would also have relevant applications in optimization-based methods that need not produce stabilizer codes \cite{VQAQEC}. For example, XS or XP codes \cite{XS,XP} do not have abelian stabilizer groups and currently lack a protocol for computing their code distances. However for general codes, reduced enumerators are likely insufficient, and a higher bond dimension will be needed.

Note that beyond QECC literature, Shor-Laflamme enumerators, also known as sector lengths in graph states \cite{Eltschka2020maximumnbody,SLD}, have been used to study the structure as well as the robustness of entanglement in entangled resource states. {\color{black} Sector lengths of graph states are difficult to compute using the brute force method. Given our tensor network decomposition of all graph states, we expect our method to carry immediate impact in the analysis of 2D or planar cluster states with $>100$ qubits using sector lengths as well as more general applications in the context of fault-tolerant resource state preparation for measurement/fusion-based quantum computations and quantum networks. 

In the context of quantum many-body magic, non-stabilizerness has been difficult to compute numerically. As stabilizer Renyi entropy and other measures have been related to quantum chaos, entanglement spectrum and the emergence of gravity in AdS/CFT correspondence, it is interesting to explore whether the tensor network methods based on enumerators are advantageous for more efficient magic computations. It is also intriguing to understand whether the quantum MacWilliams identity can provide important constraints for quantum many-body entanglement and magic.}

Tensor enumerator methods are also useful when used in conjunction with machine learning (especially reinforcement learning)-based methods for QECC search \cite{QLRL,TNCRL}. As one would typically need to evaluate certain code properties, such as distance, that are resource intensive, the tensor enumerator method can be used to drastically decrease the time needed to evaluate the cost function. It is also of interest to study the effect of approximate tensor contractions and how they impact the accuracy of the weight distribution and related distance information. 

While we have treated all i.i.d. single qubit errors, the current formalism does not tackle general location-based or correlated error efficiently. For the former, a straightforward extension exists where one can either introduce an additional variable for each location that has independent error pattern. This remains efficient as long as the types of distinct error channels is small, but can quickly become intractible if it scales with the system size. Alternatively, one can precontract the tensor with a fixed error parameter $\{p_i\}$ instead of describing them as variables. The latter reduces to a more general tensor network decoder \cite{FP2013,TNrealnoise,TNdec,TNML,DarmawanPoulinCorr}. {\color{black}In particular, \cite{DarmawanPoulinCorr} shows that correlated errors can be efficiently studied using PEPO with approximate tensor contraction. Both correlated noise and approximate contraction schemes can be interesting future avenues of research in the context of tensor weight enumerators. }In the same vein, further extension is needed to describe fault-tolerant processes which are fundamentally dynamical. Therefore, an enumerator framework compatible with space-time quantum error correction that incorporates gadgets that includes measurement errors, mid-circuit noise and POVMs will be needed.

Finally, while enumerators were first defined in classical coding theory, one yet needs an efficient method to compute them for classical codes. Therefore, it is natural to extend the current QL-based approach to classical codes and compute their weight enumerator polynomials. Such tasks may be accomplished by directly applying the current formalism for classical codes and rephrasing them as quantum stabilizer codes with trivial generators, or devising a more efficient method that performs the analog of the trace or conjoining operation for classical codes.

\section*{Acknowledgement}
We thank Y.D. Li, D. Miller, G. Sommers, and Y.J. Zou for helpful discussions and comments on the manuscript.
C.C. acknowledges the support by the U.S. Department of Defense and NIST through the Hartree Postdoctoral Fellowship at QuICS, the Air Force Office of Scientific Research (FA9550-19-1-0360), and the National Science Foundation (PHY-1733907). M.J.G. acknowledges support from the National Science Foundation (QLCI grant OMA-2120757). The Institute for Quantum Information and Matter is an NSF Physics Frontiers Center. Certain commercial equipment, instruments, or materials are identified
in this paper in order to specify the  procedure adequately and do not reflect any endorsement by NIST.
\appendix

\section{Common Scalar Enumerators}\label{app:scalarenum}
For completeness, we review a few examples below that we we have used in this work. 
\subsection{Shor-Laflamme weight enumerator}
The original weight enumerators \cite{macwilliams1963theorem,macwilliams1977theory} are important objects in classical coding theory. Their quantum counterparts were introduced by Shor and Laflamme \cite{ShorLaflamme}, which capture some key properties of an error correcting code. They feature a duo of polynomials that take the forms of

\begin{align}\label{al:shorlaflamme}
    A(z,w) &= \sum_{d=0}^n A_d(M_1,M_2) z^dw^{n-d}\\
    B(z,w) &= \sum_{d=0}^n B_d(M_1,M_2) z^dw^{n-d},
\end{align}
where 
\begin{align}
    A_d(M_1,M_2) &= \sum_{E \in \mathcal{E}[d]} \Tr(EM_1)\Tr(EM_2), \text{ and }\\
    B_d(M_1,M_2) &= \sum_{E \in \mathcal{E}[d]} \Tr(EM_1 EM_2)
\end{align}
for some Hermitian $M_1,M_2$ and $\mathcal{E}[d]$ which denotes unitary errors of weight $d$. Here without loss of generality we can simply choose the Pauli basis.
Note that they are a special case of the abstract enumerator \cite{CL2022}, and we may recover them by setting $\mathbf{u}=(w,z)$ and

\be\mathrm{wt}(E) = \left\{\begin{array}{cl} (1,0) & \text{if $E = I$}\\ (0,1) & \text{otherwise.}\end{array}\right.\ee
So that $\mathbf{u}^{wt(E)}=w^{n-wt(E)}z^{wt(E)}$, where $wt(E)$ is simply the operator weight of the Pauli string $E$.

These polynomials are related by the MacWilliams identity
\begin{equation}
    B(w,z)= A\left(\tfrac{w+(q^2-1)z}{q},\tfrac{w-z}{q}\right).
\end{equation}
Therefore, it is sufficient to obtain one of them, and perform MacWilliams transform to get the other. In practice, for a brute force algorithm, it is often easier to recover $A(z,w)$. 

Note that these polynomials are sometimes expressed in the inhomogeneous form where $A(z)=A(w=1,z), B(z)=B(w=1,z)$. As it is simple to recover the homogenized form by setting $A(z)\rightarrow w^nA(z/w)$ and similarly for $B$, we refer to both of them weight enumerators as they contain the same information as encoded by the coefficients.

\subsection{Refined Enumerators}
We can also consider a generalization of the Shor-Laflamme polynomial (\ref{al:shorlaflamme}) where we separate the weights by type \cite{hu2020weight}. One such example is the double weight enumerator. Using variables $\mathbf{u}=(w,x,y,z)$

\begin{equation}\label{eqn:doublewt}
\mathrm{wt}(E) = \left\{\begin{array}{cl} 
(0,1,0,1) & \text{if $E = I$}\\ 
(0,0,1,1) & \text{if $E = X$}\\ 
(1,0,1,0) & \text{if $E = Y$}\\ 
(1,1,0,0) & \text{if $E = Z$}\end{array}\right.\end{equation}

This is useful when for instance, we consider a biased error model where bit flip $(X)$ and phase $(Z)$ errors occur independently with different probabilities. Depending on the form of the biased Pauli noise, other weight functions may be used for the weight function. Such double enumerators may be used as long as the biased Pauli noise only admits two independent physical error parameters. 
The polynomials are
\begin{align}
    &D(x,y,z,w;M_1.M_2)\\\nonumber
    =&\sum_{w_x,w_z}^n D_{w_x,w_z}y^{w_x}w^{w_z}x^{n-w_x}z^{n-w_z}\\
    &D^{\perp}(x,y,z,w; M_1,M_2)\\\nonumber
    =&\sum_{w_x,w_z}^n D_{w_x,w_z}^{\perp}y^{w_x}w^{w_z}x^{n-w_x}z^{n-w_z},
\end{align}
where
\begin{align}
    D_{w_x,w_z}&=\sum_{E\in \mathcal{E}[w_x,w_z]}\Tr[EM_1]\Tr[E^{\dagger}M_2]\\
    D^{\perp}_{w_x,w_z}&=\sum_{E\in\mathcal{E}[w_x,w_z]}\Tr[EM_1 E^{\dagger}M_2],
\end{align}
and $\mathcal{E}[w_x,w_z]$ is the set of Paulis that have $X$ and $Z$ weights $w_x,w_z$ respectively. 

The MacWilliams identity was derived in \cite{hu2020weight} for local dimension 2 where $M_1=M_2$ are projection operators onto the code subspace. In \cite{CL2022}, it was extended arbitrary local dimension $q$ and $M_1,M_2$. We reproduced the relation here for convenience
\begin{align}
    &D^{\perp}(x,y,z,w)\\\nonumber &\quad = D\left(\tfrac{x+(q-1)y}{\sqrt{q}}, \tfrac{z-w}{\sqrt{q}}, \tfrac{z+(q-1)w}{\sqrt{q}}, \tfrac{x-y}{\sqrt{q}}\right).
\end{align}

The inhomogeneous forms are
\begin{align}
    D(y,w)&=\sum_{w_x,w_z}^n D_{w_x,w_z} y^{w_x}w^{w_z}\\
    D^{\perp}(y,w)&=\sum_{w_x,w_z}^n D^{\perp}_{w_x,w_z} y^{w_x}w^{w_z}.
\end{align}
One can easily restore the $x,z$ dependence as their powers are fixed by $n, w_x,w_z$.

\begin{theorem}\label{thm:refined_dist}
If $t_x,t_z$ are the two largest integers such that $D_{w_x,w_z}=D^{\perp}_{w_x,w_z}$ for $w_x<t_x, w_z<t_z$, then $d_x=t_x, d_z=t_z$.
\end{theorem}
\begin{proof}
    See Theorem 8 of \cite{hu2020weight}.
\end{proof}

A even more refined weight function distinguish all the Pauli operators by their types
\be\mathrm{wt}(E) = \left\{\begin{array}{cl} 
(1,0,0,0) & \text{if $E = I$}\\ 
(0,1,0,0) & \text{if $E = X$}\\ 
(0,0,1,0) & \text{if $E = Y$}\\ 
(0,0,0,1) & \text{if $E = Z$}\end{array}\right.\ee
This is known as the complete weight enumerator \cite{hu2020weight}. Again, let $\mathbf{u}=(w,x,y,z)$
\begin{align}
    &E(x,y,z,w;M_1,M_2)\\\nonumber
    =&\sum_{w_x,w_y,w_z}E_{w_x,w_y,w_z}x^{w_x} y^{w_y}z^{w_z} w^{n-w_x-w_y-w_z}\\
   &F(x,y,z,w;M_1,M_2)\\\nonumber
   =&\sum_{w_x,w_y,w_z}F_{w_x,w_y,w_z}x^{w_x} y^{w_y}z^{w_z} w^{n-w_x-w_y-w_z},
\end{align}
where
\begin{align}
      E_{w_x,w_y,w_z}  = &\sum_{Q\in \mathcal{E}[w_x,w_y,w_z]}\Tr[QM_1]\Tr[Q^{\dagger}M_2] \mathbf{u}^{wt(Q)}\\
      F_{w_x,w_y,w_z} =&\sum_{Q\in \mathcal{E}[w_x,w_y,w_z]}\Tr[QM_1 Q^{\dagger}M_2]\mathbf{u}^{wt(Q)}.
\end{align}
and $\mathcal{E}[w_x,w_y,w_z]$ are the Pauli operators with those $X,Y$ and $Z$ weights respectively. See \cite{CL2022} for general MacWilliams identities at any $q$.

\subsection{Applications to stabilizer codes}
Before we move on to tensor enumerators, let us build up some intuition as to what these polynomials are enumerating. Let us examine a special case where we set $M_1=M_2=\Pi$ to be the projection onto the code subspace of a quantum code. Furthermore, let us suppose that this is a $[[n,k]]$ stabilizer code, meaning that

\begin{equation}
    \Pi= \frac{1}{2^{n-k}}\sum_{S\in \mathcal{S}} S
\end{equation}

It is clear that $\Tr[E\Pi]\neq 0$ if and only if $E\in \mathcal{S}$ is a stabilizer element and $\Tr[E\Pi E^{\dagger}\Pi]\neq 0$ if and only if $E\in \mathcal{N}(\mathcal{S})$ is a normalizer element. Therefore, we see that, up to a constant normalization factor, the coefficients $A_d$ of $A(z;\Pi,\Pi)$ is simply enumerating the number of stabilizer elements with weight $d$ and $B_d$ enumerating the number of logical operators with weight $d$. Consequently, $\sum_dA_d=2^{n-k}$ and $\sum_d B_d=2^{n+k}$ for a $[[n,k]]$ code.

For the refined enumerators, the coefficients of the double enumerator $D,D^{\perp}$ are simply recording the number of stabilizer and normalizer elements that have $X,Z$-weights $(w_x,w_z)$. Similarly, the complete enumerator coefficients $E_{w_x,w_y,w_z},F_{w_x,w_y,w_z}$ count the elements with those corresponding $X,Y$ and $Z$ weights.

One can also set $M_1,M_2$ to different operators to extract different information about the code. For example, in coset enumerators, $A^s_d$ counts the number of coset elements with a particular weight.

\section{Instances of Tensor Enumerators}\label{app:tensorenum}
We have seen previously particular instances of tensor enumerators with $\mathbf{u}=z$. One can extend examples in the main text to other enumerators, which we have used to study other error models. 
\subsection{Refined Tensor Enumerators}
Similar to the scalar forms, we apply $\mathbf{u}=(w,x,y,z)$ and the weight function \ref{eqn:doublewt} to \ref{eqn:tensor_enum}.
The tensor coefficients are
\begin{align}\label{eqn:B3}
\begin{split}
    &D^{(J)}_{w_x,w_z}(E,\bar{E};M_1,M_2)\\
    & = \sum_{F \in \mathcal{E}^{n-m}[w_x,w_z]} \Tr((E \otimes_J F)M_1) \Tr((\bar{E} \otimes_J F)^\dagger M_2),\\
    &D^{\perp(J)}_{w_x,w_z}(E,\bar{E};M_1,M_2)\\
    & = \sum_{F \in \mathcal{E}^{n-m}[w_x,w_z]} \Tr((E \otimes_J F)M_1 (\bar{E}^\dagger \otimes_J F^\dagger) M_2),
    \end{split}
\end{align}
where $J\subseteq\{1,2,\dots n\}$ are the locations of open legs in the tensor enumerator. As in the main text, $\otimes_J$ denotes the operation where we insert $E$s at corresponding positions of $J$ indices to form a $n$-qubit Pauli string with $F$ which has length $n-m$ for $m$ open indices.

For complete tensor enumerators, we replace $\mathcal{E}[w_x,w_z]\rightarrow \mathcal{E}[d_x,d_y,d_z]$ and 
\begin{align}\begin{split}
    D_{w_x,w_z}^{(J)}(E,\bar{E},M_1,M_2)\rightarrow E_{d_x,d_y,d_z}^{(J)}(E,\bar{E},M_1,M_2),\\ D^{\perp (J)}_{w_x,w_z}(E,\bar{E},M_1,M_2)\rightarrow F_{d_x,d_y,d_z}^{(J)}(E,\bar{E},M_1,M_2)
    \end{split}
\end{align} 
in (\ref{eqn:B3}). $\mathcal{E}[d_x,d_y,d_z]$ is the set of Pauli operators with $X,Y,Z$ weights given by $d_x,d_y,d_z$ respectively.

\subsection{Generalized Tensor Enumerators}
For the most general noise model, it is also useful to define generalized abstract enumerator
\begin{widetext}
\begin{align}
    \tilde{\mathbf{A}}^{(J)}(\mathbf{u}; M_1, M_2) &= \sum_{E,\bar{E}\in\mathcal{E}^m} \sum_{F,\bar{F} \in \mathcal{E}^{n-m}} \Tr((E \otimes_J F)M_1) \Tr((\bar{E}^\dagger \otimes_J \bar{F}^\dagger) M_2) \mathbf{u}^{\mathrm{wt}(F,\bar{F})} e_{E,\bar{E}},\\
    \tilde{\mathbf{B}}^{(J)}(\mathbf{u}; M_1, M_2) &= \sum_{E,\bar{E}\in\mathcal{E}^m}  \sum_{F \in \mathcal{E}^{n-m}}  \Tr((E \otimes_J F)M_1 (\bar{E}^\dagger \otimes_J \bar{F}^\dagger) M_2) \mathbf{u}^{\mathrm{wt}(F,\bar{F})} e_{E,\bar{E}},
\end{align}
\end{widetext}
where the forms are similar to the conventional tensor enumerator but the sum and weight function now depend on two independent variables $F,\bar{F}$. These are useful for computing generalized scalar weight enumerators (Sec~\ref{subsec:genabs}) which finds applications in noise models such as coherent noise or amplitude damping channel (Sec~\ref{subsec:errdet}). 

\begin{theorem}
    Suppose $j,k\in J\subset\{1,\dots,n\}.$ Then
    \begin{align}\begin{split}
        &\wedge_{j,k}\tilde{\mathbf{A}}^{(J)}(\mathbf{u};M_1,M_2)\\
        &\quad = \tilde{\mathbf{A}}^{(J\setminus\{j,k\})}(\mathbf{u};\wedge_{j,k}M_1,\wedge_{j,k}M_2)
        \end{split}
    \end{align}
    and similarly for $\tilde{\mathbf{B}}$.
\end{theorem}
\begin{proof}
\begin{widetext}
\begin{align}
\begin{split}
    &\wedge_{jk}\tilde{\mathbf{A}}^{(J)}(\mathbf{u};M_1,M_2) = \sum_{E,\bar{E},F,\bar{F}} \left[\Tr((E\otimes_J F)M_1)\Tr((\bar{E}\otimes_J \bar{F})^\dagger M_2\right] \mathbf{u}^{wt(F,\bar{F})}\left[\wedge_{j,k}e_{E,\bar{E}}\right]\\
    &= \sum_{F,\bar{F}}\sum_{\scriptsize\begin{array}{c}E\setminus\{E_j,E_k\},\\\bar{E}\setminus\{\bar{E}_j,\bar{E}_k\}\end{array}} \sum_{G} \Big\{ \Tr([((G\otimes G^*) \otimes_{j,k} E\setminus\{E_j,E_k\}) \otimes_J F] M_1) \\
    &\qquad \qquad \cdot \sum_{\bar{G}}\Tr([((\bar{G}\otimes \bar{G}^*) \otimes_{j,k} \bar{E}\setminus\{\bar{E}_j,\bar{E}_k\}) \otimes_J \bar{F}]^\dagger M_2)\Big\} \mathbf{u}^{wt(F,\bar{F})} e_{E\setminus\{E_j,E_k\},\bar{E}\setminus\{\bar{E}_j,\bar{E}_k\}}\\
     &= \sum_{E',\bar{E}',F} \Tr((E'\otimes_{J\setminus\{j,k\}} F) (|\beta\rangle\langle\beta|_{j,k} M_1)) \Tr((\bar{E}'\otimes_{J\setminus\{j,k\}} \bar{F})^\dagger (|\beta\rangle\langle\beta|_{j,k} M_2)) \mathbf{u}^{wt(F,\bar{F})}e_{E',\bar{E}'}\\
    &= \sum_{E',\bar{E}',F} \Tr((E'\otimes_{J\setminus\{j,k\}} F) (\wedge_{j,k} M_1)) \Tr((\bar{E}'\otimes_{J\setminus\{j,k\}} \bar{F})^\dagger (\wedge_{j,k} M_2)) \mathbf{u}^{wt(F,\bar{F})}e_{E',\bar{E}'}\\
    &= \tilde{\mathbf{A}}^{(J\setminus\{j,k\})}(\mathbf{u}; \wedge_{jk}M_1, \wedge_{jk}M_2)
    \end{split}
\end{align}
\end{widetext}
where the wedge acts on the vector basis in the usual way.

We used the fact that 
\be|\beta\rangle\langle\beta| = \frac{1}{\D}\sum_{P \in \mathcal{P}} P\otimes P^{*}.\ee

\begin{comment}
$$\wedge_{j,k} e_{E,\bar{E}} = \left\{\begin{array}{cl} e_{E\setminus\{E_j,E_k\},\bar{E}\setminus\{\bar{E}_j,\bar{E}_k\}} & \text{if $E_j = E_k^* ~and~ \bar{E}_j = \bar{E}_k^*$}\\ 0 & \text{otherwise,}\end{array}\right.$$
\end{comment}

Similarly, we can repeat this argument for $B$-type generalized enumerators.
\begin{comment}
\begin{widetext}
\begin{align*}
    \wedge_{jk}\tilde{\mathbf{B}}^{(J)}(\mathbf{u};M_1,M_2) &= \sum_{E,\bar{E},F,\bar{F}} \left[\Tr((E\otimes_J F)M_1(\bar{E}\otimes_J \bar{F})^\dagger M_2)\right]\mathbf{u}^{wt(F,\bar{F})}\left[\wedge_{j,k}e_{E,\bar{E}}\right]\\
    &= \sum_{F,\bar{F}}\sum_{\scriptsize\begin{array}{c}E\setminus\{E_j,E_k\},\\\bar{E}\setminus\{\bar{E}_j,\bar{E}_k\}\end{array}} \sum_{G,\bar{G}} \left\{ \Tr([((G\otimes G^*) \otimes_{j,k} E\setminus\{E_j,E_k\}) \otimes_J F] M_1 \right.\\
    & \cdot \left. [((\bar{G}\otimes \bar{G}^*) \otimes_{j,k} \bar{E}\setminus\{\bar{E}_j,\bar{E}_k\}) \otimes_J \bar{F}]^\dagger M_2)\right\} \mathbf{u}^{wt(F,\bar{F})} e_{E\setminus\{E_j,E_k\},\bar{E}\setminus\{\bar{E}_j,\bar{E}_k\}}\\
    &= \sum_{E',\bar{E}',F} \Tr((E'\otimes_{J\setminus\{j,k\}} F) (\wedge_{j,k} M_1)(\bar{E}'\otimes_{J\setminus\{j,k\}} \bar{F})^\dagger (\wedge_{j,k} M_2)) \mathbf{u}^{wt(F,\bar{F})} e_{E',\bar{E}'}\\
    &= \tilde{\mathbf{B}}^{(J\setminus\{j,k\})}(\mathbf{u}; \wedge_{jk}M_1, \wedge_{jk}M_2).
\end{align*}
\end{widetext}
\end{comment}
We do not use MacWilliams identity for this proof as we have not been able to identify any.
\end{proof}

Note that it is often possible to cut down the computational cost when the weight function satisfies the form 
\begin{equation}
    \mathbf{u}^{wt(E,F)} = \mathbf{u}_1^{wt(E)}\mathbf{u}_2^{wt(F)}.
\end{equation}
Then we can write the generalized enumerator as
\begin{align}
    &\bar{A}(\mathbf{u};M_1,M_2)\\ \nonumber&=\sum_{E\in\mathcal{E}^n}\Tr[EM_1]\mathbf{u}_1^{wt(E)}\sum_{F\in \mathcal{E}^n}\Tr[F^{\dagger}M_2] \mathbf{u}_2^{wt(F)}
\end{align}
that factorize into two separate sums such that each piece can be computed separately. For each $M_i$, we can rewrite as a tensor network. This allows us to compute either sum using a tensor network of $\chi=q^2$ by tracing reduced tensor enumerators. We see that for stabilizer codes, the coefficients for each term are identical to the usual $A$-type scalar weight enumerator up to a constant factor normalization. 

\subsection{Stabilizer codes and reduced enumerators}\label{subapp:reducedenum}
Again, let us come back to stabilizer codes for intuition behind these constructions. Consider the reduced tensor enumerator polynomial with open indices $J=\{j_1,\dots j_m\}$ where we set $M_1=M_2=\Pi$ to be the projection onto stabilizer code subspace. We see that each coefficient $A_d^{j_1,\dots,j_m}$ simply enumerates, up to a constant normalization, the number of stabilizer elements that has Pauli string $\sigma^{(j_1)}\otimes\dots\otimes \sigma^{(j_m)}$ on the $1$st through $m$-th qubit/qudit and has weight $d$ on the remaining qubit/qudits. Similarly for the reduced double, complete, and the generalized enumerators, the same intuition applies, except the weights are separated and recorded according to the types of the Pauli operators.

The tracing of reduced enumerators for stabilizer codes can be understood as a simple consequence of operator matching. Recall that stabilizers and logical operators in the QL construction come from matching such operators on the smaller tensors. Since the tensor enumerator is counting the number of stabilizers with weight $d$ and a particular Pauli type on the open legs, tracing it with another tensor enumerator retains precisely the weight distribution of Pauli elements that are matching on the legs being glued. This in turn produces the desired weight distribution of the larger tensor network. Although Theorem \ref{thm:tensortrace1} provides a construction that is sufficient for building weight enumerator of any quantum code, the above intuition suggests that the reduced enumerators are sufficient for stabilizer codes, which allows us to reduce the bond dimension from $q^4$ to $q^2$. 

\begin{definition}
A diagonal trace is defined by
\begin{align}&\wedge_{j,k}^{DT} e_{E,\bar{E}} \\\nonumber
&= \left\{\begin{array}{cl} e_{E\setminus\{E_j,E_k\},\bar{E}\setminus\{\bar{E}_j,\bar{E}_k\}} & \text{if $E_j = E_k^* =\bar{E}_j = \bar{E}_k^*$}\\ 0 & \text{otherwise.}\end{array}\right.
\end{align}
\end{definition}

\begin{proposition}\label{prop:stab_dt}
Suppose
\begin{align}
    M_1&=\frac{1}{|S|}\sum_{S\in P\mathcal{S}} S\\
M_2&=\frac{1}{|S|}\sum_{S\in P\mathcal{S}}\omega_S S
\end{align} for any coset $P\mathcal{S}$ of Pauli operator $P$ and $\omega_S\in \mathbb{C}$. Let $\wedge_{\rm all}$ be the set of self-contractions that reduce an even rank tensor enumerator to a scalar, then 
\be\wedge_{\rm all}^{DT} \mathbf{A}^{J}(\mathbf{u};M_1,M_2)\propto A(\mathbf{u};\wedge_{\rm all} M_1,\wedge_{\rm all} M_2)\ee
and similarly for $\mathbf{B}$. The same holds if the forms of $M_1,M_2$ are switched.
\end{proposition}

\begin{hproof}
The proof is similar to Theorem 7.1 of \cite{CL2022}. Let us begin with the case where there are only two open legs in the tensor enumerator. It is clear that 
\begin{equation}\label{eqn:tracerel}
    \Tr[(G\otimes G^{*}\otimes F) M_i] \ne 0
\end{equation}
if and only if $G\otimes G^{*}\otimes F$ is a coset element.

Suppose $\mathcal{G}$ is the set of all $G$s for which the  trace (\ref{eqn:tracerel}) is nonzero, then 

\begin{align}\label{al:b15}
    \begin{split}&\sum_{G,\bar{G}\in\mathcal{E}}\Tr[(G\otimes G^{*}\otimes F) M_1]\Tr[(\bar{G}\otimes \bar{G}^{*}\otimes F)^{\dagger}  M_2] \\
    = &|\mathcal{G}|\sum_G\Tr[(G\otimes G^{*}\otimes F) M_1]\Tr[(G\otimes G^{*}\otimes F)^{\dagger} M_2]
    \end{split}
\end{align}
which is proportional to the diagonal trace $\wedge^{DT}$. We can see this by the following. For each $G\in \mathcal{G}$, we sum over $\bar{G}$, which leads to some constant $\propto\sum_S\omega_S$. Repeating for each $G$, we simply get back the same constant $|\mathcal{G}|$ times. If we only sum over the diagonal terms with $G=\bar{G}$, then we obtain the constant $\propto\sum_S\omega_S$ once. This only works because one of $M_1,M_2$ is an equal superposition of Pauli operators.

Furthermore, note that for the $F$ in (\ref{al:b15}), each $\bar{G}\in \mathcal{G}, \bar{G}=PG, P\ne I$, it is clear that $P\otimes P^{*}$ is a stabilizer of the code. Therefore, for any other $F$ such that $G\otimes G^{*}\otimes F\in P\mathcal{S}$, it must follow that $(P\otimes P^{*}\otimes I)(G\otimes G^{*}\otimes F)= \bar{G}\otimes \bar{G}^{*}\otimes F\in P\mathcal{S}$ for each $\bar{G}\in\mathcal{G}$. Therefore, the overcounting is identical for all $F$s by a factor of $|\mathcal{G}|$. Therefore the diagonal elements contain sufficient information to reproduce the scalar enumerator.

For any tensor enumerator with four open legs that needs two self-traces on two pairs $a_0$ and  $a_1$. From above arguments we know that a full trace on $a_1$ followed by diagonal trace on $a_0$ produces the correct scalar enumerator. Therefore it is sufficient to show that a diagonal trace on $a_1$ produce the correct diagonal elements for the pair $a_0$. Let $E$ denote the Pauli for open legs associated with pair $a_0$. Under a full trace on $a_1$, the diagonal elements of the remaining tensor then come from coefficients of the form 
\begin{align}
\sum_{G,\bar{G}\in\mathcal{E}}&\Tr[(G\otimes G^{*}\otimes E\otimes F) M_1]\\\nonumber
&\times\Tr[(\bar{G}\otimes \bar{G}^{*}\otimes E \otimes F)^{\dagger}  M_2]
\end{align}
where the sum comes from tracing over the legs of $a_1$. We notice that the same argument above applies by
setting $E\otimes F\rightarrow F$ since $F$ is arbitrary. Hence we conclude that the full trace on $a_1$ produce the same diagonal elements on $a_0$ as a diagonal trace up to a constant multiple. Proceed inductively with $2k$ open legs, it is clear that the diagonal components are sufficient for generating the scalar weight enumerators. 

To show that the B type enumerator is also correctly produced via diagonal trace, recall that diagonal trace is linear and commutes with the generalized Wigner transform as shown in the proof of Prop. VI.1, it can also be generated with only diagonal trace operations. 

\end{hproof}

Therefore, for practical analysis of Pauli stabilizer codes, we only need to consider the reduced tensor enumerators, that is, restricting to the diagonal elements $E=\bar{E}$ of each tensor enumerator in Definition 4.2 of \cite{CL2022}.

The same proof does not apply for tracing generalized enumerators $\tilde{\mathbf{A}},\tilde{\mathbf{B}}$ because two separate sums are required separately for $F$ and $\bar{F}$ whereas the argument is valid only when $F=\bar{F}$. Therefore we have to perform a full tensor trace even for stabilizer codes. Such is needed to analyze more general error channels like coherent noise.

\section{Tensor-only Implementation}\label{app:tensoronly}
\subsection{Multi-linear formulation}
Although the enumerator polynomials can be implemented symbolically, it is also possible to rephrase them purely as tensors with complex coefficients. 

For each tensor enumerator, one can take the coefficients in the polynomial, as a tensorial object by itself. For example, for Shor-Laflamme enumerators, coefficient $A_d$ in the scalar enumerator can be treated as a rank-1 tensor with bond dimension $\leq n+1$. Similarly, $A^j_d$ in vector enumerator has two indices where $j$ marks a bond dimension $q^4$ index (or $q^2$ in reduced enumerators) and $d$ has bond dimension $\leq n+1$. Generally, a(n) (abstract) tensor enumerators can be represented by the tensor components
$$
    A^{j_1,\dots,j_i}_{d_{\mathbf{u}}} ~\mathrm{and}~ 
    B^{j_1,\dots,j_i}_{d_{\mathbf{u}}},
$$
where $d_{\mathbf{u}}$ can be $l$-tuple of indices that tracks the powers of the monomials. For instance, for complete enumerator, $d_{\mathbf{u}}\rightarrow (d_x,d_y,d_z)$. For now, let us focus on reduced enumerators over $q=2$ where the upper and lower indices carry no additional physical meaning. To avoid clutter, let us also drop the subscript of $d_{\mathbf{u}}$. For concreteness one can take $d$ to be the usual operator weight, but it is straightforward to restore it to the most general form.

In a tensor network, instead of tracing the tensor enumerator polynomials, we now trace together these tensors. However, we need an additional operation on the two legs $d_1,d_2$ that add the powers of the monomials during polynomial multiplication. 

\begin{align}\label{al:tensortrace}
\begin{split}
  &A^{j_{l+1},\dots,j_i,r_{l+1},r_k}_d\\
  &=\sum_{d_1,d_2}^{n_1,n_2}M_{d}^{d_1d_2}\sum_{j_1,j_2,\dots, j_l}A^{j_1,j_2,\dots, j_l,\dots j_i}_{d_1}A^{j_1,j_2,\dots j_l, r_{l+1}\dots r_k}_{d_2}
  \end{split}
\end{align}
where $M_{d}^{d_1d_2}$ is a tensor such that %the bond dimension for index $d$ satisfies $\chi_d\leq \min \{n+1,n_1+n_2+1\}$, and 
\begin{equation}
M_{d}^{d_1d_2}=
    \begin{cases}
     1~\textrm{if}~d=d_1+d_2\\
     0 ~\textrm{else}.
    \end{cases}
\end{equation}
On the formalism level, this trace with $M$ tensor can be completed at any time. In practice, however, we perform such an operation every time a tensor trace like (\ref{al:tensortrace}) is completed. 

The modified trace operation $\tilde{\Tr}$ that reduces the tensor rank  can also be performed by contracting another rank-3 tensor
\begin{equation}
    T_{d'dj}= \begin{cases}
    \delta_{d'd}~~~\mathrm{if}~j=0\\
    \delta_{d'd-1}~~\mathrm{else}.
    \end{cases}
\end{equation}
\begin{figure}
    \centering
    \includegraphics[width=0.4\linewidth]{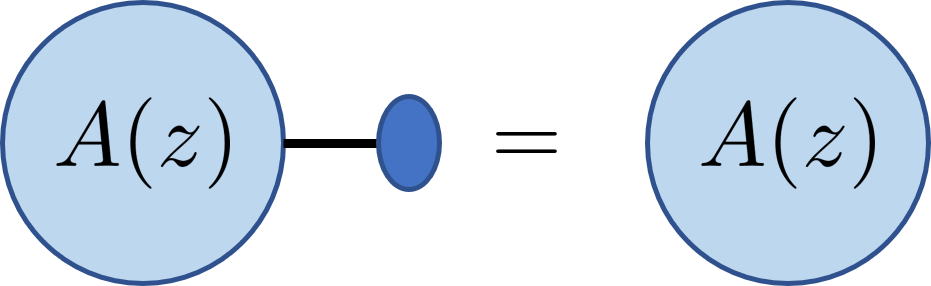}
    \caption{Modified trace operation.}
    \label{fig:vwep_bot}
\end{figure}

For example, to recover the scalar enumerator from a vector enumerator (Fig.~\ref{fig:vwep_bot}), we use $A_{d'}=T_{d'dj}A^{j}_d$, where repeated indices are summed over.

The method for tracing other tensor enumerators, such as the double and complete tensor weight enumerators, is largely the same. %The only difference is that we replace each of the tensor weight enumerators with their double or complete versions that depend on $x,y,z,w$ instead of a single variable $z$ or the pair $z,w$ in the balanced form.

For example, the contraction of two reduced double enumerator is

\begin{align}
  \begin{split}&D^{j_{l+1},\dots,j_i,r_{l+1},r_k}_{d^x,d^z}\\
  =&M_{d^x}^{d_1^xd_2^x}M_{d^z}^{d_1^zd_2^z}D^{j_1,j_2,\dots, j_l,\dots j_i}_{d_1^x,d_1^z}D^{j_1,j_2,\dots j_l, r_{l+1}\dots r_k}_{d_2^x,d_2^z}
  \end{split}
\end{align}
where repeated indices are summed. 
The modified trace is 
\begin{equation}
    D_{d_x',d_z'}=\sum_{d_x,d_z,j} T_{d'_x,d'_z,d_x,d_z,j}D_{d_x,d_z}^j
\end{equation}
with 
\begin{equation}
    T_{d'_xd'_zd_{x}d_zj}= \begin{cases}
    \delta_{d_x'd_x}\delta_{d'_zd_z}~~~&\mathrm{if}~j=0\\
    \delta_{d'_x-1d_x}\delta_{d'_zd_z}~~&\mathrm{if}~j=1\\
    \delta_{d'_x-1d_x}\delta_{d'_z-1d_z}~~&\mathrm{if}~j=2\\
    \delta_{d'_xd_x}\delta_{d'_z-1d_z}~~&\mathrm{if}~j=3\\
    \end{cases}.
\end{equation}

If $\mathbf{u}$ carries more variables, then an additional $M$ contraction is needed for each separate variable index\footnote{In the Matlab code implementation, $j=3$ and $j=2$ are swapped in the indexing convention such that $Y$ is mapped to the last index and $Z$ is mapped to the second last.}.  

For the full tensor enumerator polynomial, one has to take extra care of potential sign changes where we have $I\leftrightarrow I, X\leftrightarrow X,Z\leftrightarrow Z$ but $-Y\leftrightarrow Y$ matchings. Suppressing the $d$ index for now, we can think of each tensor enumerator index in a representation $A^{j}\rightarrow A^{\alpha,\bar{\alpha}}$ with $\alpha = 1,\dots, q^2=4$. Furthermore, we prepare the Minkowski metric $\eta_{\alpha,\beta}=diag(1,1,-1,1)$  so that tensor contractions are only performed between upper and lower indices. Indices are raised and lowered in the usual way with $A_{\beta\bar{\beta}}=\eta_{\alpha\beta}\eta_{\bar{\alpha}\bar{\beta}}A^{\alpha\bar{\alpha}}$ and contracting the two vector enumerators is by contracting the covariant vector with the contravariant one, i.e., $A^{\alpha\bar{\alpha}}A'_{\alpha\bar{\alpha}}$. We see the raised or lowered index does not matter for reduced enumerators because the diagonal elements for $\eta_{\alpha\beta}\eta_{\bar{\alpha}\bar{\beta}}$ at $\alpha=\bar{\alpha}, \beta=\bar{\beta}$ only carry positive signs.

\subsection{MacWilliams Identity as a linear transformation}

We derive the matrix representation of MacWilliams identities in the polynomial basis $\{z^dw^{n-d} : 0 \leq d  \leq n\}$ to facilitate Matlab numerics.

By Corollary 5 of \cite{rains1},
\begin{align}\begin{split}
A^{\prime}(w, z) = A((w  + z)/q, z/q), \\
B^{\prime}(w, z) = B((w  + z)/q, z/q).
\end{split}
\end{align}
By Theorem 3 of \cite{rains1}, $A^{\prime}(w, z) = B^{\prime}(z, w)$ and  $A^{\prime}_{d} = B^{\prime}_{n-d}$, which is equivalent to the quantum MacWilliams identity (Theorem 7 of \cite{rains1}):
\begin{align}
B(w, z) = A\left(\tfrac{w + (q^2 -  1)z}{q}, \tfrac{w-z}{q}\right).
\end{align}
We chose to work with $A^{\prime}$ and $B^{\prime}$ because the relation $A^{\prime}_{d} = B^{\prime}_{n-d}$ can be easily expressed by an anti-diagonal matrix with every element equal to 1 in the polynomial basis $\{z^dw^{n-d} : 0 \leq d  \leq n\}$.

To express $B_{d}$ in terms of $A_{d}$, we only need to express $A_{d}$ in terms of $A^{\prime}_{d}$, and $B^{\prime}_{d}$ in terms of $B_{d}$. By Corollary 5 of \cite{rains1},
\begin{align}\begin{split}
&A^{\prime}(w, z) \\
&= \sum_{d = 0}^{n} A_{d}\left(\frac{z}{q}\right)^{d} \left(w + \frac{z}{q}\right)^{n-d} \\
&= \sum_{d = 0}^{n} A_{d} \left(\frac{z}{q}\right)^{d} \left( \sum_{e = 0}^{n - d} {n-d \choose e} \left(\frac{z}{q}\right)^{n-d-e} w^{e} \right) \\
&= \sum_{d = 0}^{n} \sum_{e = 0}^{n - d} A_{d} {n-d \choose e} \left(\frac{z}{q}\right)^{d}  \left(\frac{z}{q}\right)^{n-d-e} w^{e} \\
&= \sum_{d = 0}^{n} \sum_{e = 0}^{n - d} A_{d} {n-d \choose e} \left(\frac{z}{q}\right)^{n-e} w^{e}  \\
&= \sum_{e = 0}^{n} \sum_{d = 0}^{n-e} A_{d} {n-d \choose e} \left(\frac{z}{q}\right)^{n-e} w^{e} \\
&= \sum_{e = 0}^{n} \left( \sum_{d = 0}^{n-e} A_{d} {n-d \choose e} \right) \left(\frac{z}{q}\right)^{n-e} w^{e} \\
&= \sum_{e^{\prime} = 0}^{n} \left( \sum_{d = 0}^{e^{\prime}} A_{d} {n-d \choose n-e^{\prime}} \right) \left(\frac{z}{q}\right)^{e^{\prime}} w^{n-e^{\prime}} \\
&= \sum_{d = 0}^{n} \left( \frac{1}{q^d} \sum_{m = 0}^{d} A_{m} {n-m \choose n-d} \right) z^{d} w^{n-d}.
\end{split}
\end{align}
Since
\begin{align}
A^{\prime}(w, z) = \sum_{d = 0}^{n} A^{\prime}_{d} z^{d} w^{n-d},
\end{align}
we have
\begin{align}
A^{\prime}_{d} = \frac{1}{q^d} \sum_{m = 0}^{d} A_{m} {n-m \choose n-d}, 0 \leq d \leq  n.
\end{align}
In other words,
\begin{align}
A^{\prime}_{d} = \sum_{m = 0}^{d} T_{dm} A_{m},
\end{align}
where
\begin{align}
T_{dm} = \frac{1}{q^d} {n-m \choose n-d}, 0 \leq m \leq  d, 0 \leq d \leq  n.
\end{align}
Similarly, $B^{\prime}(w, z) = B(w  + z/q, z/q)$ implies that
\begin{align}
B^{\prime}_{d} = \sum_{m = 0}^{d} T_{dm} B_{m}.
\end{align}
Hence
\begin{align}
A_{d} &= \sum_{d^{\prime} = 0}^{n} (T^{-1} J T)_{dd^{\prime}} B_{d^{\prime}}, 0 \leq d, d^{\prime} \leq  n,
\end{align}
where 
\begin{align*}
J = \begin{pmatrix}
    0 & \cdots & 0 & 1 \\
    0 & \cdots & 1 & 0 \\
    \vdots & \ddots & \vdots & \vdots \\
    1 & \cdots & 0 & 0
    \end{pmatrix}_{(n+1) \times (n+1)} \\
T = \begin{pmatrix}
    {n \choose n} & 0 & \cdots & 0 \\
    \frac{1}{q} {n \choose n-1} &  \frac{1}{q} {n-1 \choose n-1} & \cdots & 0 \\
    \vdots & \ddots & \vdots & \vdots \\
      \frac{1}{q^n} {n \choose 0} &  \frac{1}{q^n} {n-1 \choose 0} & \cdots &  \frac{1}{q^n} {0 \choose 0}
    \end{pmatrix}_{(n+1) \times (n+1)}
\end{align*}
Similarly,
\begin{align}
B_{d} &= \sum_{d^{\prime} = 0}^{n} (T^{-1} J T)_{dd^{\prime}} A_{d^{\prime}}, 0 \leq d, d^{\prime} \leq  n,
\end{align}
because
\begin{align}
A = T^{-1}JTB \Longleftrightarrow B = T^{-1}JTA.
\end{align}

\subsection{Connection with Farrelly, Tuckett and Stace}\label{subapp:LTNC}
Ref. \cite{LTNC} proposed a method to compute distance in local tensor network codes, which are qubit stabilizer codes obtained from contracting other smaller stabilizer codes in a manner similar to \cite{CL2021}. In particular, we see that when applied to an $[[n,k]]$ stabilizer code, the tensor $C_{w}^{l_1,\dots,l_k}$ in \cite{LTNC} is exactly a reduced tensor enumerator in the multi-linear form  where $w$ is precisely the degree of the monomial and $l_i=0,1,2,3$ are the open indices that track the Pauli type $I,X,Y,Z$. This corresponds to computing the tensor weight enumerator $B(z)^{l_1,\dots,l_k}$ if we keep all the logical legs open. Similarly, $C_{w}^{0,\dots,0}$ is the coefficient of an $A$-type tensor enumerator. Indeed, $D_w$ which is obtainable from $C_{w}^{l_1,\dots,l_k}$ is precisely the tensor coefficients of the scalar enumerators $B_d-A_d$. Both enumerators are in the usual Shor-Laflamme form.

Although both approaches rely on the tensor network method to produce weight distributions, the detailed construction differs somewhat in how the tensors in the network is implemented --- we construct the enumerator from encoding maps while \cite{LTNC} directly enumerate the logical operators of an encoding tensor and then compute their weights by contracting with another weight tensor that tabulates $4^n$ Pauli weights.  
Ref. \cite{LTNC} also produces a tensor network (See Fig 4d) with a double bond on each contraction where each edge has bond dimension $4$. Naively this appears to lead to bond dimension $16$ objects (Fig 4b,c). While such a description is sufficient, we know from the tensor enumerator formalism that it is possible to obtain the enumerator for Pauli stabilizer codes with a reduced bond dimension $4$ (Thm~\ref{thm:refined_dist}), hence enabling more efficient tensor contractions.

It is unclear how the complexity estimates for these methods compare, as none was performed for \cite{LTNC} except for 1d codes that are prepared by log depth circuits. However, given the similarities in their structure and their overall efficiency for tree tensor networks and holographic codes, they should be polynomially equivalent in that regime. In practice, however, we note that even a constant factor difference can be quite substantial. Therefore, a more in-depth comparison in their performance can be an interesting problem for future work. In particular, it is important to understand whether these methods are optimal with respect to different networks.

Although Ref.\cite{LTNC} does not discuss weight distributions for other error models, it is possible to adapt their formalism to produce double and complete weight enumerators by modifying their weight tensor. For example, to obtain the double enumerator, one can replace $W_w^{g_1,\dots, g_m}$ with $W_{w_x,w_z}^{g_1,\dots, g_m}$ such that the tensor coefficient is unity when a Pauli string $\sigma^{g_1}\otimes \dots\otimes \sigma^{g_m}$ has X and Z weights $w_x,w_z$ respectively. A similar extension should be possible for abstract enumerators. It is currently unclear whether a similar extension is possible for generalized abstract enumerators.

Another key difference lies in the use of MacWilliams transform in our work, which is polynomial\footnote{The complexity roughly scales as $O(n^3)$ from matrix multiplication. However, this is not counting the cost needed to manipulate large integers.}  in $n$. Generally, the MacWilliams identity can help reduce computational cost. When the tensor network is efficiently contractible, and when we overlook the cost in manipulating large integers, the difference of keeping $B$ vs $A$ type tensor enumerators should be relatively insignificant. However, when the minimal cuts are large, e.g. when the tensor network represents a volume law or even some area law states, the $B$ type tensor can become far more populated than the $A$ type by as much as $O(e^k)$. For instance, in the limiting case where the cost of tensor network contraction approaches the that of the brute force method\footnote{Intuitively, computing $B$ enumerators of an $[[n,k]]$ stabilizer code involves enumerating $2^{n+k}$ elements instead of $2^{n-k}$. }, e.g. in random stabilizer codes, we see that $B$ is $2^{2k}$ more expensive to compute compared to $A$. Hence in some instances, the MacWilliams identities can help reduce computational cost that is exponential in $k$.

The decoder \cite{LTNC} uses is formally similar to the usual tensor network decoder where error probability is computed for some fixed $p$ using a tensor contraction whereas the enumerator method produces an analytical expression for the error probability as a function of $\mathbf{u}$. The former is computationally advantageous when the error probabilities are heavily inhomogeneous and carry a strong locational dependence. The latter is more powerful for obtaining a continuous range of error probabilities when the physical errors are relatively uniform across the entire system.

Overall, the formalism based on tensor weight enumerator is more general as it applies to all quantum codes with uniform local dimensions. When specialized to the case of Pauli stabilizer codes over qubits, both methods can compute scalar and tensor enumerators associated with the code using tensor network methods. In this case, our method improves upon \cite{LTNC} with a reduced bond dimension and with the use of the MacWilliams identities. The former provides a polynomial speed up while the latter can provide an $O(e^k)$ speed up in some regimes. With the extension to biased error and general noise models, we also extend the maximum likelihood decoders for such stabilizer codes to general error channels. However, the enumerator method is less efficient in tackling highly inhomogeneous errors.
 
\section{Error detection for general noise channels}\label{app:errorchannel}
\subsection{Non-detectable error}
\textit{Proof for Theorem \ref{thm:gen_err}}: Let us compute here the probability of incurring a non-correctable (logical) error. Suppose the error channel is $\mathcal{E}(\rho)$, which can be written as the Kraus form in Theorem \ref{thm:gen_err} being the tensor product of single site errors. 
Let the initial state be $\rho=|\tilde{\psi}\rangle\langle \tilde{\psi}|\in L(\mathcal{C})$, and $\dim \mathcal{C}=K$. Then the probability of a non-detectable error is 
\begin{align}
    p_{nd}(\rho) &= \Tr[(\Pi-\rho)\Pi \mathcal{E}(\rho)\Pi ] \\\nonumber
    &= \sum_{\mathbf{i}}||(I-|\tilde{\psi}\rangle\langle\tilde{\psi}|)\Pi \mathcal{K}_{\mathbf{i}}|\tilde{\psi}\rangle||^2,
\end{align}
where $\Pi$ is the projection onto the code subspace. It is simply the overlap between the error state and the part of the code subspace that is orthogonal to the original codeword. 

Now averaging over all initial codewords $|\tilde{\psi}\rangle$ with respect to the normalized uniform measure $\mu(|\tilde{\psi}\rangle)$, we have 

\begin{align}\label{al:Dintegral}
\begin{split}
    {p}_{nd} &=\sum_{\mathbf{i}} \int_{\tilde{|\psi\rangle}\in\mathcal{C}}||(I-|\tilde{\psi}\rangle\langle\tilde{\psi}|)\Pi \mathcal{K}_{\mathbf{i}} |\tilde{\psi}\rangle||^2 d\mu(|\tilde{\psi}\rangle)\\
    &=\sum_{\mathbf{i}}\int_{\tilde{|\psi\rangle}\in\mathcal{C}} \langle\tilde{\psi}|\mathcal{K}_{\mathbf{i}}^{\dagger}\Pi(I-|\tilde{\psi}\rangle\langle\tilde{\psi}|)\\
    &\quad \times(I-|\tilde{\psi}\rangle\langle\tilde{\psi}|)\Pi \mathcal{K}_{\mathbf{i}}|\tilde{\psi}\rangle d\mu(\tilde{\psi})\\
    &=\sum_{\mathbf{i}}\int_{\tilde{|\psi\rangle}\in\mathcal{C}} \langle\tilde{\psi}|\mathcal{K}_{\mathbf{i}}^{\dagger}\Pi \mathcal{K}_{\mathbf{i}}|\tilde{\psi}\rangle d\mu(|\tilde{\psi}\rangle)\\
    &\quad -\sum_{\mathbf{i}} \int_{\tilde{|\psi\rangle}\in\mathcal{C}} \langle\tilde{\psi}|\mathcal{K}_{\mathbf{i}}^{\dagger}\Pi|\tilde{\psi}\rangle\langle\tilde{\psi}|\Pi \mathcal{K}_{\mathbf{i}}|\tilde{\psi}\rangle d\mu(|\tilde{\psi}\rangle)
    \end{split}
\end{align}
Similar to \cite{err_det}, the integral in (\ref{al:Dintegral}) can be evaluated.  
The first term is
\begin{align}
    \begin{split}&\sum_{\mathbf{i}}\int_{\tilde{|\psi\rangle}\in\mathcal{C}} \langle\tilde{\psi}|\mathcal{K}_{\mathbf{i}}^{\dagger}\Pi\mathcal{K}_{\mathbf{i}} |\tilde{\psi}\rangle d\mu(|\tilde{\psi}\rangle)\\
    &\quad = \sum_{\mathbf{i}}\Tr[\mathcal{K}_{\mathbf{i}}^{\dagger}\Pi \mathcal{K}_{\mathbf{i}}\int_{\tilde{|\psi\rangle}\in\mathcal{C}}|\tilde{\psi}\rangle\langle\tilde{\psi}|dv(|\tilde{\psi}\rangle)]  \\
    &\quad = \frac 1 K\sum_{\mathbf{i}}\Tr[\mathcal{K}_{\mathbf{i}}^{\dagger}\Pi \mathcal{K}_{\mathbf{i}}\Pi],
    \end{split}
\end{align}
where we use Lemma 7 in \cite{err_det} for the last step, which evaluates the integral.

The second term is
\begin{align}
    \begin{split}
&\sum_{\mathbf{i}}\int_{\tilde{|\psi\rangle}\in\mathcal{C}} \langle\tilde{\psi}|\mathcal{K}_{\mathbf{i}}^{\dagger}\Pi|\tilde{\psi}\rangle\langle\tilde{\psi}|\Pi \mathcal{K}_{\mathbf{i}}|\tilde{\psi}\rangle d\mu(|\tilde{\psi}\rangle)\\
    &\quad=\frac{1}{K(K+1)}(\sum_{\mathbf{i}}\Tr[ \mathcal{K}_{\mathbf{i}}^{\dagger}\Pi \mathcal{K}_{\mathbf{i}}\Pi]\\
    &\qquad +\sum_{\mathbf{i}}\Tr[\mathcal{K}_{\mathbf{i}}^{\dagger}\Pi]\Tr[\mathcal{K}_{\mathbf{i}}\Pi])
    \end{split}
\end{align}
where we integrate over the same measure and use Lemma 8 in \cite{err_det}. This completes our proof for Theorem \ref{thm:gen_err}.

\subsection{Errors with non-trivial syndromes}
\textit{Proof of Theorem \ref{thm:coset} and decoder}: When $\mathcal{C}$ is a stabilizer code, we can talk about syndrome measurements and decoding in the usual sense. While $\Pi$ denotes the projection onto the code subspace, i.e., measuring trivial syndromes, we can similarly ask what the probability is for measuring some other syndromes $s$ where the state is taken to a subspace $\Pi_s = E_s \Pi E_s^{\dagger}$, where $E_s$ is an error with syndrome $s$.

Again, this is given by the overlap between the state suffering from the error and some final state in the error subspace. 

\begin{align}
\begin{split}
    \bar{p}_{s}& =\sum_{\mathbf{i}} \int_{|\tilde{\psi}\rangle\in \mathcal{C}} d\mu(|\tilde{\psi}\rangle)\Tr[\Pi_s \mathcal{K}_{\mathbf{i}}\rho \mathcal{K}_{\mathbf{i}}^{\dagger} \Pi_s]\\
    &= \sum_{\mathbf{i}}\int_{|\tilde{\psi}\rangle\in \mathcal{C}} d\mu(|\tilde{\psi}\rangle)\Tr[E_s\Pi E_s^{\dagger} \mathcal{K}_{\mathbf{i}}\rho \mathcal{K}_{\mathbf{i}}^{\dagger} E_s\Pi E_s^{\dagger}]\\
    &=  \sum_{\mathbf{i}}\Tr[\int_{|\tilde{\psi}\rangle\in \mathcal{C}} d\mu(|\tilde{\psi}\rangle) |\tilde{\psi}\rangle\langle\tilde{\psi}|\mathcal{K}_{\mathbf{i}}^{\dagger}\Pi_s \mathcal{K}_{\mathbf{i}}]\\
    &= \frac 1 K \sum_{\mathbf{i}}\Tr[\Pi \mathcal{K}_{\mathbf{i}}^{\dagger} \Pi_s \mathcal{K}_{\mathbf{i}}].
    \end{split}
\end{align}
Therefore, this quantity can be easily obtained from the $B$-type generalized complete enumerator when we replace one of $\Pi$ by $\Pi_s$. Note that this recovers the syndrome probability with Pauli errors using the coset enumerator. 

It is also useful for decoding purposes to consider the probability $p(\bar{L}|s)$ so as to correct the most likely logical error. 
The probability that $\tilde{E}_s=E_s\tilde{L}$ occurs is
\begin{align}
    \bar{p}(\tilde{L}\cap s)& =   \sum_{\mathbf{i}}\int_{|\tilde{\psi}\rangle\in \mathcal{C}} d\mu(|\tilde{\psi}\rangle)\\\nonumber
    &\times\Tr[\tilde{E}_s|\tilde{\psi}\rangle\langle\tilde{\psi}|\tilde{E}_s^{\dagger}\Pi_s \mathcal{K}_{\mathbf{i}}|\tilde{\psi}\rangle\langle\tilde{\psi}|\mathcal{K}_{\mathbf{i}}^{\dagger}\Pi_s]\\\nonumber
    &= \sum_{\mathbf{i}} \int_{|\tilde{\psi}\rangle\in \mathcal{C}} d\mu(|\tilde{\psi}\rangle)\\\nonumber
    &\times(\langle\tilde{\psi}|\tilde{E}_s^{\dagger}\Pi_s \mathcal{K}_{\mathbf{i}}|\tilde{\psi}\rangle)(\langle \tilde{\psi}|\mathcal{K}_{\mathbf{i}}^{\dagger}\Pi_s\tilde{E}_s|\tilde{\psi}\rangle)
\end{align}
Because $\Pi_s\tilde{E}_s|\tilde{\psi}\rangle=\tilde{E}_s|\tilde{\psi}\rangle$.
\begin{align}
    &\bar{p}(\tilde{L}\cap s)\\\nonumber
    = &\sum_{\mathbf{i}} \int_{|\tilde{\psi}\rangle\in \mathcal{C}} d\mu(|\tilde{\psi}\rangle) (\langle\tilde{\psi}|\tilde{E}_{s}^{\dagger}\mathcal{K}_{\mathbf{i}}|\tilde{\psi}\rangle)(\langle \tilde{\psi}|\mathcal{K}_{\mathbf{i}}^{\dagger}\tilde{E}_{s}|\tilde{\psi}\rangle)\\\nonumber
    =&\frac{1}{K(K+1)}\Big(\sum_{\mathbf{i}}\Tr[\tilde{E}_s^{\dagger}\mathcal{K}_{\mathbf{i}} \Pi \mathcal{K}_{\mathbf{i}}^{\dagger} \tilde{E}_s\Pi] \\\nonumber
    &+ \sum_{\mathbf{i}}\Tr[\tilde{E}_s^{\dagger} \mathcal{K}_{\mathbf{i}}\Pi] \Tr[\mathcal{K}_{\mathbf{i}}^{\dagger}\tilde{E}_s\Pi]\Big)\\\nonumber
    =&\frac{1}{K(K+1)}\Big(\sum_{\mathbf{i}}\Tr[\mathcal{K}_{\mathbf{i}} \Pi \mathcal{K}_{\mathbf{i}}^{\dagger} \Pi_s] \\\nonumber
    &+ \sum_{\mathbf{i}}\Tr[\mathcal{K}_{\mathbf{i}}\Pi\tilde{E}_s^{\dagger} ] \Tr[\mathcal{K}_{\mathbf{i}}^{\dagger}\tilde{E}_s\Pi]\Big).
\end{align}

Hence 
\begin{align}    \label{al:log_prob}
    \bar{p}(\tilde{L}|s) &=\bar{p}(\tilde{L}\cap s)/\bar{p}_s \\\nonumber
    &= \frac {1}{K+1} (1+ \frac{\sum_{\mathbf{i}}\Tr[\mathcal{K}_{\mathbf{i}}\Pi\tilde{E}_s^{\dagger} ]\Tr[\mathcal{K}_{\mathbf{i}}^{\dagger}\tilde{E}_s\Pi]}{\sum_{\mathbf{i}}\Tr[\mathcal{K}_{\mathbf{i}} \Pi \mathcal{K}_{\mathbf{i}}^{\dagger} \Pi_s]}).
\end{align}
Each term in (\ref{al:log_prob}) can be computed by setting $M_1, M_2$ to the appropriate values in the weight enumerator $M_1=\Pi\tilde{E}_s^{\dagger}, M_2=\tilde{E}_s\Pi$ for the $\bar{A}$-type enumerator and $M_1=\Pi, M_2=\Pi_s$ for the $\bar{B}$-type enumerator. 

For the purpose of building a decoder, we do not care about the overall normalization, hence computing the $\bar{A}$-type enumerator will be sufficient. The first term in $\bar{p}(\tilde{L}\cap s)$ is independent on the logical operation $\tilde{L}$, and thus does not modify our decision based on the maximum likelihood.

\subsection{General Logical Error Channel}\label{subapp:gen_LEC}
\textit{Non-unitary logical error:} Under this more general channel, it is also natural to consider a more general logical error where for some $\rho_s$ in the error subspace with syndrome $s$ beyond the kind of coherent logical error $\tilde{L}$. For instance, we can discuss the error probability that the logical information suffers from a logical error channel in that subspace 
\begin{equation}
    \tilde{\mathcal{N}}(\tilde{\rho})=\tilde{\rho}\rightarrow \sum_j \tilde{\eta}_j \tilde{\rho}\tilde{\eta}_j^{\dagger},
\end{equation}
after obtaining syndrome $s$ by measuring the checks. More precisely, we find that
\begin{align}
    &\bar{p}(\tilde{\mathcal{N}}(\cdot) \cap s) \\\nonumber
    =&\sum_{\mathbf{i}}\int_{|\tilde{\psi}\rangle\in\mathcal{C}} d\mu(|\tilde{\psi}\rangle)\Tr[E_s\tilde{\mathcal{N}}(\tilde{\rho})E_s^{\dagger} \Pi_s \mathcal{K}_{\mathbf{i}}|\tilde{\psi}\rangle\langle\tilde{\psi}|\mathcal{K}_{\mathbf{i}}^{\dagger}\Pi_s]\\\nonumber
   =& \sum_{\mathbf{i},j}\int_{|\tilde{\psi}\rangle\in\mathcal{C}} d\mu(|\tilde{\psi}\rangle)\Tr[E_s\tilde{\eta}_j|\tilde{\psi}\rangle\langle\tilde{\psi}| \tilde{\eta}^{\dagger}_jE_s^{\dagger}  \mathcal{K}_{\mathbf{i}}|\tilde{\psi}\rangle\langle\tilde{\psi}|\mathcal{K}_{\mathbf{i}}^{\dagger}]\\\nonumber
   =& \sum_{\mathbf{i},j}\int_{|\tilde{\psi}\rangle\in\mathcal{C}} d\mu(|\tilde{\psi}\rangle)\Tr[O_{j\mathbf{i}}^s|\tilde{\psi}\rangle\langle\tilde{\psi}| O_{j\mathbf{i}}^{s\dagger}  |\tilde{\psi}\rangle\langle\tilde{\psi}|]\\\nonumber
   =&\frac{1}{K(K+1)}\sum_{\mathbf{i},j}\Big(\Tr[O_{j,\mathbf{i}}^{s\dagger}\Pi O^s_{j,\mathbf{i}}\Pi]\\\nonumber
   &+\Tr[O_{j,\mathbf{i}}^{s\dagger}\Pi]\Tr[O^s_{j,\mathbf{i}}\Pi]\Big)
\end{align}
where we defined $O^s_{j,\mathbf{i}}=\mathcal{K}_{\mathbf{i}}^{\dagger}E_s\tilde{\eta}_j$.
Note that just like for calculating error probability of syndrome $s$ under depolarizing noise with $B$ type enumerators, the first term can be expressed as a B type and one can perform the sum over $j$ in defining 

\begin{equation}
\Pi^s_{\eta}=\sum_j E_s\tilde{\eta}_j\Pi \tilde{\eta}_j^{\dagger} E_s^{\dagger}\end{equation}

and then substitute and compute the first term as 
\begin{equation}
    \sum_{\mathbf{i}}\Tr[\mathcal{K}_{\mathbf{i}}\Pi \mathcal{K}_{\mathbf{i}}^{\dagger}\Pi^s_{\eta}],
\end{equation}
which is basically identical to our computation of the non-detectable error probability except we set $M_2=\Pi_{\eta}^s$ and the remaining procedures for decomposing $\mathcal{K}_{\mathbf{i}}$ carries over identically.

For the second term, however, we have to repeat the enumerator computations for each $j$ by summing over $\mathbf{i}$. If we set $\tilde{E}_s^j=E_s\tilde{\eta}_j$, then it has the identical form as the coset enumerator we analyzed earlier except for the j dependence,

\begin{equation}
\Tr[\mathcal{K}_{\mathbf{i}}\Pi\tilde{E}_s^{j\dagger}]\Tr[\mathcal{K}_{\mathbf{i}}^{\dagger}\tilde{E}_s^{j}\Pi].
\end{equation} For generic errors, $j=1,\dots,4^k$, so it is more relevant for $k$ small.  

For a fixed error channel, the enumerator fully captures the likelihood of all error channels by decomposing $\tilde{\eta}_j$ into Paulis and varying their coefficients. It could be interesting to analyze the extrema of error probabilities with respect to these variables to find the most likely error channel. We can imagine building a decoder that seeks to undo the effect of the most likely error channel, though it is unclear when such recovery procedures exist in general.  
To correct such errors, we apply first a coset element of $E_s$. Then depending on the availability of the recovery map given the logical error, we (partially) reverse the effect of the logical error channel based on the relevant information of the code and syndrome measurement outcomes. 

\section{Quantum Tanner Graph from Quantum Lego}\label{app:qtanner}
{\color{black}
Given any $[[n,k]]$ stabilizer code whose codewords can be obtained from measuring the check operators and post-selecting on the trivial syndrome outcome, one can express the encoding or state preparation process as a tensor network. 

Consider a measurement-based state preparation process where we entangle all physical qubits with a reference using Bell pairs $|00\rangle+|11\rangle$. Then to apply the checks by measuring them. The measurement process is straightforward --- the physical qubits on which the check has support is entangled with a ancillary qubit using the usual circuit. Then one measures and post-selects on the trivial syndrome. Although the actual preparation of such a state in a quantum computer requires either adaptive measurements, decoding, and/or post-selection, thus rendering the actual process much more complicated, there is no such obstacle in the classical tensor network description where post-selection simply corresponds to contraction of a particular type of tensor. %Because we are considering the tensor network description, the processes involved need not be unitary. 

\begin{figure}
    \centering
    \includegraphics[width=\linewidth]{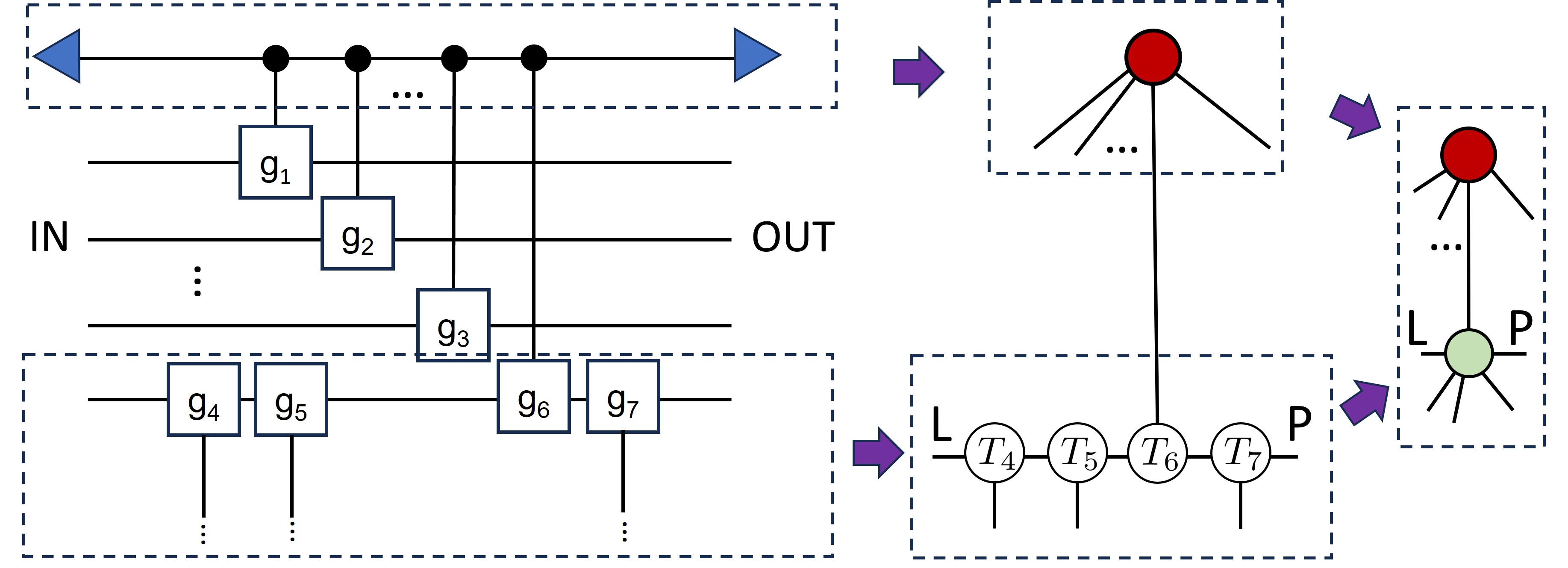}
    \caption{An ancillary is prepared and projected onto $|+\rangle$ (blue triangle). This condenses the ancillary into a check node tensor that is simply the encoding tensor of a repetition code (Z spider). The data qubit condenses into another tensor node (green) by combining the degree 3 tensors $T_i$. Here IN or L labels the input/logical degrees of freedom while OUT or P labels the output/physical degrees of freedom in the final atomic code it produces. The bottom row indicates that the qubit is checked by other checks e.g. $g_i, i=4\sim 7$.}
    \label{fig:mbsp}
\end{figure}

Suppose each check acts with a unitary $g_i$ on the physical qubit (bottom wire of Fig.~\ref{fig:mbsp}), then the action of this gate on the wire can be converted into a tensor (Fig.~\ref{fig:gate_tensor}a). The resulting tensor for common gates used in the preparation process for XP stabilizer codes are also given in Fig.~\ref{fig:gate_tensor}b. Generally, this conversion can be performed for any two qubit controlled gate by choosing the appropriate tensor $T$ in purple.

\begin{figure}
    \centering\includegraphics[width=0.9\linewidth]{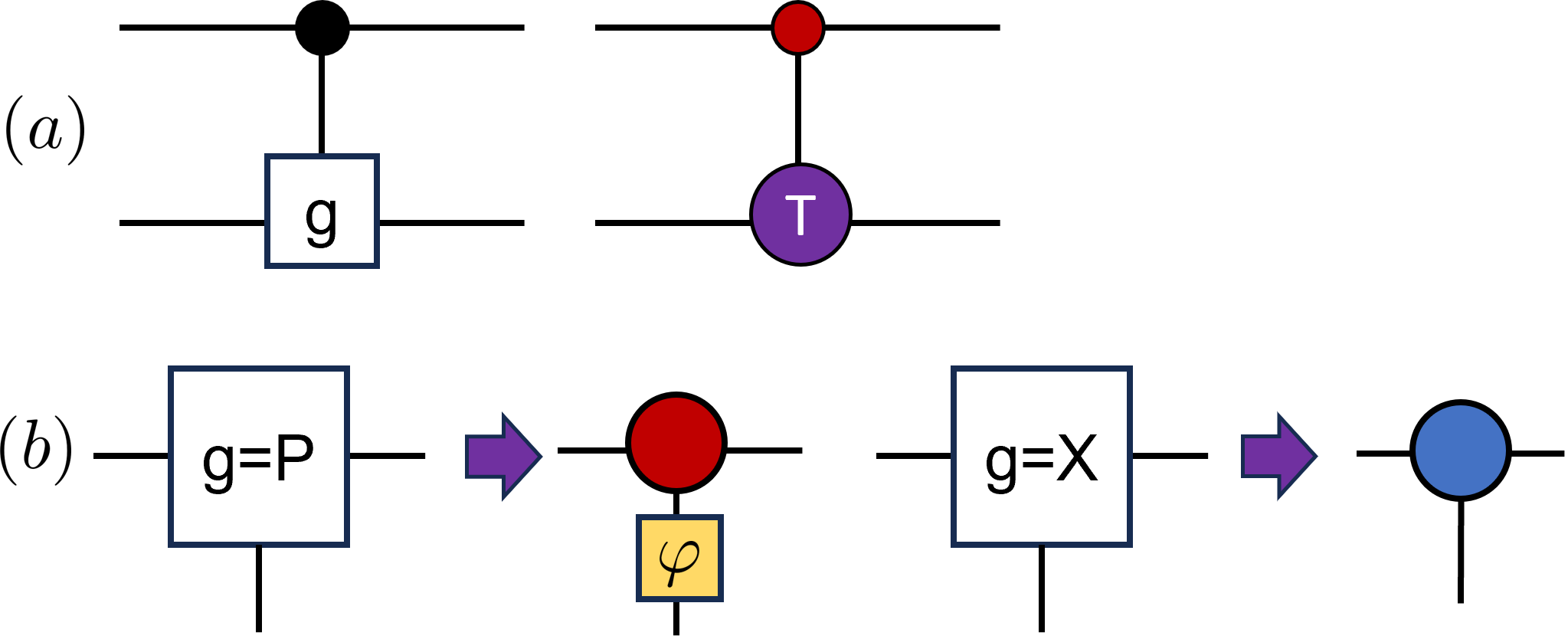}
    \caption{(a) Depending on the nature of the controlled-$g$ gate, the target action can be condensed into a tensor.(b) The action of a controlled-phase gate and that of a controlled X gate can be simplified into the corresponding tensors with $\varphi$ defined below. Red and blue tensors are Z and X spiders respectively.}
    \label{fig:gate_tensor}
\end{figure}

Here $\varphi$ is a tensor with elements
\begin{equation}
    \varphi_{ij} \propto \begin{pmatrix}
        1 & 1\\
        1 & e^{i\varphi}
    \end{pmatrix} 
\end{equation}
where an overall normalization is added as needed. For $\varphi=\pi$ it is the Hadamard gate/tensor.

Note that a code prepared this way has a non-trivial kernel in the encoding map. This is the same for the tensor network we built for the surface code or color code in the main text. 

For CSS codes, without loss of generality, one can perform first the Z checks then the X checks. By suitably subsituting $g$, one can simplify the data nodes (green) into the form of Fig.~\ref{fig:qtanner}c.
}

\bibliography{ref}
\bibliographystyle{unsrt}

\end{document}